\documentclass[11pt]{amsart}

\pdfoutput=1

\usepackage[utf8]{inputenc}
\usepackage{a4wide}

\makeatletter
\def\paragraph{\@startsection{paragraph}{4}%
  \z@\z@{-\fontdimen2\font}%
  {\normalfont\itshape}}
\makeatother

\usepackage{todonotes,wasysym,amssymb}
\usepackage[ocgcolorlinks,colorlinks=true,linkcolor=black,urlcolor=blue,citecolor=green!50!black]{hyperref}
\usepackage{oplotsymbl}
\usepackage{enumerate}
\usepackage[capitalize]{cleveref}

\newcommand{\eqdef}{\stackrel{\text{def}}=}
\newcommand{\pth}[1]{\left(#1\right)}
\newcommand{\R}{\ensuremath{\mathbb{R}}}
\newcommand{\N}{\ensuremath{\mathbb{N}}}

\newcommand{\ext}{\ensuremath{\mathrm{ext}}}
\newcommand{\nee}{\ensuremath{\mathrm{nee}}}
\newcommand{\Ex}{\ensuremath{\mathbb{E}}}
\newcommand{\cal}[1]{\ensuremath{\mathcal{#1}}}
\newcommand{\T}{\mathcal{T}}
\newcommand{\p}{\mathcal{P}}
\newcommand{\subs}[1]{\textrm{subs}}
\newcommand{\fusion}[1]{\ensuremath{\xrightarrow{#1}}}
\newcommand\eps{\varepsilon}
\newcommand\U{\mathcal{U}}
\newcommand\D{\mathcal{D}}
\newcommand\A{A}
\def\O{\mathcal O}
\usetikzlibrary{calc,positioning,backgrounds,shapes.geometric}
\DeclareMathOperator{\pap}{\textup{\Bowtie}}
\DeclareMathOperator{\sgn}{sgn}

\newcommand{\joinnormal}[2]{\ensuremath{\sideset{_{#1}}{_{#2}}\pap}}
\newcommand{\joinscript}[2]{\ensuremath{\sideset{\scriptstyle_{#1}}{_{#2}}{\scriptstyle\pap}}}
\newcommand{\join}[2]{
\mathchoice{\joinnormal{#1}{#2}}{\joinnormal{#1}{#2}}{\joinscript{#1}{#2}}{\joinscript{#1}{#2}}
}

\newcommand{\pt}[1]{\ensuremath{\mathfrak{#1}}}

\theoremstyle{plain}
\newtheorem{lemma}{Lemma}[section]
\newtheorem{theorem}[lemma]{Theorem}
\newtheorem{corollary}[lemma]{Corollary}
\newtheorem{proposition}[lemma]{Proposition}
\newtheorem{definition}[lemma]{Definition}

\theoremstyle{remark}
\newtheorem{remark}[lemma]{Remark}
\newtheorem{example}[lemma]{Example}
\newtheorem{observation}[lemma]{Observation}

\title[A canonical tree decomposition for order types]{A canonical tree decomposition for order types,\\ and some applications}
 \author[M. Bouvel]{Mathilde Bouvel}
       \address[MB,XG]{Université de Lorraine, CNRS, INRIA, LORIA, F-54000 Nancy, France}
       \email{mathilde.bouvel@loria.fr, xavier.goaoc@loria.fr}
 \author[V. Féray]{Valentin Féray}
       \address[VF]{Université de Lorraine, CNRS, IECL, F-54000 Nancy, France}
       \email{valentin.feray@univ-lorraine.fr}
 \author[X. Goaoc]{Xavier Goaoc}
 \author[F. Koechlin]{Florent Koechlin}
 \address[FK]{Université Sorbonne Paris Nord, LIPN, CNRS UMR 7030, F-93340 Villetaneuse, France}
 \email{koechlin@lipn.fr}

\begin{document}

\maketitle

\begin{abstract}
  We introduce and study a notion of decomposition of planar point
  sets (or rather of their chirotopes) as trees decorated by smaller
  chirotopes.  This decomposition is based on the concept of mutually
  avoiding sets (which we rephrase as \emph{modules}), and adapts in
  some sense the modular decomposition of graphs in the world of
  chirotopes.  The associated tree always exists and is unique up to
  some appropriate constraints.  We also show how to compute the
  number of triangulations of a chirotope efficiently, starting from
  its tree and the (weighted) numbers of triangulations of its parts.

  \medskip


\end{abstract}

\section{Introduction}

Any finite sequence $P=(p_1,p_2,\ldots, p_n)$ of points in $\R^2$
gives rise to a number of combinatorial data such as the subset of
extremal points, or the set of crossing-free matchings of $P$. Often,
such structures depend on the coordinates of the points of $P$ only
through a limited number of predicates, so one can replace the
(geometric) point sequence by a purely combinatorial object. For
instance, the convex hulls and crossing-free matchings are determined
by the map sending every triple $1 \le i,j,k \le n$ of indices to the
orientation (positive, negative or flat) of the triangle
$p_ip_jp_k$. For every $n$, there are finitely many such maps (called
{\em chirotope} or {\em order types}, see below) and they have been
extensively investigated, see {\em e.g.} the book of Bj{\"o}rner et
al.~\cite{bjorner1999oriented}.

\medskip

In this paper we examine chirotopes from the point of view of the {\em
  modular decomposition method}, which is an important tool in
algorithmic graph theory (see \emph{e.g.}~\cite{HabibPaul}).
To this end, we define a notion of module in chirotopes,
as a partition of the chirotope into two sets,
such that no line going through two points of the same set
separates points of the other set; see \Cref{f:mutuallyavoiding}.
This is reminiscent of the concept of {\em mutually avoiding sets},
previously considered in the literature,
though we use here a different perspective on such sets
(see discussion and references at the end of Section~\ref{ssec:context}).
Using this notion of modules, we encode a chirotope by a
tree whose nodes are themselves decorated by chirotopes, each point appearing in a
single node, so that subtrees joined by an edge encode modules of the original point set.\medskip
\begin{figure}
  \begin{center}  \includegraphics[page=1]{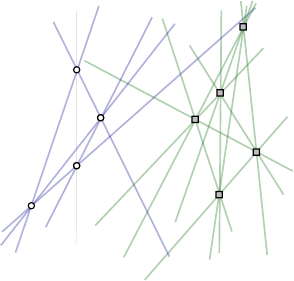} \end{center}
  \caption{A set of 9 points which can be decomposed into
  two modules of size $4$ and $5$ points.\label{f:mutuallyavoiding}}
\end{figure}

\medskip

The main contributions presented here are the following. 
\begin{itemize}
\item We introduce a notion of {\em canonical chirotope tree}
such that each chirotope is represented by a unique
canonical chirotope tree.
\item We show that the proportion of realizable chirotopes
that are indecomposable tends to 1.
\item We show how the number of triangulations of a
  chirotope given by a chirotope tree can be computed from the
  (weighted) numbers of triangulations of its nodes.
\end{itemize}
The first and third results also hold in the more general framework
of abstract chirotopes (function mapping triples of elements either to {\em clockwise}
or {\em counterclockwise}, satisfying some axioms, but not necessarily
corresponding to a set of points in $\R^2$; see details below).
We offer in \cref{ssec:perspectives} several new open questions raised
 by this concept of chirotope tree decompositions.

\bigskip

Throughout the paper, all point sets are finite, planar
and in general position, implying in particular that
  no three points are ever aligned.
    Adapting our construction to point sets in higher dimension
  would need considerable work, as our definitions really use the fact
  that the chirotope is a map of three arguments (in general, the chirotope
  of a point set in $\R^d$ takes $d+1$ arguments).
  On the other hand, we believe that most results could be adapted
  to point sets that are not necessarily in general position,
  but we decided not to discuss such cases for simplicity.

  Given a finite set $X$, we write $(X)_3$
for the set of triples $(x,y,z)$ of distinct elements in $X$. We write
$[n]$ for the set $\{1,2,\ldots, n\}$.

\subsection{Context and motivation}
\label{ssec:context}
Let us first provide some context on the objects, methods and
questions that we consider.

\subsubsection*{Realizable chirotopes.}

 The {\em chirotope of a
  set} ${\cal P} = \{\p_\ell\}_{\ell \in X}$ of points in general position
  labeled by $X$ is
the function
\[ \chi_{\cal P}: \left\{\begin{array}{rcl}
(X)_3 & \to & \{-1,+1\}\\
(x,y,z) & \mapsto & \left\{\begin{array}{rl}
+1& \text{ if $\p_x,\p_y,\p_z$ are in counterclockwise order,}\\
-1 &\text{ if $\p_x,\p_y,\p_z$ are in clockwise order.}\\
\end{array}\right.
\end{array}\right.\]
This function is sometimes called the {\em labeled order type} of the
point set~\cite{GP83}. We say that chirotopes of point sets are {\em realizable}
to distinguish them from their combinatorial (abstract)
generalization.

\subsubsection*{Representing realizable chirotopes.}
Deciding if a given function $(X)_3 \to \{-1,1\}$ is a realizable
chirotope turns out to be a challenging problem: it is equivalent to
the existential theory of the reals~\cite[Theorem
  8.7.2]{bjorner1999oriented} and
NP-hard~\cite{shor1991stretchability}. Whether this problem is in NP
is open, and, interestingly, a positive answer is equivalent to the
existence of a (pseudorandom) generator that produces a realizable
chirotope of size $n$ in time polynomial in $n$ while ensuring that
every element has nonzero probability\footnote{A random generator can
serve as a verifier for the problem, with the random bitstring as the
certificate. Conversely, we can run a random function and a random
certificate through the verifier; if accepted, we return that
function, otherwise we return a fixed chirotope, {\em e.g.} $n$ points
in convex position.}. The space of realizable chirotopes is therefore
difficult to explore, which makes it hard to test efficiently some
conjectures in discrete geometry
or to check experimentally a geometric algorithm in the exact
geometric computing paradigm~\cite{sharma2017robust}.  The present
work grew out of an attempt to devise new ways of constructing and
representing (large) realizable chirotopes.

\subsubsection*{Abstract chirotopes.}
As mentioned above, most of our results hold in a more general setting.
Let us define a {\em chirotope} on a finite set $X$ as a function $\chi: (X)_3 \to \{-1,1\}$ that
satisfies the following properties:\medskip
\begin{description}
\item[(symmetry)] for any distinct $x, y, z \in X$,
  \begin{equation}\label{eq:symmetry_condition}\chi(x,y,z)=\chi(y,z,x)=\chi(z,x,y)=-\chi(z,y,x)=-\chi(y,x,z)=-\chi(x,z,y);\end{equation}
\item[(interiority)] for any distinct $t, x, y, z \in X$,
  \begin{equation} \chi(t,y,z)=\chi(x,t,z)=\chi(x,y,t)=1 \quad \Rightarrow \quad \chi(x,y,z)=1;
  \end{equation}
\item[(transitivity)] for any distinct $s,t, x, y, z \in X$,
  \begin{equation} \chi(t,s,x)=\chi(t,s,y)=\chi(t,s,z) = \chi(x,y,t) = \chi(y,z,t)=1 \quad \Rightarrow \quad \chi(x,z,t) =1.\end{equation}
\end{description}
It can be proven that any realizable chirotope is a chirotope.
We may use the term {\em abstract chirotope}
to mean a chirotope that is not necessarily realizable.
Lastly, for readers acquainted with matroid theory, let us note that
abstract chirotopes are in correspondence with the relabeling classes of
acyclic uniform oriented matroids of
rank~3~\cite{bjorner1999oriented,knuth1992axioms}.
\medskip

Notions and properties for point sets that can be expressed only through
orientations of triples of points generalize to abstract chirotopes. For example, an
element $x \in X$ is {\em extreme} in a chirotope $\chi$ on $X$ if
there exists $y \in X\setminus\{x\}$ such that $\chi(x,y,z)$ is the
same for all $z \in X\setminus \{x,y\}$. With this definition,
extreme elements of realizable chirotopes correspond precisely
to the vertices of their convex hulls.

\subsubsection*{Triangulations.}

Counting the triangulations supported by a given set of $n$ points is
a classical problem in computational geometry (see the discussion
in~\cite{marx_et_al:LIPIcs.SoCG.2016.52}). The fastest known algorithm
is due to Marx and Miltzow~\cite{marx_et_al:LIPIcs.SoCG.2016.52} and
has complexity $\O(n^{(11+o(1))\sqrt{n}})$. It is quite involved, and a
simpler solution, due to Alvarez and Seidel~\cite{alvarez2013simple},
runs in time $\O(n^2 2^n)$.

\medskip

Given a point set $P=\{\p_\ell\}_{\ell \in X}$, two segments
$\p_x\p_y$ and $\p_z\p_t$ with distinct endpoints in $P$ cross if and
only if $\chi_P(x,y,z) = - \chi_P(x,y,t)$ and $\chi_P(z,t,x) =
-\chi_P(z,t,y)$. We can therefore define the crossing of segments for
abstract chirotopes. A {\em segment} in a chirotope $\chi$ on $X$ is a pair of
elements of $X$, and the segments $xy$ and $zt$ {\em cross in $\chi$}
if they satisfy the above condition. A {\em triangulation} of $\chi$
is an inclusion-maximal family of segments such that no two cross in
$\chi$. The algorithm of Alvarez and Seidel easily generalizes to
abstract chirotopes.

\subsubsection*{Modular decomposition.}

The gist of modular decomposition is to break down discrete structures
to allow efficient recursion. These decompositions usually start by
partitioning the elements of the structure so that any two elements in
a part (also called a {\em module}) are indistinguishable ``from the
outside''. Choosing this partition as coarse as possible while being
nontrivial, and iterating the decomposition within each part yields a
decomposition tree. These ideas originated in graph
theory~\cite{Gallai}, where modules gather vertices with the same
neighbors outside of the module, and several hereditary classes of
graphs have well-behaved modular decompositions, {\em e.g.}
comparability graphs, permutation graphs and
cographs~\cite{HabibPaul}.  We also note that a variant of modular
decomposition, called split decomposition, has been thoroughly studied
in graph theory; see, {\em e.g.}, \cite{rao2008split}.  Interestingly,
our chirotope tree decomposition share some features with both the
modular and split decompositions: our notion of modules is closer to
that of modules in graph than to that of split, but our decomposition
tree are unrooted and unordered, as split decomposition trees.

Other examples of structures for which modular decompositions were
developed include boolean functions, set systems and
permutations~\cite{MoRa84,AA05}. Each structure requires an ad hoc
analysis that the proposed notion of modules leads to a well-defined
decomposition, where each object has a unique decomposition tree. The
proportion of objects with nontrivial decomposition is often
vanishingly small. Nevertheless, such decompositions proved useful
{\em e.g.} in devising fixed-parameter polynomial algorithms for hard
algorithmic problems~\cite[$\mathsection 7$]{HabibPaul} or in studying
subclasses of objects~\cite{pinperm, CFL17}.

\subsubsection*{Mutual avoiding point sets.}

Let us discuss the link with the notion of mutually avoiding subsets,
 as considered in the litterature.
Two planar point sets are {\em mutually avoiding} if no line through
two points from one set separates two points from the other set.
 Any set of $n$
points in general position contains two mutually avoiding subsets of
size $\Omega(\sqrt{n})$ each~\cite{Aronov1994}, a bound that is
asymptotically best possible~\cite{valtr1997mutually}, and smaller
mutually avoiding subsets are actually abundant~\cite[Theorem
  1.3]{suk2023positive}. Mutually avoiding subsets have been applied
to the study of crossing families~\cite{Aronov1994} and empty
$k$-gons~\cite{valtr1997mutually}.

The {\em modules} we consider here are mutually avoiding subsets that
cover the whole chirotope. Our work goes in a different direction than
the above cited works in that they are looking for mutually avoiding
subsets {{\em contained} in a chirotope, whereas we look for
  mutually avoiding subsets that {\em decompose} a chirotope. Our
  analysis thus applies to a more restricted class of chirotopes, the
  (highly) decomposable ones, but gives more information on their
  structure.}

\subsection{Our results}\label{s:overview}

We now define our decomposition and state our main results.

\subsubsection*{Bowtie decomposition.}

The fact that a point set can be split into two modules $P$ and $Q$
 has two interesting consequences at the level of chirotopes.
 First, the
chirotope on point set $P \cup \{p^*\}$, which we denote $\chi_{P \cup \{p^*\}}$, is, up to relabeling, independent of
the choice of the point $p^* \in Q$. Second, the chirotope 
$\chi_{P \cup Q}$ on point set $P\cup Q$ is completely determined by the chirotopes 
$\chi_{P \cup \{p^*\}}$ and $\chi_{Q \cup \{q^*\}}$ for any choices of $p^* \in Q$
and $q^* \in P$. We can use this to decompose $\chi_{P \cup Q}$ in
terms of $\chi_{P \cup \{p^*\}}$ and $\chi_{Q \cup \{q^*\}}$.

\medskip

Let us express this decomposition. Throughout the paper,
we call {\em sign function} on a set $X$,
a function $(X)_3 \to \{-1,+1\}$ that is only required to satisfy the symmetry
condition~\eqref{eq:symmetry_condition}. Let $X$ and $Y$ be
disjoint sets, and let $x^* \notin X$ and $y^* \notin Y$. Given two
sign functions $\chi$ on $X\cup \{x^*\}$  and $\xi$ on $Y\cup
\{y^*\}$, we define the {\em bowtie} $\kappa \eqdef \chi
\join{x^*}{y^*} \xi$ as the sign function on $X \cup Y$ satisfying:
\begin{equation}\label{eq:def_bowtie}
  \left\{
  \begin{array}{rcll}
\kappa(x_1,x_2,x_3) &=& \chi(x_1,x_2,x_3) &\text{ if $x_1,x_2,x_3$ are all in $X$;}\\
\kappa(x_1,x_2,y) &=& \chi(x_1,x_2,x^*) &\text{ if $x_1,x_2$ are in $X $ and $y$ is in $Y$;}\\
\kappa(x,y_2,y_3) &=& \xi(y^*,y_2,y_3) &\text{ if $x$ is in $X$ and $y_2,y_3$ are in $Y$;}\\
\kappa(y_1,y_2,y_3) &=& \xi(y_1,y_2,y_3) &\text{ if $y_1,y_2,y_3$ are all in $Y$.}
  \end{array}\right.
\end{equation}
With the symmetry condition~\eqref{eq:symmetry_condition},
this completely defines $\kappa$ on $(X\cup Y)_3$.
Informally, this amounts to computing the sign function as in $\chi$ (resp. $\xi$)
 if the majority of the elements come from $X$ (resp. $Y$),
possibly replacing the only element outside $X$ (resp. $Y$) by $x^*$ (resp. $y^*$).

 Note that for disjoint modules $P$ and $Q$ of $P \cup Q$,
 we have  $\chi_{P \cup Q} = \chi_{P \cup \{p^*\}} \join{p^*}{q^*} \chi_{Q \cup \{q^*\}}$
  for any choice of $p^* \in Q$ and $q^* \in P$.

\subsubsection*{Bowtie products.}

While defined on sign functions, the bowtie operator also allows to
combine smaller chirotopes into larger ones. This works under the
following conditions.

\begin{proposition}\label{p:whenischirotope}
  Let $X$ and $Y$ be disjoint sets, with $|X|,|Y| \ge 2$, and let
  $x^* \notin X$ and $y^* \notin Y$. Let $\chi$ and $\xi$ be
  chirotopes on $X\cup \{x^*\}$ and $Y\cup \{y^*\}$.\smallskip
  \begin{enumerate}[(i)]
  \item $\chi \join{x^*}{y^*} \xi$ is a chirotope if and only if
    $x^*$ and $y^*$ are extreme in $\chi$ and~$\xi$.\smallskip

  \item $\chi \join{x^*}{y^*} \xi$ is a realizable chirotope if
    and only if $\chi$ and $\xi$ are realizable and $x^*$ and $y^*$ are
    extreme in $\chi$ and~$\xi$.\smallskip
  \item If $\chi \join{x^*}{y^*} \xi$ is a chirotope, then the extreme elements of $\chi \join{x^*}{y^*} \xi$ are
    the elements of $X \cup Y$ extreme in $\chi$ or in $\xi$.
\end{enumerate}
\end{proposition}

\noindent
The above proposition is proved in \cref{subsec:subs}. We also show
(\cref{prop:commutativity_asso_bowtie}) that $\Bowtie$ has
associativity and commutativity properties.  We say that a chirotope
$\kappa$ is {\em decomposable} if there exist chirotopes $\chi$ and
$\xi$, each on a strictly smaller set, such that $\kappa = \chi
\join{x^*}{y^*} \xi$. If no such $\chi$ and $\xi$ exist we say that
$\kappa$ is {\em indecomposable}.

\subsubsection*{Chirotope trees.}
Formally, a {\em chirotope tree} is a tree whose nodes are decorated
with chirotopes on disjoint ground sets of size at least $3$, and whose edges select an
element in the ground set of each of their extremities so that (i) no
element is selected more than once, and (ii) each selected element is
extreme in its chirotope. Any element selected by an edge of the
chirotope tree is called a {\em proxy}; for illustration purposes, we
write the decorating chirotope as the label of the node (if it is
realizable, we may draw a realization) and we draw edges of the tree
and proxy elements in red, with the edge connecting the proxy
elements. Here are two examples.

  \medskip
  \begin{center}
    \begin{tikzpicture}[baseline={(0,0)},scale=0.9, every node/.style={transform shape}]
      \node[circle,draw,minimum size=1cm,very thin] (n2) at (0,0) {$\kappa$};
      \node[circle,draw,minimum size=1cm,very thin] (n2bis) at (1.8,0) {$\xi$};
      \node[circle,draw,very thin,minimum size=1cm] (n3) at ($(n2)+(150:1.8)$) {$\chi_3$};
      \node[circle,draw,minimum size=1cm,very thin] (n1) at ($(n2)+(-150:1.8)$) {$\chi_1$};
      \node[circle,draw,minimum size=1cm,very thin] (n4) at ($(n2bis)+(0:1.8)$) {$\chi_2$};

      \node[red] (x1*) at ($(n1)+(30:.35)$) {$\bullet$};
      \node[red] (x3*) at ($(n3)+(-30:0.35)$) {$\bullet$};
      \node[red] (x4*) at ($(n4)+(180:0.35)$) {$\bullet$};
      \node[red] (y2*) at ($(n2)+(180-30:0.35)$) {$\bullet$};
      \node[red] (z2*) at ($(n2bis)+(0:0.35)$) {$\bullet$};
      \node[red] (x2*) at ($(n2)+(180+30:0.35)$) {$\bullet$};
      \node[red] (s*) at ($(n2)+(0:0.35)$) {$\bullet$};
      \node[red] (t*) at ($(n2bis)+(180:0.35)$) {$\bullet$};

      \node[red,above=-0.23 of x1*,font=\footnotesize] {$x_1^*$};
      \node[red,below=-0.23 of x2*,font=\footnotesize] {$x^*$};
      \node[red,above=-0.23 of y2*,font=\footnotesize] {$y^*$};
      \node[red,above=-0.23 of z2*,font=\footnotesize] {$z^*$};
      \node[red,below=-0.23 of x3*,font=\footnotesize] {$x_3^*$};
      \node[red,above=-0.23 of x4*,font=\footnotesize] {$x_2^*$};
      \node[red,above right=-0.3 of s*,font=\footnotesize] {$s^*$};
      \node[red,above left=-0.4 of t*,font=\footnotesize] {$t^*$};

      \draw[red] (x3*.center)--(y2*.center);
      \draw[red] (x2*.center)--(x1*.center);
      \draw[red] (x4*.center)--(z2*.center);
      \draw[red] (s*.center)--(t*.center);
  \end{tikzpicture}
      \hspace{1.2cm}
        \begin{tikzpicture}[baseline={(-.1,0)},scale=0.8, every node/.style={transform shape}]
      \node (n1) at (0,0) {};
      \node (n2) at ($(0,0)+(-90:1.5)$) {};
      \node (n3) at ($(0,0)+(150:1.5)$) {};
      \node (n4) at ($(0,0)+(30:1.5)$) {};
      \draw[very thin] (n1) circle (.65);
      \draw[very thin] (n2) circle (.65);
      \draw[very thin] (n3) circle (.65);
      \draw[very thin] (n4) circle (.65);
      \node[red] (r) at ($(n1)+(150:.35)$) {$\bullet$};
      \node[red] (t) at ($(n1)+(-90:.35)$) {$\bullet$};
      \node[red] (s) at ($(n1)+(30:.35)$) {$\bullet$};
      \node[red, above right=-0.4 of r] {$r^*$};
      \node[red, below right=-0.5 of s] {$s^*$};
      \node[red, above=-0.2 of t] {$t^*$};

      \node (f) at ($(n2)+(-150:.35)$) {$\bullet$};
      \node[red] (z) at ($(n2)+(90:.35)$) {$\bullet$};
      \node (e) at ($(n2)+(-30:.35)$) {$\bullet$};
      \node (g) at ($(n2)+(-90:.1)$) {$\bullet$};
      \node[above left=-0.5 of f] {$f$};
      \node[red, above right=-0.4 of z] {$z^*$};
      \node[above right=-0.4 of e] {$e$};
      \node[below=-0.2 of g] {$g$};

      \node (a) at ($(n3)+(-150:.35)$) {$\bullet$};
      \node (b) at ($(n3)+(90:.35)$) {$\bullet$};
      \node[red] (x) at ($(n3)+(-30:.35)$) {$\bullet$};
      \node[above=-0.2 of a] {$a$};
      \node[above right=-0.4 of b] {$b$};
      \node[red, above=-0.2 of x] {$x^*$};

      \node[red] (y) at ($(n4)+(-150:.35)$) {$\bullet$};
      \node (b) at ($(n4)+(90:.35)$) {$\bullet$};
      \node (c) at ($(n4)+(-30:.35)$) {$\bullet$};
      \node[red,above=-0.2 of y] {$y^*$};
      \node[above right=-0.4 of b] {$c$};
      \node[above=-0.2 of c] {$d$};

      \draw[red] (x.center) -- (r.center);
      \draw[red] (s.center) -- (y.center);
      \draw[red] (t.center) -- (z.center);
    \end{tikzpicture}
  \end{center}

\subsubsection*{Chirotope of a chirotope tree.}
To any chirotope tree $T=(V_T,E_T)$ we associate a sign function
$\chi_T$ as follows. For a node $v$ of $T$, let $\chi_v$ denote the
chirotope decorating it, and let $X_v$ denote the set of nonproxy
elements of the ground set of $\chi_v$. For $x \in \cup_{w \in V_T}
X_w$, the \emph{representative} $R_T(x,v)$ of $x$ in the node $v$ is
$x$ if $x \in X_v$, and otherwise it is the label of the proxy
selected in $\chi_v$ by the first edge on the path in $T$ from $v$ to
the node containing $x$.  For distinct $x,y,z \in \cup_{w \in V_T}
X_w$, we let $v(x,y,z)$ be the intersection node of the paths in $T$
from $x$ to $y$,\footnote{Rather: from the node containing $x$ to the
node containing $y$, and similarly for $(y,z)$ and $(x,z)$.}  from $y$
to $z$ and from $x$ to $z$ (if two or three elements among $x$, $y$,
$z$ are in the same $X_w$, then we set $v(x,y,z)=w$). We define a sign
function $\chi_T$ on $\cup_{w \in V_T} X_w$ by\footnote{$R(x,v)$
replaces, {\em i.e.} serves as a proxy for, $x$ for computing the sign
$\chi_T(x,y,z)$ in the node $v$.}
\begin{equation}
  \chi_T(x,y,z) \eqdef \chi_v(R(x,v),R(y,v),R(z,v)) \qquad \text{for } v = v(x,y,z).
\end{equation}

\noindent
For example, if $T$ is the chirotope tree on the right in the previous
picture we have $\chi_T(a,c,f)=\kappa(r^*,s^*,t^*)=-1$, where $\kappa$
is the chirotope decorating the central node. The definition of
$\chi_T$ ensures the following property: if removing an edge in a tree
$T$ produces two subtrees $T_1$ and $T_2$, then $\chi_T$ is the bowtie
product of $\chi_{T_1}$ and $\chi_{T_2}$
(see Lemma~\ref{lem:chiT_is_join_2}).  Furthermore, this notion behaves
well with those of chirotopes and of realizability.
\begin{proposition}
  For any chirotope tree $T$, $\chi_T$ is a chirotope. Moreover, if
  all chirotopes decorating a node in $T$ are realizable, then
  $\chi_T$ is realizable.
\label{prop:chirotope_tree}
\end{proposition}

By Proposition~\ref{p:whenischirotope}(iii), the extreme elements of
$\chi_T$ are the nonproxy extreme elements of its decorating
chirotopes.  Hence, if $T$ has $k$ nodes, then $\chi_T$ has at least
$k+2$ extreme elements. Indeed, as each node has at least $3$ extreme
elements and $T$ has exactly $k-1$ edges, therefore $2(k-1)$ proxies,
there are at least $3k-2(k-1)=k+2$ extreme elements in $\chi_T$.

\subsubsection*{Canonical chirotope trees.}
Let us consider a chirotope tree $T$ and a node $v$ of $T$. If
$\chi_v=\kappa \join{s^*}{t^*} \xi$ is decomposable, then $v$ can be
replaced by two nodes decorated with $\kappa$ and $\xi$, connected by
an edge, so that the resulting chirotope tree $T'$ satisfies
$\chi_{T'}=\chi_T$ (see Section~\ref{ssec:operation-trees} for
details). Hence, a chirotope $\chi$ corresponds to many chirotope
trees, at least one of which is decorated only by indecomposable
chirotopes. However, even requesting the chirotopes decorating the
nodes to be indecomposable does not ensure that a unique chirotope
tree represents a given chirotope. Here are two trees with only
indecomposable decorations and the same associated chirotope on
$\{a,b,c,d,e,f,g\}$.

  \begin{center}
    \begin{tikzpicture}[baseline={(-.1,0)},scale=0.8, every node/.style={transform shape}]
      \node (n1) at (0,0) {};
      \node (n2) at ($(0,0)+(-90:1.5)$) {};
      \node (n3) at ($(0,0)+(150:1.5)$) {};
      \node (n4) at ($(0,0)+(30:1.5)$) {};
      \draw[very thin] (n1) circle (.65);
      \draw[very thin] (n2) circle (.65);
      \draw[very thin] (n3) circle (.65);
      \draw[very thin] (n4) circle (.65);
      \node[red] (r) at ($(n1)+(150:.35)$) {$\bullet$};
      \node[red] (t) at ($(n1)+(-90:.35)$) {$\bullet$};
      \node[red] (s) at ($(n1)+(30:.35)$) {$\bullet$};
      \node[red, above right=-0.4 of r] {$r^*$};
      \node[red, below right=-0.5 of s] {$s^*$};
      \node[red, above=-0.2 of t] {$t^*$};

      \node (f) at ($(n2)+(-150:.35)$) {$\bullet$};
      \node[red] (z) at ($(n2)+(90:.35)$) {$\bullet$};
      \node (e) at ($(n2)+(-30:.35)$) {$\bullet$};
      \node (g) at ($(n2)+(-90:.1)$) {$\bullet$};
      \node[above left=-0.5 of f] {$f$};
      \node[red, above right=-0.4 of z] {$z^*$};
      \node[above right=-0.4 of e] {$e$};
      \node[below=-0.2 of g] {$g$};

      \node (a) at ($(n3)+(-150:.35)$) {$\bullet$};
      \node (b) at ($(n3)+(90:.35)$) {$\bullet$};
      \node[red] (x) at ($(n3)+(-30:.35)$) {$\bullet$};
      \node[above=-0.2 of a] {$a$};
      \node[above right=-0.4 of b] {$b$};
      \node[red, above=-0.2 of x] {$x^*$};

      \node[red] (y) at ($(n4)+(-150:.35)$) {$\bullet$};
      \node (b) at ($(n4)+(90:.35)$) {$\bullet$};
      \node (c) at ($(n4)+(-30:.35)$) {$\bullet$};
      \node[red,above=-0.2 of y] {$y^*$};
      \node[above right=-0.4 of b] {$c$};
      \node[above=-0.2 of c] {$d$};

      \draw[red] (x.center) -- (r.center);
      \draw[red] (s.center) -- (y.center);
      \draw[red] (t.center) -- (z.center);
    \end{tikzpicture}
    \hspace{2cm}
    \begin{tikzpicture}[baseline={(-.1,0)},scale=0.8, every node/.style={transform shape}]
      \node (n1) at (0,0) {};
      \node (n2) at ($(0,0)+(3,0)$) {};
      \node (n3) at ($(0,0)+(-1.5,0)$) {};
      \node (n4) at ($(0,0)+(1.5,0)$) {};
      \draw[very thin] (n1) circle (.65);
      \draw[very thin] (n2) circle (.65);
      \draw[very thin] (n3) circle (.65);
      \draw[very thin] (n4) circle (.65);
      \node[red] (r) at ($(n1)+(180:.35)$) {$\bullet$};
      \node (c) at ($(n1)+(90:.35)$) {$\bullet$};
      \node[red] (s) at ($(n1)+(0:.35)$) {$\bullet$};
      \node[red, below=-0.23 of r] {$r^*$};
      \node[red, below=-0.23 of s] {$s^*$};
      \node[right=-0.23 of c] {$c$};

      \node (f) at ($(n2)+(-60:.35)$) {$\bullet$};
      \node[red] (z) at ($(n2)+(180:.35)$) {$\bullet$};
      \node (e) at ($(n2)+(60:.35)$) {$\bullet$};
      \node (g) at ($(n2)+(0:.1)$) {$\bullet$};
      \node[below left=-0.4 of f] {$f$};
      \node[red, above=-0.2 of z] {$z^*$};
      \node[above left=-0.4 of e] {$e$};
      \node[right=-0.27 of g] {$g$};

      \node (a) at ($(n3)+(-120:.35)$) {$\bullet$};
      \node (b) at ($(n3)+(120:.35)$) {$\bullet$};
      \node[red] (x) at ($(n3)+(0:.35)$) {$\bullet$};
      \node[above=-0.2 of a] {$a$};
      \node[above right=-0.4 of b] {$b$};
      \node[red, above=-0.2 of x] {$x^*$};

      \node[red] (t) at ($(n4)+(180:.35)$) {$\bullet$};
      \node (d) at ($(n4)+(90:.35)$) {$\bullet$};
      \node[red] (y) at ($(n4)+(0:.35)$) {$\bullet$};
      \node[red, below=-0.23 of t] {$t^*$};
      \node[red, below=-0.23 of y] {$y^*$};
      \node[right=-0.23 of d] {$d$};

      \draw[red] (x.center) -- (r.center);
      \draw[red] (t.center) -- (s.center);
      \draw[red] (y.center) -- (z.center);
    \end{tikzpicture}
  \end{center}

\noindent
It turns out that the  {\em convex} chirotopes
(a chirotope is convex if all its elements are extreme)
 are the only source of redundancy;
 in the above example, the chirotope is not convex,
but its restriction to  $\{a,b,c,d,e\}$ is.
 This leads us to define a chirotope tree as
\emph{canonical} if every node is decorated by a convex or
indecomposable chirotope, and if no edge connects two nodes decorated
by convex chirotopes.

\begin{theorem}
  \label{thm:uniqueness_canonical}
  For any chirotope $\chi$, there is a unique canonical
  chirotope tree $T$ with $\chi_T = \chi$.
\end{theorem}
\noindent We prove this theorem in \cref{ssec:uniqueness_canonical}.

\subsubsection*{Counting triangulations.}

Let $\mathcal{T}_{\kappa}$ denote the set of triangulations of a
chirotope $\kappa$ on $X$. Given $x^* \in X$, we let
$P_{\kappa,x^*}(s)=\sum_{T\in \mathcal{T}_{\kappa}}s^{\deg_{T}(x^*)}$
denote the generating polynomial of the triangulations of $\kappa$,
marking the degree of $x^*$. We prove (see
\cref{prop:triangulation_bowtie}) that if $\kappa =
\chi\join{x^*}{y^*}\xi$, then
\begin{equation}\label{eq:counting bowtie}
 |{\cal T}_\kappa| = \sum_{\substack{T'\in \mathcal{T}_{\chi}\\T''\in \mathcal{T}_{\xi}}}\binom{\deg_{T'}(x^*)+\deg_{T''}(y^*)-2}{\deg_{T'}(x^*)-1}    = \sum_{a,b \ge 2} \tbinom{a+b-2}{a-1} [s^a]P_{\chi,x^*}(s) \  [t^b]P_{\xi,y^*}(t).
\end{equation}
More generally, given a chirotope tree~$T$, we can compute the number
of triangulations of $\chi_T$ given, for each node $v$ of $T$, a
collection of generating polynomials $Q^U_v$ for the triangulations of
$\chi_v$. These polynomials are indexed by the possible restriction
$U$ of the triangulation to the set $Y_v$ of proxy elements of
$\chi_v$, and $Q^U_v$ enumerate the triangulations with restriction
$U$, keeping track of the degrees of all proxy elements (it is a
multivariate polynomial with one variable for each proxy element).
This general principle can be used on the one hand to find explicit
formulas for the number of triangulations of some families of point
sets, and on the other hand to design an effective algorithm to
compute the number of triangulations of a chirotope, {\em given its
  decomposition tree}.

In \cref{sec:kernel-method}, we consider a family $\mathcal{F}$ of
chirotope trees, whose underlying trees are $n$-paths, with all nodes
decorated by one of two chirotopes on four elements (see
\cref{sec:kernel-method}). For the chirotopes of such chirotope paths,
we obtain a recurrence formula for the generating polynomials of
triangulations.  These recurrences can be solved via the {\em kernel
  method} from analytic combinatorics, leading to one and the same
closed formula for the number of triangulations of any chirotope
represented by a path in $\mathcal{F}$
(Proposition~\ref{p:counting-chains}).

For the algorithmic aspect, we recall that, if, for a node $v$, $\chi_v$ has $n_v$
elements, $k$ of which are proxies, we have to compute a family $Q^U_v$ of $k$-variate polynomials.
We show that the number of polynomials in the family is bounded by $\O(A^k)$, with
$A=6+4 \sqrt 2 \approx 11.66$; see \cref{lem:control_U}.
In addition, these polynomials can be computed in time $O\pth{A^k 2^{n_v}
  n_v^{k+2}}$ by a simple modification of the Alvarez-Seidel
algorithm~\cite{alvarez2013simple}.
The following result gives a
precise bound on the complexity of counting triangulations, given the
polynomials $Q^U_v$, for all nodes $v$ of $T$ and all noncrossing configurations $U$ of edges
on the proxy elements of $v$.
\begin{proposition}
  \label{p:complexity_triangulations}
  Let $T$ be a chirotope tree with $m$ edges, in which each node has
  degree at most $k$, and such that $\chi_T$ has size $n$. The number
  of triangulations of $\chi_T$ can be computed from the collections of polynomials
  $\{Q^U_v\}$ in time $O\pth{\A^{k} n^{2k}
    m}$ in the Real-RAM model.\footnote{In particular, this model
  allows arithmetic computation on arbitrary reals in time $\O(1)$. We
  refer to Erickson et al.~\cite{erickson2022smoothing} for a formal
  definition of the Real-RAM model of computation.}
\end{proposition}

We implemented our method for a family of trees of degree at most~$3$
(see \cref{a:arity3} for the details and access to the code). This
implementation is merely a proof of concept, and it is beyond the
scope of this paper to optimize it or benchmark it, but we can mention
that on a laptop, it took a few seconds to count the triangulations of
the example of Figure~\ref{f:exemple} page~\pageref{f:exemple}, that
assembles a set of $254$ points from decorating chirotopes of size
$9$, and that it handles random ternary chirotope trees with 146 nodes
(1024 elements) in less than 5 minutes.

\subsubsection*{Proportion of decomposable chirotopes among realizable ones.}

Finally, let $t_n$ denote the
number of {\em realizable} chirotopes of size $n$ and $d_n$ denote the number
of those that are decomposable.
Our last result gives a precise estimates of $d_n$, in terms of $t_n$ and $t_{n-1}$.
In particular, as announced earlier, almost all realizable chirotopes
are indecomposable (as for graphs with respect
to the split or modular decompositions).
\begin{theorem}
\label{thm:proportion_indecomposable}
For large $n$, we have $3 \pth{n-\O(n^{-2})} t_{n-1} \le d_n \le \O(n^{-3}) t_n$.
\end{theorem}

\noindent
Since $t_n$ is of order $n^{4n}$, we expect $t_{n-1}/t_n$ to behave as
$\Theta(n^{-4})$ (but only an upper bound of this order is known).  If
this is indeed the case, the lower bound in our theorem behaves as
$\Theta(n^{-3})t_{n}$, i.e.~our bounds are tight up to multiplicative
constants. Note that we do not have an analogue result for abstract chirotopes.

\subsection{Discussion and related work}

\subsubsection*{Originality with respect to other modular decomposition theories.}

The theory of modular decomposition~\cite{MoRa84} does not offer any
unifying model that specializes to various structures, but rather
proposes general guidelines for developing such a theory. The
implementation of these guidelines usually requires analyses that are
specific to the objects at hand, and the case of chirotopes is no
exception: almost every step requires some specific geometric
idea(s).

\subsubsection*{Among realizable chirotopes, the proportion of indecomposable ones tends to $\mathbf{1}$.}

This is analogue to the well-known fact that most graphs are {\em
  prime graphs} for the modular decomposition. The proof is, however,
much more difficult and interesting than in the graph case because (i)
the number of realizable chirotopes grows slower than the number of
graphs on $n$ vertices and (ii) this number is only known up to
exponential corrections. In particular, our proof uses some geometric
constructions and is specific to the realizable setting.

\medskip

Let us also mention that even if most graphs are prime graphs
for the modular decomposition, this decomposition has turned out useful in many ways:
design of efficient algorithms \cite{HabibPaul}, enumeration and study of specific classes~\cite{bassino2022cographs},~\dots\
Similarly, for chirotopes, our tree decomposition allows one for example to construct
 families of large chirotopes for which enumerating triangulations
 is algorithmically easy (see \cref{sec:triangulations}).

\subsubsection*{Recursive constructions of point sets.}

Recursive constructions of point sets abund in discrete geometry, and
with some care several of them directly translate into recursively
defined canonical chirotope trees ({\em e.g.} Horton sets). The
closest predecessor to (and inspiration for) this work is the
recursive decomposition of {\em chains} used by Rutschmann and
Wettstein~\cite[Theorem~15]{rutschmann2023chains} to produce a new
lower bound on the maximal number of triangulations of a set of $n$
points in the plane.  Their representation applies to a smaller subset
of order types (the {\em chains}, of which there are only
exponentially many of size $n$), and decomposes every chain into parts
of constant size. It allows to count the triangulations of a chain of
size $n$ in $\O(n^2)$ time based on ideas similar to the proof of
\cref{eq:counting bowtie}, but tailored to the setting of
chains. Another notion of {\em decomposability} of point sets was
proposed by Balko et al.~\cite{balko2017induced} to study Ramsey-type
questions; that notion combines conditions on orientations and
comparisons between $x$-coordinates, which more restrictive than our
setting, and these authors focus on point sets that are fully
decomposable, in the sense that every basic building block consists of
a single point.

\subsubsection*{Counting triangulations.}

It is interesting to compare our method for counting triangulations to the
experimental results for general point sets of Alvarez and
Seidel~\cite{alvarez2013simple}, and for chains by Rutschmann and
Wettstein~\cite{rutschmann2023chains}.
An extensive comparison is beyond the scope of this
paper, all the more that the methods apply to different classes of
point sets and operate on different types of input. We nevertheless
note the following points.
\begin{itemize}
\item The general method of Alvarez and Seidel takes as input an
  arbitrary point set. They tested three methods on high-memory
  hardware for various types of point sets, and none could handle any
  example of size $35$ or more in less than 10 minutes.

\item Our method takes as input a (chirotope presented by a)
  {chirotope} tree. We tested it on randomly-generated trees with
  decorating chirotopes of size $9$ summing up to sets of $\sim 1000$
  points; it took less than 5 minutes on a laptop to count the
  triangulations.

\item The method of Rutschmann and Wettstein is based on formulas
  specific to the class of chains and could handle examples of size
  $2^{21}$.
\end{itemize}
To us, this is an evidence that our approach is relevant for point sets that
are highly decomposable. This is similar to some of the usual benefits of
modular decomposition: some hard problems enjoy simple and effective
solutions for instances that decompose well.

\subsection{Perspectives}\label{ssec:perspectives}
Let us conclude this introduction by sampling some directions of enquiry that this work opens.\medskip

\begin{enumerate}
\item Here we used canonical chirotope trees as a tool to
  put together large chirotopes out of smaller ones, and we therefore considered the chirotope
  tree as given. One could, however, start with some point set and set
  out to compute the decomposition tree. This raises the following
  computational questions.  How efficiently can one find a partition
  of a given (realizable) chirotope into two mutually avoiding parts
  or decide that none exists? Or, going even further, how efficiently can one compute the
  canonical chirotope tree of a given (realizable) chirotope?\medskip

\item Can crossing-free structures besides triangulations be counted
  efficiently from the chirotope tree? What about other statistics of
  a chirotope such as the number of $k$-sets or the number of crossing
  pairs?\medskip

\item What is the proportion of indecomposable (abstract) chirotopes? (Our proof
  of Theorem~\ref{thm:proportion_indecomposable} uses that
  $t_n/t_{n-1} \ge \Omega(n^4)$ and we are not aware of an analogue of this in the
  abstract setting.)\medskip

\end{enumerate}

\subsection{Outline of the paper}

The paper is organized as follows. 
\begin{itemize}
\item Sections~\ref{s:prel} and~\ref{s:bam} introduce the notions of
  modules and bowties, and in particular proves
  \cref{p:whenischirotope}.
\item Section~\ref{sec:canonical_tree} introduces chirotope trees and
  canonical chirotope trees, and in particular proves
  \cref{prop:chirotope_tree} and \cref{thm:uniqueness_canonical}.
\item Sections~\ref{s:almost} and~\ref{s:quasi} estimate the number of
  decomposable chirotopes, and in particular prove
  \cref{thm:proportion_indecomposable}.
\item Sections~\ref{sec:triangulations},
  \ref{sec:counting_triangulations}
  and~\ref{sec:general_triangulations} show how to count the
  triangulations of a chirotope presented by a chirotope tree, and in
  particular prove \cref{p:complexity_triangulations}.
\end{itemize}

\subsection*{Acknowledgments}

The authors thank Emo Welzl for discussion leading to footnote~1. This work was done while FK was a postdoctoral fellow in LORIA, funded by ANR ASPAG (ANR-17-CE40-0017), IUF and INRIA.

\section{Preliminary: some properties of extreme points}
\label{s:prel}

{Recall that an element $x \in X$ is {\em extreme} in a chirotope
  $\chi$ on $X$ if there exists $y \in X\setminus\{x\}$ such that
  $\chi(x,y,z)$ is the same for all $z \in X\setminus \{x,y\}$.}  We
start with a standard extension to abstract chirotopes of
Carathéodory's theorem: any point not extreme in a set $X$ is
contained in a triangle with vertices in $X$
(see~\cite[Theorem~9.2.1(1)]{bjorner1999oriented}).

\begin{lemma}
\label{lem:caratheodory}
  Let $\chi$ be a chirotope on $X$. An element $t \in X$ is not
  extreme in $\chi$ if and only if there exist $a,b,c \in
  X\setminus \{t\}$ such that $\chi(t,a,b) = \chi(t,b,c) =
  \chi(t,c,a)=1$.
\end{lemma}
\begin{proof}
  Assume that there exist $a,b,c \in X\setminus \{t\}$ such
  that $\chi(t,a,b) = \chi(t,b,c) = \chi(t,c,a)=1$. For the sake of
  contradiction, assume in addition that $t$ is extreme in $\chi$.  By
  definition, this means that there exists $s \in X$ such that for all
  $x \in X \setminus \{t,s\}$, $\chi(t,s,x)=1$. We cannot have $s=a$
  as $\chi(t,a,c)=-1$. Similarly, $s \notin \{b,c\}$. Then, the tuple
  $\{s,t,a,b,c\}$ contradicts the transitivity axiom.

  \medskip

  Conversely, assume that $t$ is not extreme in $\chi$. Let us choose
  $a_0,a_1 \in X$ such that $\chi(t,a_1,a_0)=1$. Since $t$ is not
  extreme, there exists $a_2 \in X \setminus \{t,a_0,a_1\}$ such that
  $\chi(t,a_1,a_2)=-1$. If $\chi(t,a_2,a_0)=-1$ then $a = a_0$, $b=a_2$,
  $c=a_1$ are as desired. So assume otherwise. Since $t$ is not
  extreme, there exists $a_3 \in X\setminus\{t,a_2\}$ such that
  $\chi(t,a_2,a_3) = -1$; note that $a_3 \notin \{t,a_0,a_1,a_2\}$. If one
  of $\chi(t,a_3,a_0)$ or $\chi(t,a_3,a_1)$ equals $-1$ then we have
  our triple ($a\in \{a_0,a_1\}, b= a_3$ and $c=a_2$). So assume
  otherwise and using that $t$ is not extreme, pick some $a_4 \in
  X\setminus\{t,a_3\}$ such that $\chi(t,a_3,a_4) = -1$. Due to the
  previous consideration, $a_4$ must be different from $a_0, a_1$ and
  $a_2$. And so on and so forth: since $X$ is finite we eventually
  find a triple as in the lemma.
See \Cref{fig:caratheodory} for an illustration.
\end{proof}
\begin{figure}[t]
  \begin{center}
    \includegraphics{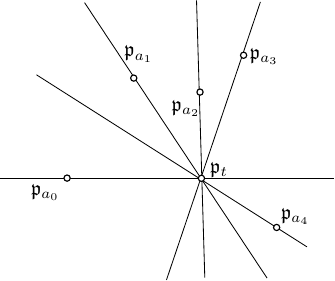}
  \end{center}
  \caption{Illustration of the proof of \cref{lem:caratheodory} for a realizable chirotope.
  The points $\mathfrak p_{a_i}$ corresponding to the elements $a_i$ are seen from $\mathfrak p_t$
  in clockwise order (since $\chi(t,a_i,a_{i+1})=-1$ for all $i$ by construction).
  Given $\mathfrak p_{a_i}$, either the line $\mathfrak p_{a_i} \mathfrak p_t$ splits the point set $\{\mathfrak p_{a_j}, j<i\}$
  in two non trivial parts (which happens here for $i=4$), and we can find a triangle containing $\mathfrak p_t$ (namely $\mathfrak p_{a_0} \mathfrak p_{a_3} \mathfrak p_{a_4}$ in the picture),
  or all points $\{\mathfrak p_{a_j}, j<i\}$ are on the same side of $\mathfrak p_{a_i} \mathfrak p_t$ (which happens for $i=2$ and $i=3$ here),
  and we let $\mathfrak p_{a_{i+1}}$ be any point on the other side (there is necessarily one, since $\mathfrak p_t$ is not extreme).
  Since the set is finite, we cannot be always in the latter situation and we eventually find a triangle containing $\mathfrak p_t$.
  }
  \label{fig:caratheodory}
\end{figure}

We now observe, for later use, that an extreme element in a chirotope
defines a total order on the other elements.

\begin{lemma}\label{lem:order}
{Let $\chi$ be a chirotope on $X$, let $ext(\chi)$ be the set of
  extreme elements of $\chi$, and take $a\in ext(\chi)$.}  The
relation $p<_aq \Leftrightarrow \chi(a,p,q)=1$ defines a strict total
order on $X\setminus\{a\}$. Moreover, letting $e_1,\ldots, e_r$
denotes the elements of $ext(\chi)\setminus \{a\}$ as ordered by
$<_a$, the cyclic order $(a,e_1,\ldots, e_r,a)$ does not depend on the
choice of $a \in ext(\chi)$.
\end{lemma}
\begin{proof}
  To prove the first assertion, we only need to check the transitivity
  of the relation $<_a$. Let us consider three elements $x,y,z\in
  X\setminus\{a\}$ such that $x<_ay$ and $y<_az$, and suppose by
  contradiction that $z<_ax$. Then we have
  $\chi(a,x,y)=\chi(a,y,z)=\chi(a,z,x)=1$ contradicting that $a$ is
  extreme, by \cref{lem:caratheodory}.

  Let $C_a$ denote the cyclic order $(a,e_1,\ldots, e_r,a)$. To prove
  the second assertion, it suffices to prove that $C_{e_1}=C_a$. For
  every $i\in\{2,\ldots ,r\}$, $e_1<_ae_i$. Therefore
  $\chi(e_1,e_i,a)=\chi(a,e_1,e_i)=1$. Hence $a$ is the maximal
  element for $<_{e_1}$. Furthermore, let $2\leq i<j\leq r$. As
  $e_1<_ae_i<_ae_j$, we have $\chi(e_i,e_j,a)=\chi(e_i,a,e_1)=1$. Then
  as $e_i$ is extreme, $\chi(e_i,e_1,e_j)=-1$ by
  \cref{lem:caratheodory}. Hence $e_i<_{e_1}e_j$ and
  $C_{e_1}=(e_1,\ldots, e_{r},a,e_1)=C_a$.
\end{proof}
\noindent
\cref{lem:order} defines a unique cyclic order on the extreme elements
of a chirotope. Geometrically, this corresponds to the anticlockwise
order on the convex hull. For any extreme element~$a$, it is hence
possible to define its previous and next extreme element for this
cyclic order, denoted respectively $a_-$ and $a_+$.  Note that $a_+$
(resp. $a_-$) is the minimum (resp. maximum) element of the total
order $<_a$.

\begin{observation}\label{obs:a-}
  Recall that, by definition, an element $a$ is extreme in a chirotope
  $\chi$ on $X$ if and only if there exists $y$ such that $\chi(a,y,z)
  =1$ for all $z$ in $X \setminus \{a,y\}$.  With the notation
  introduced above, this element $y$ is $a_+$: in particular it is
  unique.  We also have the following: $a$ in extreme in $\chi$ if and
  only if there exists $y'$ such that $\chi(a,y',z) =-1$ for all $z$
  in $X \setminus \{a,y\}$.  When it exists, this element is also
  unique and corresponds to $a_-$ with the notation introduced above.
\end{observation}

The last lemma of this section is of a more geometric nature, and only
applies to realizable chirotopes. Given a set $\cal L$ of lines in
$\R^2$, a {\em cell of the arrangement of $\cal L$} is a connected
component of $\R^2 \setminus \bigcup_{\ell \in {\cal L}} \ell$.

\begin{lemma}\label{lem:extreme_in_unbounded_cell}
Let $\chi:(X)_3 \to \{\pm 1\}$ be a realizable chirotope and let $x
\in X$ be extreme in $\chi$.  Then there exists a realization
$\mathcal P=\{\mathfrak p_x, x \in X\}$ of $\chi$ such that $\mathfrak
p_x$ is in an unbounded cell of the arrangement of the set of lines
$\big\{(\mathfrak p_y\mathfrak p_z) \colon y,z \in X \setminus \{x\},
y \neq z \big\}$.
\end{lemma}
Note that this is not true for {\em any} realization of $\chi$ but only for {\em some well-chosen}
realization of $\chi$: see \Cref{fig:extreme_bounded_unbounded} for an example of two realizations
of the same chirotope, one being as in the lemma, but not the other.

\begin{figure}[t]
\[\includegraphics[width=6cm]{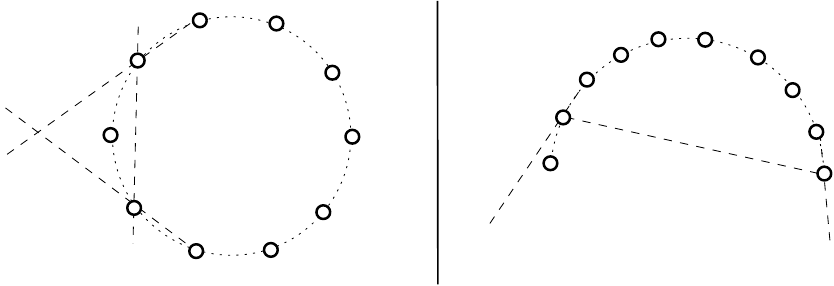}\]
\caption{On the left, 10 points uniformly distributed on a circle; on the right, the same number of points
uniformly distributed on a semicircle. Both configurations have the same chirotope. In the circle case, each point is in a
bounded cell of the line arrangement defined by the other points (as illustrated for the left-most point).
On the other hand in the semi-circle case,
 the left-most point is in an unbounded region of the line arrangement defined by the other points.}
 \label{fig:extreme_bounded_unbounded}
\end{figure}

\begin{proof}
This proof relies on the notion of projective transform, which we now recall; see {\em e.g.} the book of Samuel~\cite{samuel1988projective} or Richter-Gebert~\cite{richter2011perspectives} for a primer on projective geometry.
    For a point $\mathfrak p \in \R^2$ let $\hat{\mathfrak p}$ be the column vector $\hat{\mathfrak p} \eqdef (p_x,p_y,1)^t$, and conversely, given $u \in \R^3$ with $u_z \neq 0$ let $\bar u \eqdef (u_x/u_z, u_y/u_z)$. A projective transform is a map $f:\mathfrak p \mapsto \overline{A\hat{\mathfrak p}}$ with $A \in \R^{3 \times 3}$ a nonsingular matrix. Such a transform is undefined on a line of $\R^2$ (those points $\mathfrak p$ such that $\pth{A\hat{\mathfrak p}}_z = 0$); this line is said to be ``sent to infinity''.
    Other lines of $\mathbb R^2$ are sent to lines in $\mathbb R^2$.
    Indeed, a line is lifted to a 2-dimensional vector space
    by the application ${\mathfrak p} \mapsto \hat{\mathfrak p}$ (taking the cone generated by the image),
    a 2-dimensional vector space is sent to another 2-dimensional vector space by $A$ (since $A$ is nonsingular), and finally the projection $u \mapsto \overline u$ maps back 2-dimensional vector spaces to lines (or to infinity in the case of the plane $u_z=0$).

    For any $\mathfrak p,\mathfrak q,\mathfrak r$ in $\R^2$ we have $\chi(\mathfrak p,\mathfrak q,\mathfrak r) = \textrm{sign} \det(\hat{\mathfrak p},\hat{\mathfrak q},\hat{\mathfrak r})$.
    Indeed, $\textrm{sign} \det(\hat{\mathfrak p},\hat{\mathfrak q},\hat{\mathfrak r})=+1$ if and only if the basis $(\hat{\mathfrak p},\hat{\mathfrak q},\hat{\mathfrak r})$ of $\mathbb R^3$ is positively oriented, which is equivalent
    to the points $(\mathfrak p,\mathfrak q,\mathfrak r)$ in $\mathbb R^2$
    being in counterclockwise order.

    Now, let $f$ be a projective transform with matrix $A \in \R^{3 \times 3}$. We have
\[\begin{aligned}
\chi(f(\mathfrak p),f(\mathfrak q),f(\mathfrak r)) & = \sgn \det(\widehat{f(\mathfrak p)},\widehat{f(\mathfrak q)},\widehat{f(\mathfrak r)})\\
& = \sgn\pth{(A\hat{\mathfrak p})_z(A\hat{\mathfrak q})_z(A\hat{\mathfrak r})_z} \sgn \det(A\hat{\mathfrak p}, A\hat{\mathfrak q}, A\hat{\mathfrak r})\\
& = \sgn\pth{(A\hat{\mathfrak p})_z(A\hat{\mathfrak q})_z(A\hat{\mathfrak r})_z} \sgn(\det A) \chi(p,q,r)
\end{aligned}\]
The sign of the affine function $\mathfrak p \mapsto (A\hat{\mathfrak p})_z$ is determined by the position of $\mathfrak p$ lies with respect to the line $(A\hat{\mathfrak p})_z=0$, which is the line sent to infinity by $f$.
Consequently, if $\cal P$ is a point set, $\ell$ a line that does not split $\cal P$, and $f$ a projective transform sending $\ell$ to infinity,
then the map $f$ either preserves all orientations of three points in $\cal P$ or reverses all of them.

Let us now consider a realizable chirotope $\chi$ on $X$ and let $x$ be extreme in $\chi$,
as in the lemma.
We choose an arbitrary realization $\mathcal Q$ of $\chi$ and call $\mathfrak q\eqdef \mathfrak q_x$ the point labeled by $x$.
Let us choose a line $\ell$ such that
\begin{itemize}
  \item $\ell$ does not contain any point of $\mathcal Q$ and all points of $\mathcal Q$ are on the same side of $\ell$,
  \item there exists a point $\mathfrak r$ in $\ell$ (in particular $\mathfrak r \notin \mathcal Q$)
    such that the segment $[\mathfrak{qr}]$ does not intersect any line spanned by two points of $\cal Q \setminus \{\mathfrak q\}$.
\end{itemize}
(Such a choice is always possible since $\mathfrak q$ is extreme.)
Let $f$ be a projective transform sending $\ell$
to infinity.
Since the map $f$ sends lines to lines and since $f(\mathfrak r)$
is at infinity, $f$
sends the half-open segment $[\mathfrak q\mathfrak r))$
to a ray $[f(\mathfrak q)f(\mathfrak r))$ going to infinity.
From the second condition above, this ray
does not meet any line spanned by two points of $f\pth{\cal P \setminus \{\mathfrak p\}}$.
Hence $f(\mathfrak p)$ is in an unbounded cell of the arrangement
defined by $f\pth{\cal P \setminus \{\mathfrak p\}}$.
In particular, if $f$ preserves the orientation of all triples in $\mathcal Q$, then $\mathcal P=f(\mathcal Q)$ is a realization of $\chi$
with the desired property.
On the other hand, if $f$ reverses all orientations of triples in $\mathcal Q$, then the mirror image of $f(\mathcal Q)$ is a realization of $\chi$ with that property.
\end{proof}

\section{Bowties and modules}
\label{s:bam}

In this section, we study the bowtie operation $\Bowtie$ defined in the introduction,
introduce formally a notion of modules, and analyze the connection between
these two notions.

\subsection{The bowtie operation}
\label{subsec:subs}

We recall that the bowtie operation $\Bowtie$ is defined in \cref{eq:def_bowtie}.
It takes as input two sign functions $\chi$ and $\xi$ with marked elements $x^*$ and $y^*$
and outputs a sign function $\kappa$, denoted $\chi  \join{x^*}{y^*} \xi$.
Our first goal is to prove \cref{p:whenischirotope}, which we split into several statements
for convenience.
\medskip

We first note that $\chi$ and $\xi$ are essentially isomorphic to
restrictions of $\kappa$ (e.g., to obtain $\chi$, we restrict $\kappa$
to $(X \cup \{y\})_3$ where $y$ is an arbitrary element of $Y$, and
substitute $y^*$ for $y$).  Therefore, if $\kappa$ is a chirotope,
then $\chi$ and $\xi$ are chirotopes, and if furthermore $\kappa$ is
realizable, then $\chi$ and $\xi$ are also realizable. Here are the
converse statements.

\begin{proposition}\label{prop:Substitution_is_chirotope}
  Let $\chi$ and $\xi$ be chirotopes on $X\cup \{x^*\}$ and $Y\cup
  \{y^*\}$, with $|X|,|Y| \ge 2$. The sign function $\chi
  \join{x^*}{y^*} \xi$ is a chirotope if and only if $x^*$ and $y^*$
  are extreme in $\chi$ and~$\xi$.
\end{proposition}

\begin{proof}
Let $\kappa \eqdef \chi \join{x^*}{y^*} \xi$. Assume that $x^*$ is not
extreme in $\chi$. From \cref{lem:caratheodory}, there exist $a, b, c
\in X$ such that
  \begin{equation}
  \label{eq:xstar_inside_abc}
  \chi(x^*,a,b) = \chi(x^*,b,c) = \chi(x^*,c,a)=+1.
  \end{equation}
 Let now $y$ and $y'$ be two distinct elements of $Y$.  By definition
 of $\kappa$, we have
 \begin{equation}
 \label{eq:violating_transitivity1}
  \kappa(y,y',a) = \kappa(y,y',b) = \kappa(y,y',c) = \xi(y,y',y^*),
  \end{equation}
  which we can assume to be $+1$, up to switching $y$ and $y'$.
 Furthermore, \cref{eq:xstar_inside_abc} implies
  \begin{equation}
 \label{eq:violating_transitivity2}
  \kappa(y,a,b) = \kappa(y,b,c) = \kappa(y,c,a)=+1.
  \end{equation}
  \cref{eq:violating_transitivity1,eq:violating_transitivity2} together violates the transitivity
  axiom of chirotopes, hence, $\kappa$ is not a chirotope.
  Similarly if $y^*$ is not extreme in  $\xi$, then $\kappa$ is not a chirotope either.

\medskip

  Conversely, let us assume that $x^*$ and $y^*$ are extreme in $\chi$
  and $\xi$, and check that $\kappa$ is a chirotope. The function $\kappa$ is
  clearly a sign function, and we need to check the interiority and transitivity axioms. Let us start with the interiority
  axiom.
  Assume for the sake of contradiction that there exist
  distinct $a,b,c,t \in X \cup Y$ such that
  \begin{equation}
    \label{eq:inter-bowtie} \kappa(t,a,b)=\kappa(t,b,c)=\kappa(t,c,a) = - \kappa(a,b,c).
  \end{equation}
  We cannot have all of $a,b,c,t$ in $X$ as $\chi$ satisfies the
  interiority axiom. Similarly, we cannot have a single element in
  $Y$, as replacing it by $x^*$ would yield a quadruple violating the
  interiority axiom for $\chi$. Symmetric arguments imply that we must
  have two elements in $X$ and two in $Y$; by symmetry we can assume
  that $a,b \in X$ and $c,t \in Y$. But then $\kappa(a,b,c)=
  \chi(a,b,x^*) =\kappa(a,b,t)$, contradicting
  \cref{eq:inter-bowtie}. Hence $\kappa$ satisfies the
  interiority axiom.

  Now to the transitivity axiom. Take any distinct $a,b,c,s,t \in X
  \cup Y$ and assume for the sake of contradiction that
  \begin{equation}
    \label{eq:transi-bowtie}
    \kappa(t,s,a)=\kappa(t,s,b)=\kappa(t,s,c) = \kappa(a,b,t) = \kappa(b,c,t)=\kappa(c,a,t) = 1.
  \end{equation}
  As above, we can assume that $X$ and $Y$ contain respectively two
  and three elements of $\{a,b,c,s,t\}$.
  \begin{itemize}
  \item We cannot have $\{s,t\} \subseteq X$ and $\{a,b,c\} \subseteq Y$ as this would imply $\xi(a,b,y^*) = \xi(b,c,y^*)=\xi(c,a,y^*)=1$, and thus $y^*$ would not be extreme in $\xi$ by Lemma~\ref{lem:caratheodory}.
  \item We cannot have $\{s,a\} \subseteq X$ and $\{b,c,t\} \subseteq Y$ as this would imply $\kappa(c,t,a) = \xi(c,t,y^*) = \kappa(c,t,s)$, contradicting \cref{eq:transi-bowtie}.
  \item We cannot have $\{t,a\} \subseteq X$ and $\{b,c,s\} \subseteq Y$ as this would imply $\kappa(a,t,b)=\chi(a,t,x^*) = \kappa(a,t,c)$, contradicting \cref{eq:transi-bowtie}.
  \item We cannot have $\{b,c\} \subseteq X$ and $\{a,s,t\} \subseteq Y$ as this would imply $\kappa(a,t,b)=\xi(a,t,y^*) = \kappa(a,t,c)$, contradicting \cref{eq:transi-bowtie}.
  \end{itemize}
  The other partitions are ruled out using the cyclic symmetry of $(a,b,c)$.
\end{proof}

\begin{proposition}\label{prop:Realization_Substitution}
Let $\chi$ and $\xi$ be {\em realizable} chirotopes on $X\cup \{x^*\}$
and $Y\cup \{y^*\}$, with $|X|,|Y| \ge 2$.  The sign function $\chi
\join{x^*}{y^*} \xi$ is a realizable chirotope if and only if $x^*$
and $y^*$ are extreme in, respectively, $\chi$ and $\xi$.
\end{proposition}

\begin{proof}
Let $\kappa \eqdef \chi \join{x^*}{y^*} \xi$. One direction follows
from \cref{prop:Substitution_is_chirotope}, as if $x^*$ or $y^*$ is
not extreme, then $\kappa$ is not even a chirotope. So let us assume
that $\chi$ and $\xi$ are realizable and that $x^*$ and $y^*$ are
extreme in $\chi$ and $\xi$, and build a point set realizing
$\kappa$. The reader may want to look at
\Cref{fig:realization-substitution} while reading the proof.

\begin{figure}[t]
\[\includegraphics[width=14cm]{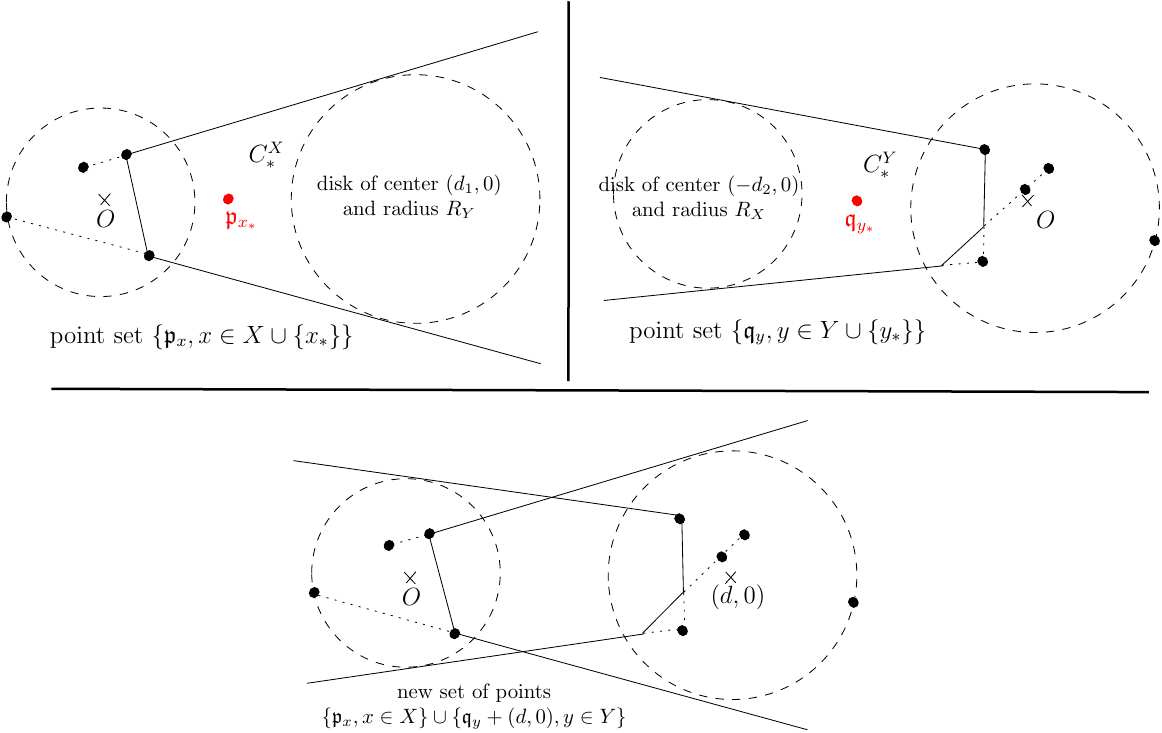}\]
\caption{Illustration  of the proof of \cref{prop:Realization_Substitution}. Marked points are drawn in red. }
\label{fig:realization-substitution}
\end{figure}

\medskip

Since $x^*$ is extreme in $\chi$, by
Lemma~\ref{lem:extreme_in_unbounded_cell} there exists a point set
$\p$ labeled by $X \cup \{x^*\}$ with chirotope $\chi$ such that $\pt
p_{x^*}$ is in an unbounded cell $C^X_*$ of the arrangement of the set
of lines $\big\{(\pt p_x\pt p_{x'}), x,x' \in X\}$. Up to rotating and
translating this point set, we can assume that $\pt p_{x^*}$ is on the
positive real axis and that the half-line $[\mathfrak
  p_{x^*},+\infty)$ does not cross any line $(\mathfrak p_x\mathfrak
  p_{x'})$ with $x,x' \in X$. We can further assume that no line
  $(\mathfrak p_x\mathfrak p_{x'})$ with $x,x'$ in $X$ is horizontal,
  since the points are in generic position. We set $R_X = \max_{x \in
    X} \|\mathfrak p_x\|$. Similarly, there exists a point set $\cal
  Q$ labeled by $Y \cup \{y^*\}$ with chirotope $\xi$ such that $\pt
  q_{y^*}$ is on the negative real axis and the half-line
  $(-\infty,\mathfrak q_{x^*}]$ does not cross any line $(\pt
q_y,\mathfrak q_{y'})$ with $ y,y' \in Y$. We set $R_Y=\max_{y \in
  Y} \|\mathfrak q_y\|$.

\medskip

Now, let $\mathfrak p$ be the point with coordinate $(d,0)$ for some
$d>\|\mathfrak p_{x^*}\|$. The distance between $\mathfrak p$ and the
union of the lines $(\mathfrak p_x\mathfrak p_{x'})$ with $x,x' \in
X$ goes to infinity as $d$ goes to $+\infty$: indeed it is positive,
piecewise affine and could only remain constant if the cell $C_*^X$ were
bounded by horizontal lines (which we ruled out).
 Thus, for $d$ large enough, say $d >d_1$,
 the disk of radius $R_Y$ and center $(d,0)$ entirely lies in $C^X_*$.
  Similarly for $d$ large enough, say $d>d_2$,
  the disk of radius $R_X$ and center $(-d,0)$ entirely lies in $C^Y_*$.

We take $d>\max(d_1,d_2)$ and consider the point set
\[\big\{\mathfrak p_x, x \in X\big\} \cup \big\{\mathfrak q_y +(d,0), y \in Y\big\},\]
i.e.~we take all points $\mathfrak p_x$ (except $\mathfrak p_{x^*}$) as they are,
and all points $\mathfrak q_y$ (except $\mathfrak q_{y^*}$) translated by the vector $(d,0)$.
Let $\kappa:(X \cup Y)_3 \to \{\pm 1\}$ be the chirotope of this point set.
Clearly, $\kappa$ is a realizable chirotope, and our goal is now to prove that $\kappa = \chi \join{x^*}{y^*} \xi$.
Obviously,
\[\begin{cases}
\kappa(x_1,x_2,x_3) = \chi(x_1,x_2,x_3) \text{ if $x_1,x_2,x_3$ are all in $X$;}\\
\kappa(y_1,y_2,y_3)= \xi(y_1,y_2,y_3) \text{ if $y_1,y_2,y_3$ are all in $Y$.}
\end{cases} \]
Take $x_1$,$x_2$ in $X$ and $y$ in $Y$. The point $\mathfrak q_y
+(d,0)$ is in the circle of radius $R_Y$ and center $(d,0)$, and is
therefore in $C^X_*$; in other words, $\mathfrak q_y +(d,0)$ is in the
same cell of the arrangement of the lines $\big\{(\mathfrak
p_x\mathfrak p_{x'}), x,x' \in X\}$ as $\pt p_{x^*}$, which implies
that $\kappa(x_1,x_2,y) = \chi(x_1,x_2,x^*)$.The case where
$y_1$,$y_2$ are in $Y$ and $x$ in $X$ is handled similarly, after
translating all the points by $(-d,0)$ (which does not change the
chirotope).  In this case, $\kappa(x,y_1,y_2) = \xi(y^*,y_1,y_2)$.
Altogether, this shows that $\kappa = \chi \join{x^*}{y^*} \xi$,
proving that $\chi \join{x^*}{y^*} \xi$ is realizable.
\end{proof}

\cref{prop:Substitution_is_chirotope,prop:Realization_Substitution} emphasize the special role of extreme points for the bowtie operation. The following proposition will therefore be useful to iterate bowtie operations.

\begin{proposition}\label{prop:extreme_stay_extreme}
  Let $\chi$ and $\xi$ be chirotopes on disjoint sets $X\cup \{x^*\}$
  and $Y\cup \{y^*\}$, with $|X|,|Y| \ge 2$, and assume that $x^*$ and
  $y^*$ are extreme in $\chi$ and in $\xi$.  The extreme elements of
  $\chi\join{x^*}{y^*}\xi$ are the extreme elements of $\chi$ other
  than $x^*$ and those of $\xi$ other than $y^*$.
\end{proposition}

\begin{proof}
  Let $\kappa\eqdef \chi\join{x^*}{y^*}\xi$. Let $x_e$ be
  extreme in $\chi$, different from $x^*$. There exists $x'_e \in X$
  such that
  \[ \forall x \in X\cup
  \{x^*\} \setminus \{x_e,x'_e\}, \quad \kappa(x_e,x'_e,x) =
  \chi(x_e,x'_e,x) =1.\]
  If $x'_e \neq x^*$, then for all $y \in Y$ we have
    $\kappa(x_e,x_e',y) = \chi(x_e,x'_e,x^*)=1$ and $x_e$ is extreme
    in $\kappa$. If $x'_e = x^*$, then we use the equivalent
    definition (see \cref{obs:a-}) that $x_e$ is extreme in $\chi$ if
    and only if there exists $x_e'' \in X$ such that
  \[ \forall x \in X\cup \{x^*\} \setminus \{x_e,x''_e\}, \quad \chi(x_e,x_e'',x) =-1,\]
  and now that $x_e''\neq x^*$ we are back to the previous case, {\em
    mutatis mutandis}. Altogether, the extreme elements of $\chi$
  different from $x^*$ are extreme for $\kappa$.  Similarly, the
  extreme elements of $\xi$ different from $y^*$ are extreme for
  $\kappa$.

\medskip

Conversely, let $z_e$ be extreme in $\kappa$, and assume w.l.o.g.~that
$z_e \in X$.  By definition, there exists $z'_e \in X\cup Y$ such that
for all $z \in X\cup Y \setminus \{z_e,z'_e\}$, we have
$\kappa(z_e,z'_e,z) =1$. On the one hand, if $z'_e \in X$, then
$\chi(z_e,z'_e,x) = \kappa(z_e,z'_e,x) = 1$ for all $x \in X\setminus
\{z_e,z'_e\}$. In addition, choosing any element $y \in Y$, we have
$\chi(z_e,z'_e,x^*) = \kappa(z_e,z'_e,y) = 1$, and we recognize that
$z_e$ is extreme in $\chi$. On the other hand, if $z'_e \in Y$, then
for all $x \in X \setminus \{z_e\}$, we have $\chi(z_e,x^*,x) =
\kappa(z_e,z'_e,x) =1$, and $z_e$ is again extreme in~$\chi$.
\end{proof}

\cref{prop:Substitution_is_chirotope,prop:Realization_Substitution,prop:extreme_stay_extreme}
correspond to \cref{p:whenischirotope} in the introduction.  We
conclude this section with two additional results, used respectively
in \cref{sec:canonical_tree,sec:triangulations}.

\medskip

For the first result, we recall that a chirotope is called convex if all its elements are extreme.
The following statement is an immediate corollary of \cref{prop:extreme_stay_extreme}.
\begin{corollary}\label{lem:bowtie_of_convex}
  Let $\kappa$ be a chirotope with a nontrivial decomposition $\kappa =
  \chi\join{x^*}{y^*}\xi$. The chirotope $\kappa$ is convex if and
  only if both $\chi$ and $\xi$ are convex.
\end{corollary}

\subsection{Modules} \label{subsec:modules}

Let us now consider inverting the bowtie operation: given a chirotope
$\kappa$ on ground set $Z$, we want a bipartition $X\uplus Y$ of $Z$,
and chirotopes $\chi$ and $\xi$ on ground sets $X\cup \{x^*\}$ and
$Y\cup \{y^*\}$ (for new elements $x^*$ and $y^*$) such that $\kappa=
\chi\join{x^*}{y^*}\xi$.  We call such a decomposition $\kappa=
\chi\join{x^*}{y^*}\xi$ a \emph{bowtie decomposition} of $\kappa$.

\smallskip

These considerations lead naturally to the following notion.

\begin{definition}\label{def:modules}
  A subset $X\subseteq Z$ is a \emph{module} of a chirotope $\kappa$
  on $Z$ if for every $x, x' \in X$ and $y, y' \in Z \setminus
  X$, one has $\kappa(x,x',y)=\kappa(x,x',y')$ and
  $\kappa(x,y,y')=\kappa(x',y,y')$. The bipartition $\{X,Z\setminus
  X\}$ of $Z$ is called a \emph{modular bipartition} of $\kappa$.
\end{definition}

\noindent
Note that $X$ is a module of $\kappa$ if and only if $Z\setminus X$ is
a module of $\kappa$. The notion of module is chosen so that the
following observation holds.

\begin{observation}
\label{prop:module_and_join}
Let $\kappa$ be a chirotope on ground set $Z$ and $\{X,Y\}$ a
bipartition of $Z$.  There exist chirotopes $\chi$ and $\xi$ on
disjoint ground sets $X\cup \{x^*\}$ and $Y\cup \{y^*\}$ such that
$\kappa= \chi\join{x^*}{y^*}\xi$ if and only if $X$ is a module of
$\kappa$. Moreover, when the chirotopes $\chi$ and $\xi$ exist, they
are unique up to the choice of the labels for $x^*$ and $y^*$.
\end{observation}

\begin{remark}
If one of $X$ or $Y$ has size at most $1$, then $\{X,Y\}$ is always a modular bipartition of $\kappa$, but the size of $\chi$ or $\xi$ is at least that of $\kappa$. We refer to such modular bipartitions, modules or bowtie decomposition as {\em trivial}. Note that in a nontrivial bipartition $\{X,Y\}$ of $\kappa$, both $X$ and $Y$ have size at least $2$, which implies in particular that both $\chi$ and $\xi$ in the decomposition $\kappa= \chi\join{x^*}{y^*}\xi$ have ground sets of size at least $3$.
\end{remark}

\begin{remark}\label{rk:module-realisable}
When $\kappa$ is realizable, its nontrivial modular bipartitions have
a simple geometric interpretation: $X$ is a nontrivial module if there exists a
realization of $\kappa$ such that no line through two points with labels in $X$
separates two points with labels outside $X$, and conversely no
line through two points with labels outside $X$ separates
two points with labels in $X$. This explains the connection
with the notion of \emph{mutually avoiding sets} given in the
introduction: a modular bipartition is a pair of mutually avoiding
sets covering the whole chirotope.\end{remark}

The next lemma studies the stability of modules under intersection,
and will be useful in the next section.

\begin{lemma}\label{lm:intersection_modules}
Let $X_1, X_2$ be two modules of a chirotope $\kappa$ with ground set $Z$. If $X_1\cup X_2\neq Z$, then $X_1\cap X_2$ is a module of $\kappa$.
\end{lemma}

\begin{proof}
  Let $x,x'\in X_1\cap X_2$ and $y,y'\in Z\setminus (X_1\cap
  X_2)$. Let $p\in Z\setminus (X_1\cup X_2)$. Since $y$ is outside
  $X_1 \cap X_2$, the set $\{y,p\}$ is disjoint from $X_1$ or from
  $X_2$. In either case, we have $\kappa(x,x',y)= \kappa(x,x',p)$. A
  symmetric argument yields $ \kappa(x,x',p)=\kappa(x,x',y')$, so we
  have $\kappa(x,x',y)=\kappa(x,x',y')$.

  \smallskip

  It remains to argue that $\kappa(x,y,y')=\kappa(x',y,y')$. This is
  clear if both $y$ and $y'$ are outside of $X_1$, as $X_1$ is a
  module, so assume otherwise. Suppose, without loss of generality,
  that $y \in X_1 \setminus X_2$, and $y'\in X_2 \setminus X_1$. As
  $x,x',y\in X_1$ and $p,y'\notin X_1$ we have $\kappa(x,y,y')
  =\kappa(x,y,p)$ and $\kappa(x',y,y') =\kappa(x',y,p)$. As $p,y\notin
  X_2$ and $x,x'\in X_2$, $\kappa(x,y,p)=\kappa(x',y,p)$. Altogether,
  we get $\kappa(x,y,y') = \kappa(x',y,y')$ as desired.  
\end{proof}

\begin{remark}
  The condition $X_1\cup X_2\neq Z$ in
  Lemma~\ref{lm:intersection_modules} is indeed necessary. To see
  this, consider the chirotope induced by the following point set,
  where $X_1=\{a,b,x,y\}$ and $X_2=\{c,d,x,y\}$ are nontrivial modules
  but $X_1\cap X_2=\{x,y\}$ is not a module:
\begin{center}
	\begin{tikzpicture}
	\node (x) at (0,.7) {$\bullet$};
	\draw (x) node[below ] {$x$} ;
	\node (y) at (0,-.7) {$\bullet$};
	\draw (y) node[below ] {$y$} ;

	\node (e) at (1,0) {$$};
	\node (c) at ($(e)+(0,.15)$) {$\bullet$};
	\draw (c) node[right ] {$c$};
	\node (d) at ($(e)-(+0.09,.25)$) {$\bullet$};
	\draw (d) node[right ] {$d$};

	\node (c) at (-1,0) {$$};
	\node (a) at ($(c)+(-0.12,.15)$) {$\bullet$};
	\draw (a) node[left ] {$a$};
	\node (b) at ($(c)-(0,.15)$) {$\bullet$};
	\draw (b) node[left ] {$b$};
\end{tikzpicture}
\end{center}
\end{remark}

\subsection{Commutativity and associativity}\label{subsec:sequences_bowtie}

\begin{proposition}\label{prop:commutativity_asso_bowtie}
	Let $\chi$, $\xi$, $\kappa$ be chirotopes on disjoint ground sets $X\cup \{x^*\}$, $Y\cup \{y_1^*,y_2^*\}$ and $Z\cup \{z^*\}$, with $|X|,|Z| \ge 2, |Y|\geq 1$. If the starred elements are extreme in their respective chirotopes, then $\chi\join{x^*}{y_1^*}\xi=\xi\join{y_1^*}{x^*}\chi$ and $(\chi\join{x^*}{y_1^*}\xi)\join{y_2^*}{z^*}\kappa=\chi\join{x^*}{y_1^*}(\xi\join{y_2^*}{z^*}\kappa)$.
\end{proposition}
\begin{proof}
	The commutativity comes from the symmetry of the definition of
        the bowtie operation. Let us focus on the associativity
        condition.  Let
        $\lambda_1=(\chi\join{x^*}{y_1^*}\xi)\join{y_2^*}{z^*}\kappa$,
        and
        $\lambda_2=\chi\join{x^*}{y_1^*}(\xi\join{y_2^*}{z^*}\kappa)$,
        which are both sign functions on $X\cup Y\cup Z$. Let us take
        three elements $t,u,v$ in $X\cup Y\cup Z$.  Suppose first that
        $t,u,v$ all belong to the set $X$. Then by definition of the
        bowtie
        $\lambda_1(t,u,v)=[(\chi\join{x^*}{y_1^*}\xi)\join{y_2^*}{z^*}\kappa]
        (t,u,v)=(\chi\join{x^*}{y_1^*}\xi)(t,u,v)=\chi(t,u,v)$, and
        similarly $\lambda_2(t,u,v)=\chi(t,u,v)$. The cases where
        $t,u,v$ are all in $Y$, or all in $Z$, are similar.  If
        $t,u,v$ all belong to different sets among $X$, $Y$ and $Z$,
        for example if $t\in X$, $u\in Y$, $v\in Z$, then
        $\lambda_1(t,u,v)=(\chi\join{x^*}{y_1^*}\xi)(t,u,y_2^*)=\xi(y_1^*,u,y_2^*)$,
        and
        $\lambda_2(t,u,v)=(\xi\join{y_2^*}{z^*}\kappa)(y_1^*,u,v)=\xi(y_1^*,u,y_2^*)$. The
        final cases where one element is isolated in one set are very
        similar.
\end{proof}

\begin{remark}
  When using the associativity of bowties, one should pay particular
  attention to the ground sets of the chirotopes at play.  For
  instance, with the notation of
  \cref{prop:commutativity_asso_bowtie}, writing
  $(\xi\join{y_1^*}{x^*}\chi)\join{y_2^*}{z^*}\kappa=\xi\join{y_1^*}{x^*}(\chi\join{y_2^*}{z^*}\kappa)$
  does not make sense, as in the right-hand term $y_2^*$ does not
  belong to the ground set of $\chi$ (even though the left-hand term
  is well-defined).
\end{remark}

\section{Existence and uniqueness of the canonical chirotope tree}
\label{sec:canonical_tree}

{We now examine how iterated bowtie products can be described and manipulated via trees.}

\subsection{Chirotope trees and their chirotopes} \label{ssec:chirotope_trees}

We first recall some definitions from the introduction.  A {\em
  chirotope tree} is a tree whose nodes are decorated with chirotopes
on disjoint ground sets, {each of size at least $3$}, and whose
edges $\{v_1,v_2\}$ select an element in the ground set of each of 
the chirotopes decorating $v_1$ and $v_2$, 
in a way that (i) no element is selected more than once, and (ii) each
selected element is extreme in its chirotope.  To be complete, let us
stress out that all trees we consider are graph-theoretical trees $T=(V_T,E_T)$,
i.e.~simple connected graphs without cycles.  In
particular, trees are always unrooted and unordered (neighbors of a
given node are not ordered).

Given a chirotope tree $T$, one associates to it a sign function $\chi_T$ on $\bigcup_{w \in V_T} X_w$ defined by
\begin{equation}
  \chi_T(x,y,z) \eqdef \chi_v(R(x,v),R(y,v),R(z,v)),
\label{eq:chirotope_tree}
\end{equation}
where $v$ is the branching point of the paths from $x$ to $y$, $y$ to
$z$ and $x$ to $z$, and $R(x,v)$, $R(y,v)$ and $R(z,v)$ the
representatives of $x$, $y$ and $z$ in $v$ (i.e.~the proxy points of
$\chi_v$ to which the path going from $v$ to the nodes containing $x$,
$y$ and $z$ are attached).  If two or more elements among $x$, $y$ and
$z$ belong to the same node, then $v$ is taken to be this node, and we
use the convention that $R(s,v)=s$ for $s$ in $X_v$.

It is easy to see that this defines a sign function on $\bigcup_{w \in V_T} X_w$;
we now prove \cref{prop:chirotope_tree}, which state that it is always a chirotope,
and furthermore realizable if all chirotopes decorating the nodes of $T$ are realizable.
Let us start with a lemma.

\begin{lemma}\label{lem:chiT_is_join_2}
  Let $T$ be a chirotope tree and $e=\{v_1,v_2\}$ be an edge of $T$.
  Let $s_1$ and $s_2$ be the elements selected by $e$ in $v_1$ and
  $v_2$ respectively.  If $T_1$ and $T_2$ denote the two chirotope
  trees obtained by removing $e$ from $T$, with $v_i \in T_i$, then
  $\chi_{T}=\chi_{T_1}\join{s_1}{s_2}\chi_{T_2}$.
\end{lemma}
\begin{proof}
  First, we note that after the removal of $e$, $s_1$ and $s_2$ are no
  longer proxies in $T_1$ and $T_2$, and the bowtie
  $\kappa=\chi_{T_1}\join{s_1}{s_2}\chi_{T_2}$ is thus well-defined.

  \smallskip

  It remains to check that $\kappa=\chi_T$. Let $X$ denote the ground
  set of $\kappa$. This set decomposes into $X=X_1 \uplus X_2$, where
  $X_i \uplus \{s_i\}$ is the ground set of $T_i$, for $i=1,2$. It is
  easy to check that given $x \in X$ and a node $v \in T$, we have
  \[ R_T(x,v) = \left\{\begin{array}{ll} R_{T_1}(x,v) & \text{if } v \in T_1\text{ and } x
  \in X_1,\\
  R_{T_1}(s_1,v) & \text{if } v \in T_1 \text{ and } x
  \in X_2,\\
  R_{T_2}(s_2,v) & \text{if } v \in T_2 \text{ and } x
  \in X_1,\\
  R_{T_2}(x,v) & \text{if } v \in T_2 \text{ and } x
  \in X_2.\\
  \end{array}\right.\]
  From this, it is easy to check that $\kappa(x,y,z)=\chi_T(x,y,z)$
  for any $x,y,z$ in $X$. For instance, if $x,y \in X_1$ and $z \in
  X_2$, then the node $v\eqdef v(x,y,z)$ lies in $T_1$ and
 \[R_{T}(x,v)=R_{T_1}(x,v), \qquad R_T(y,v)=R_{T_1}(y,v), \qquad R_T(z,v)=R_{T_1}(s_1,v).\]
 We therefore get
 \[\chi_T(x,y,z) = \chi_v(R_{T_1}(x,v),R_{T_1}(y,v),R_{T_1}(s_1,v))=\chi_{T_1}(x,y,s_1)=\kappa(x,y,z),\]
 where the middle equality uses \cref{eq:chirotope_tree} for $T_1$, and the last equality
 follows from the definition of the bowtie operation (\cref{eq:def_bowtie}).
 Other cases (two elements in $X_2$ and one in $X_1$, or the three elements in the same set)
 are similar or easier. This proves  $\kappa=\chi_T$ as wanted.
\end{proof}

\smallskip

The proof of \cref{prop:chirotope_tree} is now a straightforward
induction using \cref{lem:chiT_is_join_2} and the fact that the bowtie
of two (realizable) chirotopes whose starred elements are extreme is a
(realizable) chirotope (see
\cref{prop:Substitution_is_chirotope,prop:Realization_Substitution,prop:extreme_stay_extreme}).
Iterating this lemma, we can actually compute $\chi_T$ starting from
the sign functions $\chi_v$ and iterating bowtie operations.  For
example, if $T$ is the chirotope tree on
\Cref{fig:example-chirotope-tree}, we have
\begin{equation}\label{eq:chiT-example-first} \chi_T = (\chi_1 \join{x_1^*}{x_2^*} (\chi_2 \join{y_2^*}{x_3^*}\chi_3)) \join{z_2^*}{x_4^*}\chi_4.
\end{equation}

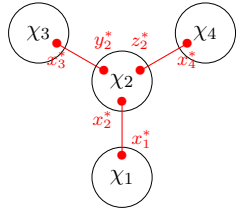
\begin{figure}
\[
\begin{tikzpicture}[baseline={(0,-.3)},scale=0.8, every node/.style={transform shape}]

\node[circle,draw,minimum size=1cm,very thin] (n2) at (0,0) {$\chi_2$};
\node[circle,draw,minimum size=1cm,very thin] (n1) at ($(n2)+(-90:1.6)$) {$\chi_1$};
\node[circle,draw,minimum size=1cm,very thin] (n3) at ($(n2)+(150:1.6)$) {$\chi_3$};
\node[circle,draw,minimum size=1cm,very thin] (n4) at ($(n2)+(30:1.6)$) {$\chi_4$};

\node[red] (y1*) at ($(n1)+(0,.35)$) {$\bullet$};
\node[red] (x3*) at ($(n3)+(-30:0.35)$) {$\bullet$};
\node[red] (x4*) at ($(n4)+(180+30:0.35)$) {$\bullet$};
\node[red] (y2*) at ($(n2)+(180-30:0.35)$) {$\bullet$};
\node[red] (z2*) at ($(n2)+(30:0.35)$) {$\bullet$};
\node[red] (x2*) at ($(n2)+(-90:0.35)$) {$\bullet$};

\node[red,above right=-0.3 of y1*,font=\footnotesize] {$x_1^*$};
\node[red,below left=-0.3 of x2*,font=\footnotesize] {$x_2^*$};
\node[red,above=-0. of y2*,font=\footnotesize] {$y_2^*$};
\node[red,above=-0. of z2*,font=\footnotesize] {$z_2^*$};
\node[red,below=-0.23 of x3*,font=\footnotesize] {$x_3^*$};
\node[red,below=-0.23 of x4*,font=\footnotesize] {$x_4^*$};

\draw[red] (x3*.center)--(y2*.center);
\draw[red] (x2*.center)--(y1*.center);
\draw[red] (x4*.center)--(z2*.center);
\end{tikzpicture}
		\]
  \caption{An example of chirotope tree.
  \label{fig:example-chirotope-tree}}
\end{figure}
We note that, when iterating \cref{lem:chiT_is_join_2}, we can choose
at each step which edge plays the role of $e$ in the lemma.
This leads to different expressions of $\chi_T$ as a sequence of bowtie operations.
Considering again the chirotope tree $T$ on \Cref{fig:example-chirotope-tree},
we could have written
\begin{equation}\label{eq:chiT-example-second} \chi_T =\chi_1 \join{x_1^*}{x_2^*} ((\chi_2 \join{z_2^*}{x_4^*}\chi_4) \join{y_2^*}{x_3^*}\chi_3)
\end{equation}
The equivalence of \cref{eq:chiT-example-first,eq:chiT-example-second}
follows from the commutativity and associativity properties of the bowtie operation
(\cref{prop:commutativity_asso_bowtie}).

\begin{remark}
From \cref{eq:chirotope_tree}, it follows immediately that the sign function $\chi_T$
associated to a chirotope tree $T$ does not depend on the labels of its proxies.
Accordingly, in the following, we consider chirotope trees up to relabeling of their proxies.
\end{remark}

\subsection{Two Operations on chirotope trees}
\label{ssec:operation-trees}

We now adapt two standard operations on trees (edge contraction
and vertex split) for chirotope trees in a way that leaves the
associated chirotopes invariant.

\subsubsection{Contraction}

\begin{definition}\label{def:merge}
  Let $T$ be a chirotope tree and $e=\{v,v'\}$ be an edge in $T$ that
  selects $s^*$ and $t^*$ in $\chi_v$ and $\chi_{v'}$, respectively. The
  \emph{contraction of $e$ in $T$} is the tree $T'$ obtained from $T$
  by merging the nodes $v$ and $v'$ into a new node $v_{new}$,
  decorated with the chirotope $\chi_{new}\eqdef \chi_v \join{s^*}{t^*}\chi_{v'}$.  We denote this transformation by $T \fusion{e} T'$.
\end{definition}

Note that, by
\cref{prop:Substitution_is_chirotope,prop:Realization_Substitution},
$\chi_{new}$ is a chirotope, and is realizable if both $\chi_v$ and
$\chi_{v'}$ are realizable.  Furthermore, by
\cref{prop:extreme_stay_extreme}, the extreme elements of $\chi_v$ and
$\chi_{v'}$ different from $s^*$ and $t^*$ are extreme elements of
$\chi_{new}$ so that the proxy elements corresponding to other edges
of $T$ are still extreme in their respective chirotopes in $T'$; it
follows that $T'$ is indeed a chirotope tree.

\begin{example}
  Here is an example of a contraction of an edge in a chirotope tree (with $\chi = \kappa \join{s^*}{t^*}
  \xi$).

  \medskip
  \begin{center}
    \begin{tikzpicture}[baseline={(0,0)},scale=0.8, every node/.style={transform shape}]

      \node[circle,draw,minimum size=1cm,very thin] (n2) at (0,0) {$\kappa$};
      \node[circle,draw,minimum size=1cm,very thin] (n2bis) at (1.8,0) {$\xi$};
      \node[circle,draw,very thin,minimum size=1cm] (n3) at ($(n2)+(150:1.8)$) {$\chi_3$};
      \node[circle,draw,minimum size=1cm,very thin] (n1) at ($(n2)+(-150:1.8)$) {$\chi_1$};
      \node[circle,draw,minimum size=1cm,very thin] (n4) at ($(n2bis)+(0:1.8)$) {$\chi_2$};

      \node[red] (x1*) at ($(n1)+(30:.35)$) {$\bullet$};
      \node[red] (x3*) at ($(n3)+(-30:0.35)$) {$\bullet$};
      \node[red] (x4*) at ($(n4)+(180:0.35)$) {$\bullet$};
      \node[red] (y2*) at ($(n2)+(180-30:0.35)$) {$\bullet$};
      \node[red] (z2*) at ($(n2bis)+(0:0.35)$) {$\bullet$};
      \node[red] (x2*) at ($(n2)+(180+30:0.35)$) {$\bullet$};
      \node[red] (s*) at ($(n2)+(0:0.35)$) {$\bullet$};
      \node[red] (t*) at ($(n2bis)+(180:0.35)$) {$\bullet$};

      \node[red,above=-0.23 of x1*,font=\footnotesize] {$x_1^*$};
      \node[red,below=-0.23 of x2*,font=\footnotesize] {$x^*$};
      \node[red,above=-0.23 of y2*,font=\footnotesize] {$y^*$};
      \node[red,above=-0.23 of z2*,font=\footnotesize] {$z^*$};
      \node[red,below=-0.23 of x3*,font=\footnotesize] {$x_3^*$};
      \node[red,above=-0.23 of x4*,font=\footnotesize] {$x_2^*$};
      \node[red,above right=-0.3 of s*,font=\footnotesize] {$s^*$};
      \node[red,above left=-0.4 of t*,font=\footnotesize] {$t^*$};

      \draw[red] (x3*.center)--(y2*.center);
      \draw[red] (x2*.center)--(x1*.center);
      \draw[red] (x4*.center)--(z2*.center);
      \draw[red] (s*.center)--(t*.center);
    \end{tikzpicture}
    \qquad $\xrightarrow[]{(s^*,t^*)}$\qquad
    \begin{tikzpicture}[baseline={(0,0)},scale=0.8, every node/.style={transform shape}]
      \node[circle,draw,minimum size=1cm,very thin] (n2) at (0,0) {$\chi$};
      \node[circle,draw,minimum size=1cm,very thin] (n1) at ($(n2)+(-90:1.8)$) {$\chi_1$};
      \node[red] (y1*) at ($(n1)+(0,.35)$) {$\bullet$};

      \node[circle,draw,minimum size=1cm,very thin] (n3) at ($(n2)+(150:1.8)$) {$\chi_3$};
      \node[red] (x3*) at ($(n3)+(-30:0.35)$) {$\bullet$};

      \node[circle,draw,minimum size=1cm,very thin] (n4) at ($(n2)+(30:1.8)$) {$\chi_2$};
      \node[red] (x4*) at ($(n4)+(180+30:0.35)$) {$\bullet$};

      \node[red] (y2*) at ($(n2)+(180-30:0.35)$) {$\bullet$};
      \node[red] (z2*) at ($(n2)+(30:0.35)$) {$\bullet$};
      \node[red] (x2*) at ($(n2)+(-90:0.35)$) {$\bullet$};

      \node[red,above right=-0.3 of y1*,font=\footnotesize] {$x_1^*$};
      \node[red,below left=-0.3 of x2*,font=\footnotesize] {$x^*$};
      \node[red,above=-0. of y2*,font=\footnotesize] {$y^*$};
      \node[red,above=-0. of z2*,font=\footnotesize] {$z^*$};
      \node[red,below=-0.23 of x3*,font=\footnotesize] {$x_3^*$};
      \node[red,below=-0.23 of x4*,font=\footnotesize] {$x_2^*$};

      \draw[red] (x3*.center)--(y2*.center);
      \draw[red] (x2*.center)--(y1*.center);
      \draw[red] (x4*.center)--(z2*.center);
    \end{tikzpicture}
  \end{center}
\end{example}

\subsubsection{Split}

\begin{definition}\label{def:split}
  Let $T$ be a chirotope tree with a node $v_0$ whose decoration
  $\chi_{v_0} = \chi$ has a nontrivial bowtie decomposition $\chi=
  \kappa\join{s^*}{t^*}\xi$. The {\em split of $T$ according to
    $\chi= \kappa\join{s^*}{t^*}\xi$} is the tree $T'$ obtained
  from $T$ by replacing $v_0$ by two nodes $v_1$ and $v_2$ that are
  decorated by $\kappa$ and $\xi$, respectively, and form an edge
  selecting $s^*$ in $v_1$ and $t^*$ in $v_2$. Denote this
  transformation by
  $T\xrightarrow[\kappa\join{s^*}{t^*}\xi]{\chi}T'$.
\end{definition}

\noindent
Note that for any edge $e=\{s^e_{v_0},s^e_w\}$ connecting $v_0$ to
another vertex $w$ in $T$, the element $s^e_{v_0}$ belongs to the
ground set of either $\kappa$ or $\xi$.  Therefore, $e$ is naturally
seen as an edge of $T'$, connecting either $v_1$ or $v_2$ to $w$,
depending on the set in which $s^e_{v_0}$ lies.\smallskip

\begin{example}
  Here is an example of a split of a tree according to the bowtie
  decomposition $\chi=\kappa\join{s^*}{t^*}\xi$.

  \medskip
  \begin{center}
    \begin{tikzpicture}[baseline={(0,0)},scale=0.8, every node/.style={transform shape}]
      \node[circle,draw,minimum size=1cm,very thin] (n2) at (0,0) {$\chi$};

      \node[circle,draw,minimum size=1cm,very thin] (n1) at ($(n2)+(-90:1.8)$) {$\chi_1$};
      \node[red] (y1*) at ($(n1)+(0,.35)$) {$\bullet$};

      \node[circle,draw,minimum size=1cm,very thin] (n3) at ($(n2)+(150:1.8)$) {$\chi_3$};
      \node[red] (x3*) at ($(n3)+(-30:0.35)$) {$\bullet$};

      \node[circle,draw,minimum size=1cm,very thin] (n4) at ($(n2)+(30:1.8)$) {$\chi_2$};
      \node[red] (x4*) at ($(n4)+(180+30:0.35)$) {$\bullet$};

      \node[red] (y2*) at ($(n2)+(180-30:0.35)$) {$\bullet$};
      \node[red] (z2*) at ($(n2)+(30:0.35)$) {$\bullet$};
      \node[red] (x2*) at ($(n2)+(-90:0.35)$) {$\bullet$};

      \node[red,above right=-0.3 of y1*,font=\footnotesize] {$x_1^*$};
      \node[red,below left=-0.3 of x2*,font=\footnotesize] {$x^*$};
      \node[red,above=-0. of y2*,font=\footnotesize] {$y^*$};
      \node[red,above=-0. of z2*,font=\footnotesize] {$z^*$};
      \node[red,below=-0.23 of x3*,font=\footnotesize] {$x_3^*$};
      \node[red,below=-0.23 of x4*,font=\footnotesize] {$x_2^*$};

      \draw[red] (x3*.center)--(y2*.center);
      \draw[red] (x2*.center)--(y1*.center);
      \draw[red] (x4*.center)--(z2*.center);
    \end{tikzpicture}\qquad $\xrightarrow[\kappa\join{s^*}{t^*}\xi]{\chi}$\qquad
    \begin{tikzpicture}[baseline={(0,0)},scale=0.8, every node/.style={transform shape}]
      \node[circle,draw,minimum size=1cm,very thin] (n2) at (0,0) {$\kappa$};
      \node[circle,draw,minimum size=1cm,very thin] (n2bis) at (1.8,0) {$\xi$};
      \node[circle,draw,very thin,minimum size=1cm] (n3) at ($(n2)+(150:1.8)$) {$\chi_3$};
      \node[circle,draw,minimum size=1cm,very thin] (n1) at ($(n2)+(-150:1.8)$) {$\chi_1$};
      \node[circle,draw,minimum size=1cm,very thin] (n4) at ($(n2bis)+(0:1.8)$) {$\chi_2$};

      \node[red] (x1*) at ($(n1)+(30:.35)$) {$\bullet$};
      \node[red] (x3*) at ($(n3)+(-30:0.35)$) {$\bullet$};
      \node[red] (x4*) at ($(n4)+(180:0.35)$) {$\bullet$};
      \node[red] (y2*) at ($(n2)+(180-30:0.35)$) {$\bullet$};
      \node[red] (z2*) at ($(n2bis)+(0:0.35)$) {$\bullet$};
      \node[red] (x2*) at ($(n2)+(180+30:0.35)$) {$\bullet$};
      \node[red] (s*) at ($(n2)+(0:0.35)$) {$\bullet$};
      \node[red] (t*) at ($(n2bis)+(180:0.35)$) {$\bullet$};

      \node[red,above=-0.23 of x1*,font=\footnotesize] {$x_1^*$};
      \node[red,below=-0.23 of x2*,font=\footnotesize] {$x^*$};
      \node[red,above=-0.23 of y2*,font=\footnotesize] {$y^*$};
      \node[red,above=-0.23 of z2*,font=\footnotesize] {$z^*$};
      \node[red,below=-0.23 of x3*,font=\footnotesize] {$x_3^*$};
      \node[red,above=-0.23 of x4*,font=\footnotesize] {$x_2^*$};
      \node[red,above right=-0.3 of s*,font=\footnotesize] {$s^*$};
      \node[red,above left=-0.4 of t*,font=\footnotesize] {$t^*$};

      \draw[red] (x3*.center)--(y2*.center);
      \draw[red] (x2*.center)--(x1*.center);
      \draw[red] (x4*.center)--(z2*.center);
      \draw[red] (s*.center)--(t*.center);
  \end{tikzpicture}\end{center}  
\end{example}

\subsubsection{Properties}

The contraction and split decomposition operations are inverse of one another in the sense that, with the notation of \cref{def:split,def:merge},
\begin{equation}\label{eq:contraction-split-inverse}
  T\xrightarrow[\kappa\join{s^*}{t^*}\xi]{\chi}T' \quad \Rightarrow \quad T'\xrightarrow{(s^*,t^*)}T, \qquad \text{and}\qquad T\xrightarrow{(s^*,t^*)}T' \quad \Rightarrow \quad T'\xrightarrow[\chi_v \join{s^*}{t^*} \chi_{v'}]{\chi_{new}}T.
\end{equation}

\noindent
We next show that performing an edge contraction or a vertex split on a chirotope tree does not change its associated chirotope.

\begin{proposition}
  \label{lem:merge_sign_function}
  If $T \fusion{e} T'$ or $T\xrightarrow[\kappa\join{s^*}{t^*}\xi]{\chi}T'$,  then $\chi_T = \chi_{T'}$.
\end{proposition}

\begin{proof}
We first prove the statement under the hypothesis that $T \fusion{e}
T'$.  We proceed by induction on the number of nodes in $T$. If $T$
has two nodes, then the result follows from \cref{lem:chiT_is_join_2}.
Assume the result holds for all chirotope trees with $n \ge 2$ nodes,
and let $T$ be a chirotope trees with $n+1$ nodes, let $e$ be an edge
of $T$, and and let $T'$ be such that $T \fusion{e} T'$. Since $T$ has
at least three nodes, there is an edge $f \neq e$ in $T$, and that
edge also exist in $T'$. Removing $f$ (from $T$) splits $T$ into two
parts $T_1$ and $T_2$, and removing $f$ (from $T'$) splits $T'$ into
two parts $T'_1$ and $T'_2$; without loss of generality, we assume
that $e$ belongs to $T_1$ and $T'_1$, so that $T'_2=T_2$ and $T_1
\fusion{e} T'_1$. By the induction hypothesis, $\chi_{T_1} =
\chi_{T'_1}$. Also applying \cref{lem:chiT_is_join_2} twice yields
\[\chi_T= \chi_{T_1}  \join{s_1}{s_2}\chi_{T_2}, \qquad \chi_{T'}= \chi_{T'_1}  \join{s_1}{s_2}\chi_{T'_2},\]
where $f=\{s_1,s_2\}$. Altogether this implies $\chi_T = \chi_{T'}$, and the lemma is proved under the hypothesis that $T \fusion{e} T'$.

If $T\xrightarrow[\kappa\join{s^*}{t^*}\xi]{\chi}T'$, then \cref{eq:contraction-split-inverse} and the first statement proved in this proposition also ensures that $\chi_T = \chi_{T'}$, concluding the proof.
\end{proof}

\subsection{Canonical chirotope trees}\label{ssec:uniqueness_canonical}

We recall from the introduction that a chirotope tree is called \emph{canonical}
 if every node is decorated by a
  convex or indecomposable chirotope, and if no edge connects two
  nodes decorated by convex chirotopes.
The goal of the section is to prove \cref{thm:uniqueness_canonical},
which we copy here for convenience.\medskip

\noindent \textbf{\cref{thm:uniqueness_canonical}.}
  For every chirotope $\chi$, there exists a unique canonical
  chirotope tree $T$ such that $\chi_T = \chi$.\medskip

\noindent
Suppose that we start with the trivial chirotope tree with a single
node decorated with $\chi$, $T^{tr}_{\chi}$, and, for as long as we
can and in any order, we split the nodes decorated by decomposable
nonconvex chirotopes, and contract the edges joining nodes decorated
by convex chirotopes. The gist of the proof of
\cref{thm:uniqueness_canonical} is to prove that (i) any such
sequence of operations terminates, (ii) any canonical tree $T$ with
$\chi_T=\chi$ can be obtained by such a sequence of operations, and
(iii) all sequences of operations produce the same canonical tree. A
convenient setting to carry out this analysis is the theory of rewriting
systems.

\medskip

\subsubsection{The rewriting system}
Using the contraction and split operations, we define a rewriting system on
chirotope trees.

\begin{definition}\label{def:rwrule}
  We consider the set of chirotope trees equipped with the following rewriting rules:
  \begin{itemize}
  \item $T\xrightarrow{\lozenge} T'$ if $T'$ is obtained from $T$ by contracting an edge of $T$ between two \emph{convex} chirotopes,
  \item $T\xrightarrow{\pap} T'$ if $T'$ is obtained from $T$ by splitting a \emph{nonconvex} node of $T$,
  \end{itemize}
We write $T\Rightarrow T'$ if $T\xrightarrow{\lozenge} T'$ or $T\xrightarrow{\pap} T'$.
\end{definition}

\noindent
As usual, for any rewriting rule $\to$, we denote by $\to^*$ the
rewriting rule which consists in any number (possibly zero) of
successive applications of $\to$. By definition, a chirotope tree is
canonical if the rewriting rule $\Rightarrow$ does not apply to it.

\subsubsection{Termination}\label{ssec:existence_canonical}
The existence of a canonical chirotope tree representing any given
chirotope $\chi$ follows from the termination property of
$\Rightarrow$. The proof of that property uses the notion of multiset
ordering.  
Given two {finite} multisets $M$ and $N$ of positive
integers, we write $M >_{mul} N$ if there exist multisets $X \subseteq
M$ and $Y$ such that $N= (M \setminus X) \uplus Y$ and
$\max(X)>\max(Y)$.  It is standard that  $>_{mul}$ defines a
well-founded order (i.e.~without infinite decreasing subsequences) on
mutisets of $\mathbb N$.

\begin{lemma}\label{lem:terminates}
  Every rewriting sequence $T_0 \Rightarrow T_1 \Rightarrow \ldots  \Rightarrow T_n \Rightarrow \ldots$ is finite.
\end{lemma}
\begin{proof}
   Let us associate with a chirotope tree $T$ a pair $(N(T),C(T))$,
   where  $N(T)$ is the multiset of  sizes of nonconvex chirotopes in $T$,
   and $C(T)$ is the number of convex chirotopes in $T$.
   Suppose that $T\Rightarrow T'$.
   If $T\xrightarrow{\lozenge} T'$, then $C(T')=C(T)-1$ and $N(T')=N(T)$; indeed,
   by \cref{lem:bowtie_of_convex}, we replace two nodes decorated by convex chirotopes, by a single one,
   also decorated by a convex chirotope.  If $T\xrightarrow{\pap} T'$, the split operation replaces one node with a nonconvex decoration
   by two new nodes with decorations of smaller sizes. Thus we have $N(T) >_{mul} N(T')$ (in this case, $C(T)$
   can change as one of the new node can have a convex decoration).
   In both cases, we have $(N(T'),C(T'))<_{lex} (N(T),C(T))$,
   where $<_{lex}$ denotes the lexicographic order corresponding to the pair of orders $(<_{mul}, <_\mathbb{N})$.
   This lexicographic order is well-founded (since both original orders are well-founded),
   implying that applying the rewriting rule $\Rightarrow$ always terminates.
\end{proof}

\subsubsection{Accessibility} The next step is to show that the rewriting rule can turn $T^{tr}_\chi$ into any canonical chirotope tree whose associated chirotope is $\chi$.

\begin{lemma}
  Let $T$ be a canonical chirotope tree with $\chi_T=\chi$. Then $T^{tr}_\chi \, \Rightarrow^* \, T$.
  \label{lem:T_accessible}
\end{lemma}
\begin{proof}
  We start from $T$ and contract edges one by one, until we have a
  tree with a single node.  Since contractions do not modify the
  associated chirotope, we end up with
  $T^{tr}_\chi$. \cref{lem:bowtie_of_convex} asserts that the node
  obtained by an edge contraction is decorated by a convex chirotope
  (if and) only if the two merged nodes are decorated by convex
  chirotopes. Since the tree we start with is canonical, an immediate
  induction yields that none of the trees appearing in this sequence
  of contractions has an edge joining two nodes decorated by convex
  chirotopes. Each contraction we performed is therefore the inverse of a
  split operation allowed in $\Rightarrow$, proving the claim.
\end{proof}

\subsubsection{Confluence}

It remains to show that, starting from $T^{tr}_\chi$, independently of
the choice of rewriting operations, iterating $\Rightarrow$ as much as
possible always yields the same output.  This property is known as
{\em confluence} in the theory of rewriting rules.  It is well-known
(Newman's lemma) that, for a terminating rewriting rule (like $\Rightarrow$, see \cref{lem:terminates}), confluence is implied by the following {\em
  local confluence} property.

\begin{proposition}%
  \label{prop:local_confluence}
  If $T,T_1,T_2$ are three chirotope trees such that $T \Rightarrow T_1$ and $T\Rightarrow T_2$, then there exists a chirotope tree $T'$ such that  $T_1\Rightarrow^* T'$ and $T_2\Rightarrow^* T'$.
\end{proposition}

\begin{proof}
  The statement is immediate if the transformations $T \Rightarrow
  T_1$ and $T\Rightarrow T_2$ operate on disjoint sets of nodes:
  indeed, each operation can be performed after the other one, and the
  resulting chirotope tree does not depend on the order of
  operations. If the transformations operate on nondisjoint sets of
  nodes then they must be of the same type as $\xrightarrow{\lozenge}$
  operates on two nodes decorated by convex chirotopes whereas any of
  the nodes on which $\xrightarrow{\pap}$ operates is decorated by a
  nonconvex chirotope.  If $T\xrightarrow{\lozenge}T_1$ and
  $T\xrightarrow{\lozenge}T_2$, the local confluence follows from the
  associativity of $\pap$. So let us consider we are in the last case, with
  \[ T\xrightarrow[\kappa\join{y^*}{z^*}\xi]{\chi}T_1 \qquad \text{and} \qquad T\xrightarrow[\alpha\join{a^*}{b^*}\beta]{\chi}T_2,\]
  and $\chi$ a nonconvex chirotope decorating a node $v$ of $T$. The
  analysis turns out to be considerably more intricate.

  \medskip

  Let $G\eqdef X\uplus S$ denote the ground set of $\chi$, with $X$ denoting the nonproxy elements of the ground set, and $S$ its proxy elements.
 Throughout the analysis, we only make modifications in the tree that are local to the node labeled by $\chi$.
 Accordingly we will only represent in pictures the node decorated with $\chi$ and the nodes replacing it after various operations.
 For example $T$, $T_1$ and $T_2$ are visualized as follows:

 \begin{center}\includegraphics{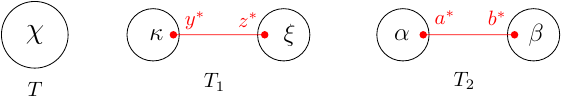} .\end{center}

 \noindent
 Let $\{Y,Z\}$ be the modular bipartition of $X\cup S$ associated to
 the bowtie decomposition $\chi=\kappa\join{y^*}{z^*}\xi$, and
 $\{A,B\}$ be the modular bipartition of $X\cup S$ associated with
 $\chi=\alpha\join{a^*}{b^*}\beta$. We assume that $\{A,B\} \neq
 \{Y,Z\}$ as otherwise $T_1=T_2$ and the statement trivially
 follows. Since $\chi$ is nonconvex, \cref{lem:bowtie_of_convex}
 ensures that $\alpha$ or $\beta$ is nonconvex, and the same goes for
 $\kappa$ or $\xi$.

 \bigskip

 Let us first asssume that $A\subsetneq Y$, and therefore $Z\subsetneq
 B$.  If $\kappa$ is convex, then $\alpha$ is convex too, hence
 $\beta$ is not convex.  Therefore, up to exchanging the role of
 $(A,B)$ and $(Y,Z)$, we can assume that $A\subsetneq Y$ and that
 $\kappa$ is not convex.

 From $A\subsetneq Y$, it comes that the pair $\{A,\{y^*\}\cup
 Y\setminus A\}$ is then a non trivial modular bipartition of
 $\kappa$, as $\chi$ and $\kappa$ coincide on $Y$, and $y^*$ is a
 proxy representing the elements of $Z$.
 Let $T_3$ be the split of $T_1$ according to the corresponding bowtie
 decomposition $\kappa=\kappa_1\join{c^*}{d^*}\kappa_2$.
 Since $\kappa$ is not convex, we have $T_1 \Rightarrow T_3$. The tree
 $T_3$ is visualized as follows:

 \[\begin{array}{c}\includegraphics{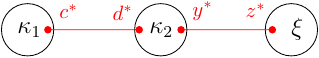}\end{array}.\]

 \noindent
 By associativity (\cref{prop:commutativity_asso_bowtie}), we have
 $\chi=(\kappa_1\join{c^*}{d^*}\kappa_2)\join{y^*}{z^*}\xi
 =\kappa_1\join{c^*}{d^*}(\kappa_2\join{y^*}{z^*}\xi)$.  This means
 that $\kappa_1$ (resp.~$\kappa_2\join{y^*}{z^*}\xi$) is the chirotope
 associated to the module $A$ (resp.~$B$) in the corresponding bowtie
 decomposition.  Hence, we have $\kappa_1=\alpha$ and
 $(\kappa_2\join{y^*}{z^*}\xi)=\beta$ (up to renaming the proxy
 elements $c^*$ in $a^*$ and $d^*$ in $b^*$).  This implies
 \[T_2 \xrightarrow[\kappa_2\join{y^*}{z^*}\xi]{\beta} T_3.\]
 If $\beta$ is not convex, then $T_2 \Rightarrow T_3$ and the
 statement follows with $T'=T_3$. If $\beta$ is convex, then both
 $\kappa_2$ and $\xi$ are, and $T_3 \Rightarrow T_2$, proving the
 statement with $T'=T_2$.

 \medskip

A schematic representation of the argument is presented on \Cref{fig:local_confluence_cas_simple}.

 \begin{figure}[h]
 \includegraphics{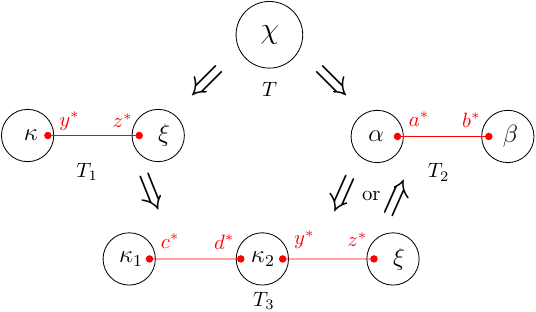}
 	\caption{Diamond diagram proving local confluence in the case $A \subsetneq Y$.}
 	\label{fig:local_confluence_cas_simple}
 \end{figure}

 \bigskip

 The cases $A\subsetneq Z$, $B \subsetneq Y$, $B \subsetneq Z$ are
 similar to the case $A\subsetneq Y$, so from now on we can assume
 that $A\cap Y$, $A\cap Z$, $B\cap Y$, $B\cap Z$ are all nonempty.
 \cref{lm:intersection_modules} then implies that each of $A\cap Y$,
 $A\cap Z$, $B\cap Y$ and $B\cap Z$ is a nontrivial module of
 $\chi$. We use this to further decompose $T_1$ and $T_2$ in ways that
 eventually reduce to the same term. The specific decomposition
 depends on which of $A\cap Y$, $A\cap Z$, $B\cap Y$, $B\cap Z$
 contain nonextremal elements. We spell out two representative cases
 and leave the remaining, routine, check to the reader.

 \bigskip

\paragraph{{\it Case 1: each of $A\cap Y$, $A\cap Z$, $B\cap Y$, $B\cap Z$ contains
   at least one nonextremal element.}}

Many different chirotopes now enter the stage, so we record their
ground sets and proxies in tables to help the reader keep track of
them.

\bigskip

 \noindent
 \begin{minipage}{7.25cm}
   The ground set of $\kappa$ is $Y \cup \{y^*\}= (A\cap Y) \cup (B
   \cap Y) \cup \{y^*\}$ and $\chi$ and $\kappa$ coincide on $Y$. It
   follows that $A\cap Y$ and $B \cap Y$ are disjoint nontrivial
   modules of $\kappa$, and $\kappa$ therefore decomposes into $\kappa = \sigma
   \join{s^*}{x_{AY}} \kappa^\triangle \join{x_{BY}}{p^*} \pi$.
 \end{minipage}
 \hfill
 \begin{minipage}{7.25cm}
   \hfill
       {\begin{tabular}{|c|c|c|}
     \hline
     Chirotope & Ground set & Proxies \\
     \hline
     $\kappa$ & $ Y \cup \{y^*\}$ & $y^*$\\
     \hline
     $\sigma$ & $(A\cap Y) \cup \{s^*\}$ & $s^*$\\
     \hline
     $\pi$ & $(B\cap Y) \cup \{p^*\}$ & $p^*$\\
     \hline
     $\kappa^{\triangle}$ & $ \{x_{AY}, x_{BY},y^*\}$ & $x_{AY}, x_{BY},y^*$\\
     \hline
   \end{tabular}}
 \end{minipage}

 \bigskip

 \noindent
 \begin{minipage}{7.25cm}
   {
     \begin{tabular}{|c|c|c|}
 \hline
Chirotope & Ground set & Proxies \\
 \hline
$\xi$ & $ Z \cup \{z^*\}$ & $z^*$\\
 \hline
$\tau$ & $(B\cap Z) \cup \{t^*\}$ & $t^*$\\
 \hline
$\rho$ & $ (A\cap Z) \cup \{r^*\}$ & $r^*$\\
 \hline
$\xi^{\triangle}$ & $ \{x_{AZ}, x_{BZ},z^*\}$ & $x_{AZ}, x_{BZ},z^*$\\
 \hline
   \end{tabular}}
 \end{minipage}
 \hfill
 \begin{minipage}{7.25cm}
   Similarly, the ground set of $\xi$ is $Z \cup \{z^*\}= (A\cap Z)
   \cup (B \cap Z) \cup \{z^*\}$ and $\chi$ and $\xi$ coincide on
   $Z$. It follows that $A\cap Z$ and $B \cap Z$ are disjoint
   nontrivial modules of $\xi$, and $\xi$ therefore decomposes into
   $\xi = \tau \join{t^*}{x_{BZ}} \xi^\triangle \join{x_{AZ}}{r^*}
   \rho$.
 \end{minipage}

 \bigskip

 \noindent
 Altogether, we can transform $T_1$ into

\[
T_1 \quad = \qquad
\begin{tikzpicture}[baseline={(0,-0.1)},scale=0.65, every node/.style={transform shape}]
\node[circle,draw,minimum size=  1cm, very thin] (kappa) at (0,0) {$\kappa$};
\node[circle,draw,minimum size=  1cm, very thin] (xi) at (2,0) {$\xi$};
\node[red] (y*) at ($(kappa)+(0:0.35)$) {$\bullet$};
\node[red] (z*) at ($(xi)+(180:0.35)$) {$\bullet$};
\node[red,font=\footnotesize] at ($(y*)+(.35,0.2)$) {$y^*$};
\node[red,font=\footnotesize] at ($(z*)-(.25,-0.22)$) {$z^*$};
\draw[red, thin] (y*.center)--(z*.center);
\end{tikzpicture}
\qquad \xrightarrow{\pap}^* \qquad
\begin{tikzpicture}[baseline={(0,-0.1)},scale=0.65, every node/.style={transform shape}]
\node[circle,draw,minimum size=  1.2cm, very thin] (kappaT) at (0,0) {$\kappa^\triangle$};
\node[circle,draw,minimum size=  1.2cm, very thin] (xiT) at (2,0) {$\xi^\triangle$};
\node[circle,draw,minimum size=  1cm, very thin] (sigma) at ($(kappaT)+(180-45:1.8)$) {$\sigma$};
\node[circle,draw,minimum size=  1cm, very thin] (pi) at ($(kappaT)+(180+45:1.8)$) {$\pi$};
\node[circle,draw,minimum size=  1cm, very thin] (rho) at ($(xiT)+(45:1.8)$) {$\rho$};
\node[circle,draw,minimum size=  1cm, very thin] (tau) at ($(xiT)+(-45:1.8)$) {$\tau$};

\node[red] (s*) at ($(sigma)+(-45:0.35)$) {$\bullet$};
\node[red] (p*) at ($(pi)+(45:0.35)$) {$\bullet$};
\node[red] (r*) at ($(rho)+(180+45:0.35)$) {$\bullet$};
\node[red] (t*) at ($(tau)+(180-45:0.35)$) {$\bullet$};

\node[red] (y*) at ($(kappaT)+(0:0.45)$) {$\bullet$};
\node[red] (z*) at ($(xiT)+(180:0.45)$) {$\bullet$};
\node[red] (xAY) at ($(kappaT)+(180-45:0.45)$) {$\bullet$};
\node[red] (xBY) at ($(kappaT)+(180+45:0.45)$) {$\bullet$};
\node[red] (xAZ) at ($(xiT)+(45:0.45)$) {$\bullet$};
\node[red] (xBZ) at ($(xiT)+(-45:0.45)$) {$\bullet$};

\node[red,font=\footnotesize] at ($(y*)+(.35,0.2)$) {$y^*$};
\node[red,font=\footnotesize] at ($(z*)-(.25,-0.22)$) {$z^*$};
\node[red,font=\footnotesize] at ($(p*)+(10:0.45)$) {$p^*$};
\node[red,font=\footnotesize] at ($(t*)+(180:0.35)$) {$t^*$};
\node[red,font=\footnotesize] at ($(r*)+(180:0.4)$) {$r^*$};
\node[red,font=\footnotesize] at ($(s*)+(0:0.4)$) {$s^*$};
\node[red,font=\footnotesize] at ($(xAY)+(170:0.55)$) {$x_{AY}$};
\node[red,font=\footnotesize] at ($(xAZ)+(10:0.55)$) {$x_{AZ}$};
\node[red,font=\footnotesize] at ($(xBZ)+(-10:0.55)$) {$x_{BZ}$};
\node[red,font=\footnotesize] at ($(xBY)+(190:0.55)$) {$x_{BY}$};

\draw[red, thin] (y*.center)--(z*.center);
\draw[red, thin] (r*.center)--(xAZ.center);
\draw[red, thin] (s*.center)--(xAY.center);
\draw[red, thin] (t*.center)--(xBZ.center);
\draw[red, thin] (p*.center)--(xBY.center);
\end{tikzpicture}
\qquad  \xrightarrow{\lozenge}  \qquad
\begin{tikzpicture}[baseline={(0,0)},scale=0.65, every node/.style={transform shape}]
\node[circle,draw,minimum size=  1.2cm, very thin] (chiS) at (0,0) {$\chi^\square$};
\node[circle,draw,minimum size=  1cm, very thin] (sigma) at ($(chiS)+(180-45:1.8)$) {$\sigma$};
\node[circle,draw,minimum size=  1cm, very thin] (pi) at ($(chiS)+(180+45:1.8)$) {$\pi$};
\node[circle,draw,minimum size=  1cm, very thin] (rho) at ($(chiS)+(45:1.8)$) {$\rho$};
\node[circle,draw,minimum size=  1cm, very thin] (tau) at ($(chiS)+(-45:1.8)$) {$\tau$};

\node[red] (s*) at ($(sigma)+(-45:0.35)$) {$\bullet$};
\node[red] (p*) at ($(pi)+(45:0.35)$) {$\bullet$};
\node[red] (r*) at ($(rho)+(180+45:0.35)$) {$\bullet$};
\node[red] (t*) at ($(tau)+(180-45:0.35)$) {$\bullet$};

\node[red] (xAY) at ($(chiS)+(180-45:0.45)$) {$\bullet$};
\node[red] (xBY) at ($(chiS)+(180+45:0.45)$) {$\bullet$};
\node[red] (xAZ) at ($(chiS)+(45:0.45)$) {$\bullet$};
\node[red] (xBZ) at ($(chiS)+(-45:0.45)$) {$\bullet$};

\node[red,font=\footnotesize] at ($(p*)+(10:0.45)$) {$p^*$};
\node[red,font=\footnotesize] at ($(t*)+(180:0.35)$) {$t^*$};
\node[red,font=\footnotesize] at ($(r*)+(180:0.4)$) {$r^*$};
\node[red,font=\footnotesize] at ($(s*)+(0:0.4)$) {$s^*$};
\node[red,font=\footnotesize] at ($(xAY)+(170:0.55)$) {$x_{AY}$};
\node[red,font=\footnotesize] at ($(xAZ)+(10:0.55)$) {$x_{AZ}$};
\node[red,font=\footnotesize] at ($(xBZ)+(-10:0.55)$) {$x_{BZ}$};
\node[red,font=\footnotesize] at ($(xBY)+(190:0.55)$) {$x_{BY}$};

\draw[red, thin] (r*.center)--(xAZ.center);
\draw[red, thin] (s*.center)--(xAY.center);
\draw[red, thin] (t*.center)--(xBZ.center);
\draw[red, thin] (p*.center)--(xBY.center);
\end{tikzpicture}
\]

\noindent
where the last transformation is a contraction of the central edge
between the two convex chirotopes $\kappa^\Delta$ and $\xi^\Delta$;
the resulting node is decorated with a convex chirotopes on four
elements, all proxies; we denote this chirotope by
$\chi^{\square}$, and its proxy elements are $x_{AY}, x_{AZ}, x_{BZ}, x_{BY}$. Similarly, we have
\[
T_2 \quad =\qquad \begin{tikzpicture}[baseline={(0,-0.1)},scale=0.65, every node/.style={transform shape}]
\node[circle,draw,minimum size=  1cm, very thin] (kappa) at (0,0) {$\alpha$};
\node[circle,draw,minimum size=  1cm, very thin] (xi) at (2,0) {$\beta$};
\node[red] (y*) at ($(kappa)+(0:0.35)$) {$\bullet$};
\node[red] (z*) at ($(xi)+(180:0.35)$) {$\bullet$};
\node[red,font=\footnotesize] at ($(y*)+(.35,0.2)$) {$a^*$};
\node[red,font=\footnotesize] at ($(z*)-(.25,-0.22)$) {$b^*$};
\draw[red, thin] (y*.center)--(z*.center);
\end{tikzpicture}
\quad \xrightarrow{\pap}^* \quad
\begin{tikzpicture}[baseline={(0,-0.1)},scale=0.65, every node/.style={transform shape}]
\node[circle,draw,minimum size=  1.2cm, very thin] (kappaT) at (0,0) {$\alpha^\triangle$};
\node[circle,draw,minimum size=  1.2cm, very thin] (xiT) at (2,0) {$\beta^\triangle$};
\node[circle,draw,minimum size=  1cm, very thin] (sigma) at ($(kappaT)+(180-45:1.8)$) {$\rho$};
\node[circle,draw,minimum size=  1cm, very thin] (pi) at ($(kappaT)+(180+45:1.8)$) {$\sigma$};
\node[circle,draw,minimum size=  1cm, very thin] (rho) at ($(xiT)+(45:1.8)$) {$\tau$};
\node[circle,draw,minimum size=  1cm, very thin] (tau) at ($(xiT)+(-45:1.8)$) {$\pi$};

\node[red] (s*) at ($(sigma)+(-45:0.35)$) {$\bullet$};
\node[red] (r*) at ($(pi)+(45:0.35)$) {$\bullet$};
\node[red] (p*) at ($(rho)+(180+45:0.35)$) {$\bullet$};
\node[red] (t*) at ($(tau)+(180-45:0.35)$) {$\bullet$};

\node[red] (a*) at ($(kappaT)+(0:0.45)$) {$\bullet$};
\node[red] (b*) at ($(xiT)+(180:0.45)$) {$\bullet$};
\node[red] (xAY) at ($(kappaT)+(180-45:0.45)$) {$\bullet$};
\node[red] (xAZ) at ($(kappaT)+(180+45:0.45)$) {$\bullet$};
\node[red] (xBY) at ($(xiT)+(45:0.45)$) {$\bullet$};
\node[red] (xBZ) at ($(xiT)+(-45:0.45)$) {$\bullet$};

\node[red,font=\footnotesize] at ($(a*)+(.35,0.2)$) {$a^*$};
\node[red,font=\footnotesize] at ($(b*)-(.25,-0.22)$) {$b^*$};
\node[red,font=\footnotesize] at ($(p*)+(190:0.35)$) {$t^*$};
\node[red,font=\footnotesize] at ($(t*)+(180:0.35)$) {$p^*$};
\node[red,font=\footnotesize] at ($(r*)+(10:0.4)$) {$s^*$};
\node[red,font=\footnotesize] at ($(s*)+(0:0.4)$) {$r^*$};
\node[red,font=\footnotesize] at ($(xAY)+(170:0.55)$) {$x_{AZ}$};
\node[red,font=\footnotesize] at ($(xAZ)+(190:0.55)$) {$x_{AY}$};
\node[red,font=\footnotesize] at ($(xBZ)+(-10:0.55)$) {$x_{BY}$};
\node[red,font=\footnotesize] at ($(xBY)+(10:0.55)$) {$x_{BZ}$};

\draw[red, thin] (a*.center)--(b*.center);
\draw[red, thin] (r*.center)--(xAZ.center);
\draw[red, thin] (s*.center)--(xAY.center);
\draw[red, thin] (t*.center)--(xBZ.center);
\draw[red, thin] (p*.center)--(xBY.center);
\end{tikzpicture}
\qquad \xrightarrow{\lozenge} \qquad
\begin{tikzpicture}[baseline={(0,0)},scale=0.65, every node/.style={transform shape}]
\node[circle,draw,minimum size=  1.2cm, very thin] (chiS) at (0,0) {$\chi^\square$};
\node[circle,draw,minimum size=  1cm, very thin] (sigma) at ($(chiS)+(180-45:1.8)$) {$\rho$};
\node[circle,draw,minimum size=  1cm, very thin] (pi) at ($(chiS)+(180+45:1.8)$) {$\sigma$};
\node[circle,draw,minimum size=  1cm, very thin] (rho) at ($(chiS)+(45:1.8)$) {$\tau$};
\node[circle,draw,minimum size=  1cm, very thin] (tau) at ($(chiS)+(-45:1.8)$) {$\pi$};

\node[red] (s*) at ($(sigma)+(-45:0.35)$) {$\bullet$};
\node[red] (p*) at ($(pi)+(45:0.35)$) {$\bullet$};
\node[red] (r*) at ($(rho)+(180+45:0.35)$) {$\bullet$};
\node[red] (t*) at ($(tau)+(180-45:0.35)$) {$\bullet$};

\node[red] (xAY) at ($(chiS)+(180-45:0.45)$) {$\bullet$};
\node[red] (xBY) at ($(chiS)+(180+45:0.45)$) {$\bullet$};
\node[red] (xAZ) at ($(chiS)+(45:0.45)$) {$\bullet$};
\node[red] (xBZ) at ($(chiS)+(-45:0.45)$) {$\bullet$};

\node[red,font=\footnotesize] at ($(p*)+(10:0.45)$) {$s^*$};
\node[red,font=\footnotesize] at ($(t*)+(180:0.35)$) {$p^*$};
\node[red,font=\footnotesize] at ($(r*)+(180:0.4)$) {$t^*$};
\node[red,font=\footnotesize] at ($(s*)+(0:0.4)$) {$r^*$};
\node[red,font=\footnotesize] at ($(xAY)+(170:0.55)$) {$x_{AZ}$};
\node[red,font=\footnotesize] at ($(xAZ)+(10:0.55)$) {$x_{BZ}$};
\node[red,font=\footnotesize] at ($(xBZ)+(-10:0.55)$) {$x_{BY}$};
\node[red,font=\footnotesize] at ($(xBY)+(190:0.55)$) {$x_{AY}$};

\draw[red, thin] (r*.center)--(xAZ.center);
\draw[red, thin] (s*.center)--(xAY.center);
\draw[red, thin] (t*.center)--(xBZ.center);
\draw[red, thin] (p*.center)--(xBY.center);
\end{tikzpicture},
\]
with the same chirotope $\chi^{\square}$ as before, and all other appearing chirotopes recalled or defined in the tables below.

\noindent
\begin{tabular}{|c|c|c|}
 \hline
Chirotope & Ground set & Proxies \\
 \hline
$\alpha$ & $ A\cup \{a^*\}$ & $a^*$\\
 \hline
$\sigma$ & $(A\cap Y) \cup \{s^*\}$ & $s^*$\\
 \hline
$\rho$ & $ (A\cap Z) \cup \{r^*\}$ & $r^*$\\
 \hline
$\alpha^{\triangle}$ & $ \{x_{AY}, x_{AZ},a^*\}$ & $x_{AY}, x_{AZ},a^*$\\
 \hline
   \end{tabular}
 \hfill
   \begin{tabular}{|c|c|c|}
 \hline
Chirotope & Ground set & Proxies \\
 \hline
$\beta$ & $ B\cup \{b^*\}$ & $b^*$\\
 \hline
$\pi$ & $(B\cap Y) \cup \{p^*\}$ & $p^*$\\
 \hline
$\tau$ & $(B\cap Z) \cup \{t^*\}$ & $t^*$\\
 \hline
$\beta^{\triangle}$ & $ \{x_{BY}, x_{BZ},b^*\}$ & $x_{BY}, x_{BZ},b^*$\\
 \hline
   \end{tabular}

   \medskip

   \noindent
   We stress that the chirotopes $\sigma,\rho,\tau,\pi$ and
   $\chi^{\square}$ are the same in both decompositions, as they are
   restriction of our initial chirotope $\chi$ to the same sets
   (replacing proxy elements by arbitrary elements in the sets that
   they represent). The two trees obtained at the far right of both
   derivations are identical.  Since each of $A\cap Y$, $A\cap Z$,
   $B\cap Y$ and $B\cap Z$ contains a nonextremal point, the
   chirotopes $\sigma,\rho,\tau$ and $\pi$ are non convex, and all
   transformations in the decompositions of $T_1$ and $T_2$ given
   above are valid rules for $\Rightarrow$, and the statement follows.

   \bigskip

 \paragraph{{\it Case 2: each of $A\cap Z$, $B\cap Y$, $B\cap Z$ contains
 a nonextremal element, but $A\cap Y$ does not.}}

We now define $\kappa^{\pentago} \eqdef \sigma \join{s^*}{x_{AY}}
\kappa^\triangle$ and $\alpha^{\pentago} \eqdef \sigma \join{s^*}{x_{AY}}
\alpha^\triangle$.  We note that $\kappa^{\pentago} \join{y^*}{z^*}
\xi^\triangle = \alpha^{\pentago} \join{a^*}{b^*} \beta^\triangle$,
since both are restrictions of $\chi$ on $(A \cap Y) \cup
\{x_{AZ},x_{BZ},x_{BY}\}$ (replacing proxy elements by arbitrary elements
in the sets that they represent).  We denote this chirotope by
$\chi^{\hexago}$.  All of $\alpha^{\pentago}$, $\kappa^{\pentago}$ and
$\chi^{\hexago}$ are convex.  Then we can write
\begin{align*}
T_1&=\ \begin{tikzpicture}[baseline={(0,-0.1)},scale=0.65, every node/.style={transform shape}]
\node[circle,draw,minimum size=  1cm, very thin] (kappa) at (0,0) {$\kappa$};
\node[circle,draw,minimum size=  1cm, very thin] (xi) at (2,0) {$\xi$};
\node[red] (y*) at ($(kappa)+(0:0.35)$) {$\bullet$};
\node[red] (z*) at ($(xi)+(180:0.35)$) {$\bullet$};
\node[red,font=\footnotesize] at ($(y*)+(.35,0.2)$) {$y^*$};
\node[red,font=\footnotesize] at ($(z*)-(.25,-0.22)$) {$z^*$};
\draw[red, thin] (y*.center)--(z*.center);
\end{tikzpicture}
\quad \xrightarrow{\pap}^* \quad
\begin{tikzpicture}[baseline={(0,-0.1)},scale=0.65, every node/.style={transform shape}]
  \node[circle,draw,minimum size=  1.2cm, very thin] (kappaT) at (0,0) {$\kappa^{\pentago}$};
\node[circle,draw,minimum size=  1.2cm, very thin] (xiT) at (2,0) {$\xi^\triangle$};
\node[circle,draw,minimum size=  1cm, very thin] (pi) at ($(kappaT)+(180+45:1.8)$) {$\pi$};
\node[circle,draw,minimum size=  1cm, very thin] (rho) at ($(xiT)+(45:1.8)$) {$\rho$};
\node[circle,draw,minimum size=  1cm, very thin] (tau) at ($(xiT)+(-45:1.8)$) {$\tau$};
\node[red] (p*) at ($(pi)+(45:0.35)$) {$\bullet$};
\node[red] (r*) at ($(rho)+(180+45:0.35)$) {$\bullet$};
\node[red] (t*) at ($(tau)+(180-45:0.35)$) {$\bullet$};
\node[red] (y*) at ($(kappaT)+(0:0.45)$) {$\bullet$};
\node[red] (z*) at ($(xiT)+(180:0.45)$) {$\bullet$};
\node[red] (xBY) at ($(kappaT)+(180+45:0.45)$) {$\bullet$};
\node[red] (xAZ) at ($(xiT)+(45:0.45)$) {$\bullet$};
\node[red] (xBZ) at ($(xiT)+(-45:0.45)$) {$\bullet$};
\node[red,font=\footnotesize] at ($(y*)+(.35,0.2)$) {$y^*$};
\node[red,font=\footnotesize] at ($(z*)-(.25,-0.22)$) {$z^*$};
\node[red,font=\footnotesize] at ($(p*)+(10:0.45)$) {$p^*$};
\node[red,font=\footnotesize] at ($(t*)+(180:0.35)$) {$t^*$};
\node[red,font=\footnotesize] at ($(r*)+(180:0.4)$) {$r^*$};
\node[red,font=\footnotesize] at ($(xAZ)+(10:0.55)$) {$x_{AZ}$};
\node[red,font=\footnotesize] at ($(xBZ)+(-10:0.55)$) {$x_{BZ}$};
\node[red,font=\footnotesize] at ($(xBY)+(190:0.55)$) {$x_{BY}$};
\draw[red, thin] (y*.center)--(z*.center);
\draw[red, thin] (r*.center)--(xAZ.center);
\draw[red, thin] (t*.center)--(xBZ.center);
\draw[red, thin] (p*.center)--(xBY.center);
\end{tikzpicture}
\quad  \xrightarrow{\lozenge} \quad
\begin{tikzpicture}[baseline={(0,0)},scale=0.65, every node/.style={transform shape}]
  \node[circle,draw,minimum size=  1.2cm, very thin] (chiS) at (0,0) {$\chi^{\hexago}$};
\node[circle,draw,minimum size=  1cm, very thin] (pi) at ($(chiS)+(180+45:1.8)$) {$\pi$};
\node[circle,draw,minimum size=  1cm, very thin] (rho) at ($(chiS)+(45:1.8)$) {$\rho$};
\node[circle,draw,minimum size=  1cm, very thin] (tau) at ($(chiS)+(-45:1.8)$) {$\tau$};
\node[red] (p*) at ($(pi)+(45:0.35)$) {$\bullet$};
\node[red] (r*) at ($(rho)+(180+45:0.35)$) {$\bullet$};
\node[red] (t*) at ($(tau)+(180-45:0.35)$) {$\bullet$};
\node[red] (xBY) at ($(chiS)+(180+45:0.45)$) {$\bullet$};
\node[red] (xAZ) at ($(chiS)+(45:0.45)$) {$\bullet$};
\node[red] (xBZ) at ($(chiS)+(-45:0.45)$) {$\bullet$};
\node[red,font=\footnotesize] at ($(p*)+(10:0.45)$) {$p^*$};
\node[red,font=\footnotesize] at ($(t*)+(180:0.35)$) {$t^*$};
\node[red,font=\footnotesize] at ($(r*)+(180:0.4)$) {$r^*$};
\node[red,font=\footnotesize] at ($(xAZ)+(10:0.55)$) {$x_{AZ}$};
\node[red,font=\footnotesize] at ($(xBZ)+(-10:0.55)$) {$x_{BZ}$};
\node[red,font=\footnotesize] at ($(xBY)+(190:0.55)$) {$x_{BY}$};
\draw[red, thin] (r*.center)--(xAZ.center);
\draw[red, thin] (t*.center)--(xBZ.center);
\draw[red, thin] (p*.center)--(xBY.center);
\end{tikzpicture};\\
T_2&=\ \begin{tikzpicture}[baseline={(0,-0.1)},scale=0.65, every node/.style={transform shape}]
\node[circle,draw,minimum size=  1cm, very thin] (kappa) at (0,0) {$\alpha$};
\node[circle,draw,minimum size=  1cm, very thin] (xi) at (2,0) {$\beta$};
\node[red] (y*) at ($(kappa)+(0:0.35)$) {$\bullet$};
\node[red] (z*) at ($(xi)+(180:0.35)$) {$\bullet$};
\node[red,font=\footnotesize] at ($(y*)+(.35,0.2)$) {$a^*$};
\node[red,font=\footnotesize] at ($(z*)-(.25,-0.22)$) {$b^*$};
\draw[red, thin] (y*.center)--(z*.center);
\end{tikzpicture}
\quad \xrightarrow{\pap}^* \quad
\begin{tikzpicture}[baseline={(0,-0.1)},scale=0.65, every node/.style={transform shape}]
  \node[circle,draw,minimum size=  1.2cm, very thin] (kappaT) at (0,0) {$\alpha^{\pentago}$};
\node[circle,draw,minimum size=  1.2cm, very thin] (xiT) at (2,0) {$\beta^\triangle$};
\node[circle,draw,minimum size=  1cm, very thin] (sigma) at ($(kappaT)+(180-45:1.8)$) {$\rho$};
\node[circle,draw,minimum size=  1cm, very thin] (rho) at ($(xiT)+(45:1.8)$) {$\tau$};
\node[circle,draw,minimum size=  1cm, very thin] (tau) at ($(xiT)+(-45:1.8)$) {$\pi$};
\node[red] (s*) at ($(sigma)+(-45:0.35)$) {$\bullet$};
\node[red] (p*) at ($(rho)+(180+45:0.35)$) {$\bullet$};
\node[red] (t*) at ($(tau)+(180-45:0.35)$) {$\bullet$};
\node[red] (a*) at ($(kappaT)+(0:0.45)$) {$\bullet$};
\node[red] (b*) at ($(xiT)+(180:0.45)$) {$\bullet$};
\node[red] (xAY) at ($(kappaT)+(180-45:0.45)$) {$\bullet$};
\node[red] (xBY) at ($(xiT)+(45:0.45)$) {$\bullet$};
\node[red] (xBZ) at ($(xiT)+(-45:0.45)$) {$\bullet$};
\node[red,font=\footnotesize] at ($(a*)+(.35,0.2)$) {$a^*$};
\node[red,font=\footnotesize] at ($(b*)-(.25,-0.22)$) {$b^*$};
\node[red,font=\footnotesize] at ($(p*)+(190:0.35)$) {$t^*$};
\node[red,font=\footnotesize] at ($(t*)+(180:0.35)$) {$p^*$};
\node[red,font=\footnotesize] at ($(s*)+(0:0.4)$) {$r^*$};
\node[red,font=\footnotesize] at ($(xAY)+(170:0.55)$) {$x_{AZ}$};
\node[red,font=\footnotesize] at ($(xBZ)+(-10:0.55)$) {$x_{BY}$};
\node[red,font=\footnotesize] at ($(xBY)+(10:0.55)$) {$x_{BZ}$};
\draw[red, thin] (y*.center)--(z*.center);
\draw[red, thin] (s*.center)--(xAY.center);
\draw[red, thin] (t*.center)--(xBZ.center);
\draw[red, thin] (p*.center)--(xBY.center);
\end{tikzpicture}
\quad \xrightarrow{\lozenge} \quad
\begin{tikzpicture}[baseline={(0,0)},scale=0.65, every node/.style={transform shape}]
  \node[circle,draw,minimum size=  1.2cm, very thin] (chiS) at (0,0) {$\chi^{\hexago}$};
\node[circle,draw,minimum size=  1cm, very thin] (sigma) at ($(chiS)+(180-45:1.8)$) {$\rho$};
\node[circle,draw,minimum size=  1cm, very thin] (rho) at ($(chiS)+(45:1.8)$) {$\tau$};
\node[circle,draw,minimum size=  1cm, very thin] (tau) at ($(chiS)+(-45:1.8)$) {$\pi$};
\node[red] (s*) at ($(sigma)+(-45:0.35)$) {$\bullet$};
\node[red] (r*) at ($(rho)+(180+45:0.35)$) {$\bullet$};
\node[red] (t*) at ($(tau)+(180-45:0.35)$) {$\bullet$};
\node[red] (xAY) at ($(chiS)+(180-45:0.45)$) {$\bullet$};
\node[red] (xAZ) at ($(chiS)+(45:0.45)$) {$\bullet$};
\node[red] (xBZ) at ($(chiS)+(-45:0.45)$) {$\bullet$};
\node[red,font=\footnotesize] at ($(t*)+(180:0.35)$) {$p^*$};
\node[red,font=\footnotesize] at ($(r*)+(180:0.4)$) {$t^*$};
\node[red,font=\footnotesize] at ($(s*)+(0:0.4)$) {$r^*$};
\node[red,font=\footnotesize] at ($(xAY)+(170:0.55)$) {$x_{AZ}$};
\node[red,font=\footnotesize] at ($(xAZ)+(10:0.55)$) {$x_{BZ}$};
\node[red,font=\footnotesize] at ($(xBZ)+(-10:0.55)$) {$x_{BY}$};
\draw[red, thin] (r*.center)--(xAZ.center);
\draw[red, thin] (s*.center)--(xAY.center);
\draw[red, thin] (t*.center)--(xBZ.center);
\end{tikzpicture},
\end{align*}
showing the local confluence in this case. Other cases are treated similarly.
\end{proof}

\subsubsection{Concluding the proof of \cref{thm:uniqueness_canonical}}

We can now put things together to prove existence and uniqueness of the canonical chirotope tree representing any given chirotope $\chi$.

Let $\chi$ be a chirotope. As noticed in \cref{ssec:existence_canonical},
the existence of a canonical chirotope tree representing $\chi$ is a consequence of \cref{lem:terminates}.
Moreover, by \cref{lem:T_accessible} any canonical chirotope tree representing $\chi$ can be obtained from $T^{tr}_\chi$
through the rewriting rule $\Rightarrow$. We also observe that we can never apply any further rewriting $\Rightarrow$ to a canonical
chirotope tree, i.e.~these are terminal states.
But we have proved the confluence of the rewriting rule $\Rightarrow$ (\cref{prop:local_confluence}, combined with Newman's lemma),
so that for a given initial state, there is a unique possible terminal state.
Consequenly, there is a unique canonical chirotope tree representing $\chi$, and \cref{thm:uniqueness_canonical} is proved. \qed

\section{Almost all realizable chirotopes are indecomposable}
\label{s:almost}

In this section, we consider only realizable chirotopes.
Recall from the introduction that $t_n$ and $d_n$ are respectively the number of
realizable chirotopes and decomposable realizable chirotopes on $\{1,\dots,n\}$.
Our goal is to prove \cref{thm:proportion_indecomposable},
which we copy here for convenience.
\medskip

\noindent \textbf{\cref{thm:proportion_indecomposable}.}
For large $n$, we have $3 \pth{n-\O(n^{-2})} t_{n-1} \le d_n \le \O(n^{-3}) t_n$.
\medskip

Stated differently, a uniform random realizable chirotope $\chi_n$ on $\{1,\dots,n\}$,
 is indecomposable with probability $1-\O(n^{-3})$.

 \subsection{A consequence for unlabeled chirotopes}
Before proving the theorem, let us state a consequence for {\em unlabeled
chirotopes}.
By definition, an  {\em unlabeled chirotope} is an equivalence class of chirotopes
for the following relation:
 $\chi:(X)_3 \to \{\pm 1\}$ and $\chi':(Y)_3 \to \{\pm 1\}$ are equivalent
  if there exists a bijection $\phi: X \to Y$ such that, for all $x_1,x_2,x_3$ in $X$
\[\chi'(\phi(x_1),\phi(x_2),\phi(x_3)) = \chi(x_1,x_2,x_3).\]
Many natural notions and properties of chirotopes (being realizable, the
number of extreme points, etc.) do not depend on its labeling
and are thus well-defined on unlabeled chirotopes.
The notion of indecomposability is such a notion.

We remark that, to an unlabeled chirotope correspond several chirotopes on $\{1,\cdots,n\}$,
the number of which can vary depending on its automorphism group.
This implies that taking a uniform random realizable chirotope on $\{1,\dots,n\}$ and forgetting its labeling
or taking directly a uniform random unlabeled realizable chirotope on $n$  elements
lead to two different probability distributions.
However, symmetry can be controlled, and \cref{thm:proportion_indecomposable}
implies the following.
\begin{corollary}
  Let $\bar\chi_n$ be a uniform random unlabeled realizable chirotope on $n$
  elements. Then, $\bar\chi_n$ is indecomposable with probability  $1-\O(n^{-2})$.
\end{corollary}
\begin{proof}
  Let $t^u_n$ denote the number of unlabeled realizable chirotopes on $n$
  elements and $d^u_n$ the number of decomposable ones.
   Since a chirotope on $n$ elements has at most $n$
  symmetries~\cite[Theorem 1.5]{goaoc2023convex-hulls}, we have
  \[ \frac{t_n}{n!} \le t_n^u \le \frac{t_n}{(n-1)!} \qquad \text{and} \qquad \frac{d_n}{n!} \le d_n^u \le \frac{d_n}{(n-1)!}.\]
  Theorem~\ref{thm:proportion_indecomposable} asserts that $d_n/t_n = \O(n^{-3})$, from which we get
  \[ \frac{d_n^u}{t_n^u} \le n \frac{d_n}{t_n} = \O(n^{-2}). \qedhere\]
\end{proof}

\subsection{Some classical estimates on the number of realizable chirotopes}
\label{ssec:counting}~

We start by recalling essentially known results
about the number $t_n$ of realizable chirotopes on $\{1,\dots,n\}$.
A general reference on this topic is \cite[Section 6]{goodman1993order-types}.
The first statement that we need
 is the following nonasymptotic version of the case $d=2$
  in~\cite[Theorem 6.6]{goodman1993order-types}.
\begin{proposition}
\label{prop:upper-bound-chirotopes}
For all $n \ge 2$, we have $t_n \le e^{2n}(n-1)^{4n}$.
\end{proposition}
\begin{proof}
We reproduce here the argument essentially given in \cite{goodman1993order-types}.
Let $\mathfrak p_i =(p_{i,x},p_{i,y})$, for $1 \le i\le n$, represent
an unknown set of $n$ points in $\mathbb R^2$, indexed by $\{1,\dots,n\}$.
As already noticed in the proof of \cref{lem:extreme_in_unbounded_cell},
for $i<j<k$ in $\{1,\dots,n\}$,
we have
\[\chi(\mathfrak p_i,\mathfrak p_j,\mathfrak p_k) = \sgn \left[\det
\left( \begin{array}{ccc} p_{i,x} & p_{j,x} & p_{k,x} \\ p_{i,y} & p_{j,y} & p_{k,y} \\ 1&1 &1 \end{array} \right) \right].\]
The determinant in the right-hand side is a polynomial of degree $2$ in the variables $\{ p_{i,x},p_{i,y}, 1 \le i \le n\}$, which
we denote $D_{i,j,k}$.
A realizable chirotope on $\{1,\dots,n\}$ corresponds to a {\em sign pattern}
of the collection $(D_{i,j,k})_{i<j<k}$ of polynomials,
i.e.~a choice of a sign for each $D_{i,j,k}$ which can be realized
by some real specialization of the variables $\{ p_{i,x},p_{i,y}, 1 \le i \le n\}$.
From a result of Warren \cite{warren1968lower}, the number of such sign patterns for general collection of polynomials
is bounded by $(\tfrac{4edm}{N})^N$, where $d$ is the degree of the polynomials, $m$ their number and $N$ the number of variables.
In our case $d=2$, $m=\binom{n}3$ and $N=2n$, and we get
\[ t_n\leq \left(4e\frac{2\binom{n}{3}}{2n}\right)^{2n} \]
The lemma follows using the simple bound $\binom{n}3 \leq n(n-1)^2/6$.
\end{proof}
The following lemma is implicit in the discussion preceeding Theorem 6.1 in \cite{goodman1993order-types}.
\begin{lemma}
\label{lem:bound-extensions}
There exists a universal constant $K$ in $(0,1/8]$ such that, for all $n \ge 3$,
we have $t_{n+1} \ge K \, n^4 \, t_n$.
\end{lemma}
\begin{proof}[Sketch of proof]
  The number $t_{n+1}$ rewrites as the sum, over all realizable
  chirotopes $\chi$ on $[n]$, of the number of extensions of $\chi$,
  meaning the number of realizable chirotopes on $[n+1]$ whose
  restriction to (triples of) $[n]$ is $\chi$. Fix a geometric
  realization $P$ of $\chi$; the number of extensions of $\chi$ is
  bounded from below by the number of chirotopes realized by $1$-point
  extensions of $P$. This is exactly the number of cells in the
  arrangement of the $\binom n2$ lines through pairs of points of $P$,
  which is~\cite[p. 65]{za1997}
  \[\binom{\binom{n}2}2 + \binom{n}2 + 1 - n\binom{n-2}{2} = \frac{n^4}8 - \frac{r(n)}{8},\]
with $r(n)=6n^3-23n^2+26n-8\leq \frac{n^4}{2}$ for $n\geq 3$, by a simple function analysis. The statement follows from this bound by taking  $K=\frac1{16}$.
\end{proof}

\subsection{Some basic relations between numbers of chirotopes}

Let $t^*_n$ be the number of realizable chirotopes on $\{1,\dots,n\}$,
in which $n$ is an extreme element.
We now relate $d_n$, $t_n$ and $t_n^*$, starting with the last two.

\begin{lemma}\label{l:tn et tn*}
	For each $n \ge 3$, we have $\frac3n t_n \le t_n^*\leq \frac{4}{n}t_n$.
\end{lemma}
\begin{proof}
For a chirotope $\chi$, we denote by $\ext(\chi)$ its number of extreme elements.
  By symmetry, we have:
  \[ \sum_{\chi \in {\mathcal X}_n} \ext(\chi) = \big|\{(\chi,i) \in {\mathcal X}_n \times [n] \colon i \text{ is extreme in } \chi\} \big| = nt_n^*.\]
  Hence,
  \[ \frac{t_n^*}{t_n} = \frac1{nt_n} \sum_{\chi \in {\mathcal X}_n} \ext(\chi) = \frac1n \Ex[\ext(\chi)],\]
  where $\chi$ is chosen uniformly at random in
  $\mathcal{X}_n$. By~\cite[Theorem~1.2]{goaoc2023convex-hulls},
  $$\Ex[\ext(\chi)] = 4-\frac{8}{n^2-n+2},$$ which is
  between $3$ and $4$.
\end{proof}

\noindent
Now, let us turn our attention to $d_n$.

\begin{lemma}\label{lem:bound_dn}
  For each $n \ge 5$, we have
  \begin{equation}\label{eq:bounding-dn}
    3n(t_{n-1}-d_{n-1}) \le d_n \le \sum_{n_1+n_2=n \atop n_1,n_2 \ge 2} \binom{n}{n_1}\, t^*_{n_1+1}\, t^*_{n_2+1}.
  \end{equation}
\end{lemma}
\begin{proof}
Let $\mathcal{X}^*_{n}$ denote the set of realizable chirotopes on $[n] \eqdef \{1, 2, \dots , n\}$ such that $n$ is extreme,
 and $\mathcal D_n$ the set of decomposable realizable chirotopes of size $n$.
We define an application $\phi$ that associates to a pair $(\chi,M)$
-- where $\chi\in\mathcal D_n$ and $M \subseteq [n]$ is a nontrivial module of $\chi$ --
 a tuple $\phi(\chi,M)=(\chi_1,\chi_2,M)$ constructed as follows.
 Letting $n_1 \eqdef |M|$ and $n_2 \eqdef n-n_1$,
  \begin{itemize}
  \item $\chi_1$ is the chirotope  obtained by replacing the elements of $[n] \setminus M$ by a single element labeled $n_1+1$, and relabeling $M$ from $1$ to $n_1$ consistently with the order of $M$;
  \item  $\chi_2$ is the chirotope obtained, similarly, by replacing the elements of $M$ by a single element labeled $n_2+1$, and relabeling $[n]\setminus M$ from $1$ to $n_2$ consistently.
  \end{itemize}
  From \cref{prop:Realization_Substitution,prop:module_and_join}, it
  holds that $n_1+1$ and $n_2+1$ are extreme elements in $\chi_1$ and
  $\chi_2$ respectively, i.e.~ $\chi_1 \in \mathcal X^*_{n_1+1}$ and
  $\chi_2 \in \mathcal X^*_{n_2+1}$.  We claim that $\phi$ is
  bijective.  Indeed, take a triple $(\chi_1,\chi_2,M)$, with $\chi_1
  \in \mathcal X^*_{n_1+1}$, $\chi_2 \in \mathcal X^*_{n_2+1}$ and $M
  \subseteq [n]$.  Let us define $\chi_1^M$ as the chirotope over $M
  \cup \{x^*\}$ obtained by renaming, in $\chi_1$, the (extreme)
  element $n_1+1$ as $x^*$ and $[n_1]$ as $M$, increasingly. The
  chirotope $\chi_2^{[n]\setminus M}$ is defined similarly with label
  set $([n]\setminus M) \cup\{y^*\}$. Finally, we set $\chi = \chi_1^M
  \join{x^*}{y^*} \chi_2^{[n]\setminus M}$ and define
  $\Psi(\chi_1,\chi_2,M)=(\chi,M)$.  One readily checks that $\phi$
  and $\Psi$ are inverse of each other, showing that $\phi$ is
  bijective.

 As a consequence, we have
\[
	d_n \leq |\{(\chi,M)\,:\, \chi\in \mathcal D_n,\, M\text{ nontrivial module of }\chi\}| = \sum_{n_1+n_2=n \atop n_1,n_2 \geq 2} \binom{n}{n_1}\, t^*_{n_1+1}\, t^*_{n_2+1}, \]
proving the upper bound in~\cref{eq:bounding-dn}.

On the other hand, the chirotopes $\chi_1 \join{x^*}{y^*}
  \chi_2$ are different for any $\chi_1 \in {\cal X}_3$ and $\chi_2
  \in {\cal X}_{n-1}\setminus{\cal D}_{n-1}$ (this can easily be seen
      {\em e.g.} via Theorem~\ref{thm:uniqueness_canonical}).
Since there are $2$ chirotopes on $\{1,2,3\}$, it follows that $d_n
\ge 2 \binom{n}{2} i^*_{n-1}$
where $i^*_{n-1}$ is the number of indecomposable realizable chirotopes
on $[n-1]$ in which $n-1$ is an extreme point.
Using the same argument as in \cref{l:tn et tn*}, we have
\[i^*_{n-1} \ge \frac{3}{n-1} (t_{n-1}-d_{n-1}),\]
concluding the proof of the lower bound in~\cref{eq:bounding-dn}.
	\end{proof}

Combining \cref{l:tn et tn*,lem:bound_dn}, we get
\begin{equation}\label{eq:bounding-dn-sum}
d_n \le 16 \sum_{n_1+n_2=n \atop n_1,n_2 \geq 2} \binom{n}{n_1}\, \frac{t_{n_1+1}\, t_{n_2+1}}{(n_1+1)\,(n_2+1)}.
\end{equation}

\subsection{Bounding the sum}
The goal is to prove that the upper bound in \cref{eq:bounding-dn-sum} behaves as $\O(n^{-3} t_n)$.
We fix some $\varepsilon$ in $(0,1/2)$ independent of $n$
and split the sum into two parts.
\begin{align} d^{\rm small}_n&=\sum_{\substack{n_1+n_2=n \\ 2 \leq \min(n_1,n_2) < \varepsilon n}} \binom{n}{n_1}\, \frac{t_{n_1+1}\, t_{n_2+1}}{(n_1+1)\,(n_2+1)},\text{ and } \label{eq:def_dsmall}\\
 d^{\rm large}_n&=\sum_{\substack{n_1+n_2=n \\ \varepsilon n \le n_1,n_2 \le (1-\varepsilon) n}} \binom{n}{n_1}\, \frac{t_{n_1+1}\, t_{n_2+1}}{(n_1+1)\,(n_2+1)}. \label{eq:def_dlarge}
 \end{align}

The part $d^{\rm small}_n$ can be bounded using the classical
estimates on the number of realizable chirotopes given in \cref{ssec:counting}.
\begin{lemma}
\label{lem:bound-dn-small}
	For $\varepsilon$ small enough, we have $\frac{d^{\rm small}_n}{t_n}=\O(n^{-3})$
	as $n$ tends to $+\infty$.
\end{lemma}
\begin{proof}
	Assuming $\varepsilon<\tfrac12$, we write
	\begin{align*}
		n^3\frac{d^{\rm small}_n}{t_n}&=2n^3\sum_{n_1+n_2=n\atop 2\leq n_1< \varepsilon n} \binom{n}{n_1}\, \frac{t_{n_1+1}}{n_1+1}\, \frac{t_{n_2+1}}{(n_2+1)t_n}.
	\end{align*}
	Iterating \cref{lem:bound-extensions}, we get
	\[\frac{t_{n_2+1}}{t_n}\leq \frac{1}{K^{n_1-1} (n_2+1)^{4}\ldots (n-1)^{4}}.\]

	On the other hand, \cref{prop:upper-bound-chirotopes} yields:
	\[\frac{t_{n_1+1}}{n_1+1}\leq e^{2n_1+2}n_1^{4(n_1+1)-1}.\]
	Bringing everything together gives us:
	\begin{align*}
		n^3\frac{d^{\rm small}_n}{t_n}&\leq 2\sum_{n_1+n_2=n\atop 2\leq n_1\leq \varepsilon n} \frac{n^4 (n-1)\ldots (n_2+1)}{n_1!} \frac{e^{2n_1+2}n_1^{4(n_1+1)-1}}{K^{n_1-1}(n_2+1)^5(n_2+2)^4\ldots (n-1)^4}
	\end{align*}
	Using that for $n$ large enough, $n^4\leq 2(n-2)(n-1)^3$, and $n_1!\geq \left(\frac{n_1}{e}\right)^{n_1}$, we get
	\begin{align*}
		n^3\frac{d^{\rm small}_n}{t_n}&\leq 4Ke^2\sum_{n_1+n_2=n\atop 2\leq n_1\leq \varepsilon n} \left(\frac{e}{n_1}\right)^{n_1} \frac{e^{2n_1}n_1^{4(n_1+1)-1}}{K^{n_1}(n_2+1)^4(n_2+2)^3\ldots (n-3)^3(n-2)^2}\,.
	\end{align*}
	When $n_2=n-2$, the product in the denominator $(n_2+1)^4(n_2+2)^3\ldots (n-3)^3(n-2)^2$ should be interpreted as $1$ (in this case, we use $n^4\leq 2(n-1)^4$).

	Each of the $3(n-n_2-2)=3n_1-6$ polynomial factors in the denominator are bounded from below by
	\[n_2\geq (1-\varepsilon)n\geq \frac{1-\varepsilon}{\varepsilon}n_1\,.\]
	Hence, we get, for $n$ large enough,
	\begin{align*}
		n^3\frac{d^{\rm small}_n}{t_n}&\leq 4Ke^2\left(\frac{1-\varepsilon}{\varepsilon}\right)^6\sum_{2\leq n_1\leq \varepsilon n} \left(\frac{e^3}{K}\left(\frac{\varepsilon}{1-\varepsilon}\right)^3\right)^{n_1} \frac{n_1^{4(n_1+1)-1}}{n_1^{n_1}n_1^{3n_1-6}}\\
		&\leq 4Ke^2\left(\frac{1-\varepsilon}{\varepsilon}\right)^6\sum_{n_1\geq 2} \left(\frac{e^3}{K}\left(\frac{\varepsilon}{1-\varepsilon}\right)^3\right)^{n_1} n_1^{9}.
	\end{align*}
	For $\varepsilon$ small enough, we have a convergent series, showing that the right-hand size is bounded by a finite constant.
\end{proof}

We note that the above strategy
(via \cref{lem:bound-extensions,prop:upper-bound-chirotopes})
would even fail to show that $t_{n/2+1}^2 \le t_n$, and
can therefore not be used to control $d^{\rm large}_n$.
In the next section, we will prove the following lemma with a completely different method.
\begin{lemma}\label{lem:bounding-d-large}
For $n \ge \frac3{\varepsilon}$, we have
$\displaystyle t_n\geq \frac{3K\varepsilon^6}{4}\, n^3 \, d^{\rm large}_n$, where $K$ is the constant introduced in Lemma~\ref{lem:bound-extensions}.
\end{lemma}

\begin{proof}[Proof of \cref{thm:proportion_indecomposable}]
Combining \cref{eq:bounding-dn-sum,eq:def_dsmall,eq:def_dlarge,lem:bound-dn-small,lem:bounding-d-large},
we find that $\frac{d_n}{t_n} =\O(n^{-3})$, showing the upper bound in \cref{thm:proportion_indecomposable}.
The lower bound follows from \cref{lem:bound_dn}, using that
$\frac{d_{n-1}}{t_{n-1}} =\O(n^{-3})$.
\end{proof}

\section{Quasi-modules and the proof of Lemma~\ref{lem:bounding-d-large}}~
\label{s:quasi}

 \cref{lem:bounding-d-large} gives a lower bound on $t_n$ in terms of
 products $t_{n_1+1}t_{n_2+1}$ where $n_1+n_2 =n$. We prove it by
 constructing, given two small realizable chirotopes, a larger one via a
 generalized substitution operation ``along a segment'' (defined in
 Section~\ref{sec:subs along a segment}). This operation is not
 one-to-one, but the redundancy can be controlled using a notion of
 ``quasi-modules'' (defined in Section~\ref{sec:quasi-modules}).

\subsection{Quasi-modules}\label{sec:quasi-modules}

Let $\chi$ be a chirotope\footnote{Even though we are interested in realizable chirotope,
the results in \cref{sec:quasi-modules} are valid for abstract chirotopes in general, and thus
stated in this setting.} on a set $X$. Let $a, b \in X$ and let $Y$
be a subset of $X$ disjoint from $\{a,b\}$. The line $ab$ splits $Y$
into $Y=Y_{-1} \uplus Y_1$, where, for $i \in \{-1,1\}$,
\[Y_i=\{y \in Y, \, \chi(a,b,y)=i\}.\]
We call the set $\{Y_{-1},Y_1\}$ the {\em bipartition} of $Y$ induced
by $a$ and $b$. Note that although switching $a$ and $b$ exchanges
$Y_{-1}$ and $Y_1$, the induced bipartition is the same. A bipartition
is called nontrivial if \emph{both} $Y_{-1}$ and $Y_1$ are nonempty. This gives rise to a variant of the notion of modules.

\begin{definition}
  A subset $W$ of $X$ is called a {\em quasi-module} of a chirotope
  $\chi$ on $X$ if every pair $\{a,b\}$ of elements of $W$ defines the
  same bipartition of $X \setminus W$, and if this bipartition is
  nontrivial (meaning that neither part is empty).
\end{definition}

We compare the above definition to our earlier definition of module
(see \cref{def:modules}).  If $M \subset X$ is a module of a chirotope
$\chi$, then every pair of elements of $M$ defines the same
bipartition on $X \setminus M$.  However, this bipartition is trivial,
so a module is not a quasi-module. Another difference is that the
definition of module is symmetric in $M$ and $X \setminus M$, while
that of quasi-module is not.

An example of quasi-module is given on Figure~\ref{f:weakmod}. On the other hand,
consider a chirotope $\chi$ on $X$ with exactly three extreme elements
$\{a,b,c\}$. Each of $\{a,b\}$, $\{a,c\}$ and $\{b,c\}$ defines the
same bipartition $\{X\setminus \{a,b,c\}, \emptyset\}$. Yet, the set
$\{a,b,c\}$ is not a quasi-module as it fails the nontriviality
condition.

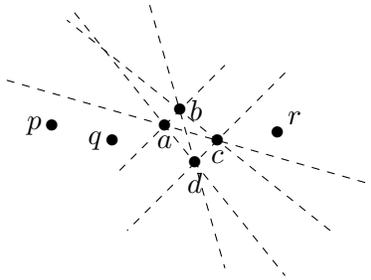
\begin{figure}[h]
\centering
\tikzset{
    add/.style args={#1 and #2}{
        to path={
 ($(\tikztostart)!-#1!(\tikztotarget)$)--($(\tikztotarget)!-#2!(\tikztostart)$)%
  \tikztonodes},add/.default={.2 and .2}}
}
	\begin{tikzpicture}
	\node (p1) at (0,0) {$\bullet$};
	\node[below] at (0,0) {$a$};
	\node (p2) at (0.2,0.2) {$\bullet$};
	\node[right] at (0.2,0.2) {$b$};
	\node (p3) at (0.7,-0.2) {$\bullet$};
	\node[below] at (0.7,-0.2) {$c$};
	\node (p4) at (0.4,-0.5) {$\bullet$};
	\node[below] at (0.4,-0.5) {$d$};
	\draw [add= 3 and 3, black, thin, dashed] (p1) to (p2);
	\draw [add= 3 and 3, black, thin, dashed] (p1) to (p3);
	\draw [add= 3 and 3, black, thin, dashed] (p3) to (p2);
	\draw [add= 3 and 3, black, thin, dashed] (p1) to (p4);
	\draw [add= 2 and 2, black, thin, dashed] (p2) to (p4);
	\draw [add= 3 and 3, black, thin, dashed] (p3) to (p4);
	\node at (-.7,-0.2) {$\bullet$};
	\node[left] at (-.7,-0.2) {$q$};
	\node at (-1.5,-0) {$\bullet$};
	\node[left] at (-1.5,-0) {$p$};
	\node (p3) at (1.5,-0.1) {$\bullet$};
	\node[right] at (1.5,0.1) {$r$};
\end{tikzpicture}
\caption{Every line of two points in $W=\{a,b,c,d\}$ defines the same bipartition of points in $\{p,q,r\}$, so $W$ is a quasi-module. On the other hand, $\{a,b,c\}$ is not a quasi-module of that chirotope: indeed, the line $ab$ induces the bipartition $\{\{p,q\},\{d,r\}\}$,
while the line $ac$ induces $\{\{p,q,d\},\{r\}\}$.\label{f:weakmod}}
\end{figure}

It will be useful to control the number of quasi-modules of macroscopic size of any given chirotope.

\begin{lemma}
  For every $\eps>0$ and for every $n \ge 1$ such that $\eps n \ge 2$,
  every chirotope of size $n$ has at most $n^2/\eps$ quasi-modules of
  size at least $\eps\, n$.
  \label{lem:counting_weak_modules}
\end{lemma}
\begin{proof}
  Fix a chirotope $\chi$ on a set $X$ of size $n$. We will use a
  double-counting argument based on an upper bound on the number of
  pairs $(W,a)$, where $W$ is a quasi-module and $a$ a
  point of $W$.

  To construct such triples, we start by choosing a point $a$
  in $X$. We label the elements of $X\setminus \{a\}$ as $x_0$,
  $x_1$, \ldots, $x_{n-2}$ as follows. 
  We choose $x_0 \ne a$ in an arbitrary way.
  Then we start with the line $ax_0$ and
  rotate it in clockwise direction around $a$. 
  The $i$-th point encountered gets the label $x_i$. Note that we rotate the whole line
  (infinite in both directions) and not only a half-line: 
  in particular, $x_1$ can be on either side of the line $ax_0$.
  An example is given on \Cref{fig:labelling_XseminusA}.
  The resulting labelling of $X\setminus \{a\}$ depends on the chosen point $x_0$
  (namely, changing $x_0$ results in a cyclic shift of the labels).

  \begin{figure}[h]
  \centering
  	\begin{tikzpicture}
	\node at (0,0) {$\bullet$};
	\node[below right] at (0,0) {$a$};
	\node at (0,2) {$\bullet$};
	\node[right] at (0,2) {$x_0$};
	\draw[dashed] (0,0) ++(90:3cm) -- (-180+90:3cm);

	\draw[dashed] (0,0) ++(70:3cm) node {$\ell_1$} -- (-180+70:3cm);
	\draw (0,0) ++(70:1.2cm) node {$\bullet$};
	\draw[right] (0,0) ++(70:1.2cm) node {$x_1$};

	\draw[dashed] (0,0) ++(34:3cm) node {$\ell_2$} -- (-180+34:3cm);
	\draw (0,0) ++(34:-.5cm) node {$\bullet$};
	\draw[below] (0,0) ++(34:-.5cm) node {$x_2$};

	\draw[dashed] (0,0) ++(14:3cm) node {$\ell_3$} -- (-180+14:3cm);
	\draw (0,0) ++(14:-2cm) node {$\bullet$};
	\draw[below] (0,0) ++(14:-2cm) node {$x_3$};

	\draw[dashed] (0,0) ++(150:3cm)  -- (-180+150:3cm) node {~$~\ell_4$};
	\draw (0,0) ++(150:-2cm) node {$\bullet$};
	\draw[below] (0,0) ++(150:-2cm) node {$x_4$};

        \end{tikzpicture}\caption{Turning the line $ax_0$ clockwise around $a$ reaches  $\ell_1, \ell_2, \ell_3$ and $\ell_4$ in this order, yielding the labels $x_1, x_2, x_3$ and $x_4$.\label{fig:labelling_XseminusA}}
\end{figure}
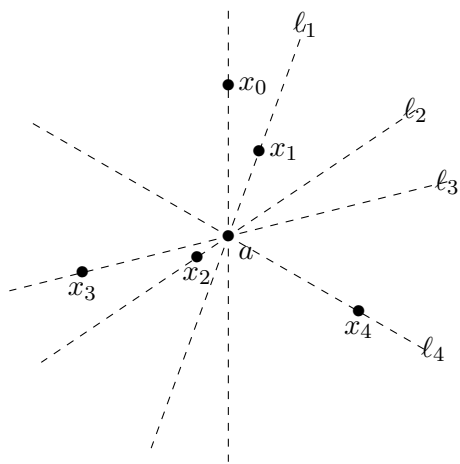

Consider a quasi-module $W$ containing $a$.
Assume that $W$ contains $x_i$ and $x_j$ for some $i<j$.
For any $h,k$ such that $i<h<j<k$ or $k<i<h<j$, the lines $ax_i$ and
  $ax_j$ induce different bipartitions of $\{x_h,x_k\}$,
   see Fig~\ref{f:quasi-modules-are-ordered}.
   This implies that either $x_h$ or $x_k$ is in $W$ (possibly both).
   Since this holds for any $h$ and $k$ such that $i<h<j<k$ or $k<i<h<j$,
 any quasi-module $W$ containing $a$ is of the form
  \begin{equation}\label{eq:form_weak_modules}
    W=\{a\} \cup \big\{x_i, i \in I \big\},
  \end{equation}
 where $I$ is a {\em circular interval} of $\{0,1,\dots,n-2\}$
 (The converse is not true in general, as our
  construction does not control the bipartitions of $X \setminus W$
  induced by lines $x_sx_t$ for $x_s$, $x_t$ in $W$.
  In the example of \Cref{fig:labelling_XseminusA}, $\{a,x_1,x_2,x_3\}$ is a quasi-module,
  but $\{a,x_3,x_4\}$ is not.)
  It follows that at most $n^2$ quasi-modules $W$
  contain a given point $a$. There are thus at most
  $n^3$ pairs $(W,a)$, where $W$ is a quasi-module and $a$ a
  point of $W$.
  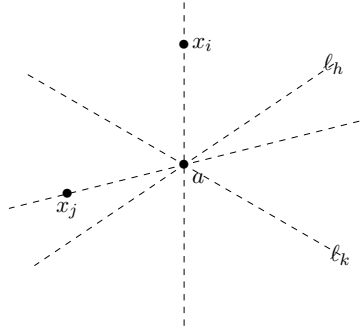
\begin{figure}
  	\centering
  	\scalebox{0.8}{\begin{tikzpicture}
	\node at (0,0) {$\bullet$};
	\node[below right] at (0,0) {$a$};
	\node at (0,2) {$\bullet$};
	\node[right] at (0,2) {$x_i$};
	\draw[dashed] (0,0) ++(90:2.7cm) -- (-180+90:2.7cm);

	\draw[dashed] (0,0) ++(34:3cm) node {$\ell_h$} -- (-180+34:3cm);

	\draw[dashed] (0,0) ++(14:3cm) -- (-180+14:3cm);
	\draw (0,0) ++(14:-2cm) node {$\bullet$};
	\draw[below] (0,0) ++(14:-2cm) node {$x_j$};

	\draw[dashed] (0,0) ++(150:3cm)  -- (-180+150:3cm) node {$\ell_k$};

\end{tikzpicture}}\quad\quad\quad
\caption{The lines $ax_i$ and $ax_j$ divide the plane into four quadrants. These lines induce the same bipartition of $\{x_h,x_k\}$ if and only if these two points lie in equal or opposite quadrants. The condition $i<h<j<k$ (or $k<i<h<j$) and the labelling scheme of $X\setminus\{a\}$ ensures that this is not possible if $x_h$ and $x_k$ lie respectively on $\ell_h$ and $\ell_k$. \label{f:quasi-modules-are-ordered}}
  \end{figure}

  If we denote $w_k$ the number of quasi-modules of size $k$ of the
  chirotope $\chi$, then the number of pairs $(W,a)$ as above is
  $\sum_{k=0}^n k\, w_k$. Summing up, we have
  \[ \sum_{k \ge \varepsilon n}w_k \le \frac1{\varepsilon n}\sum_{k=0}^n k w_k \le \frac{n^3}{\varepsilon n} = \frac{n^2}{\varepsilon},\]
  as announced.
\end{proof}

We conclude this section with the following lemma, which will be useful later.
\begin{lemma}\label{lem:extreme de weak}
  Let $\chi$ be a chirotope on $X$ and $W$ be a quasi-module 
  of $\chi$. There exists a unique pair $\{w_1,w_2\} \subseteq W$ such
  that for all $a \in X\setminus W$ we have
  \begin{equation}\label{eq:extreme_weak}
     \forall w \in W\setminus \{w_1,w_2\}, \quad \chi(a,w_1,w)= \chi(a,w,w_2) \text{ and is independent of } w.
   \end{equation}
\end{lemma}
\begin{proof}
  The statement is trivial if $W$ has size two. Otherwise, fix an arbitrary point
  $b$ in $ X\setminus W$.  We claim that $b$ is extremal in the
  restriction of $\chi$ to $W \cup \{b\}$.  Otherwise, by
  \cref{lem:caratheodory}, $b$ would be contained in a triangle
  $w_1w_2w_3$ of points in $W$ i.e.~one could find $w_1,w_2,w_3$ in
  $W$ such that
  \[\chi(b,w_1,w_2)=\chi(b,w_2,w_3)=\chi(b,w_3,w_1)=1.\]
  By the interiority axiom, this implies $\chi(w_1,w_2,w_3)=1$.
  But from the definition of quasi-module,
  there must exist a point $c  \in X\setminus W$
  such that, for all $w,w'$ in $W$, $\chi(b,w,w')=\chi(c,w',w)$.
  In particular,
  \[\chi(c,w_2,w_1)=\chi(c,w_3,w_2)=\chi(c,w_1,w_3)=1.\]
  Using again the interiority axiom, we get $\chi(w_2,w_1,w_3)=1$,
  leading to a contradiction.
  Hence, $b$ is indeed extremal in the restriction of $\chi$ to $W \cup \{b\}$.

  Therefore, by \cref{lem:order}, the relation $w <_b w' \Leftrightarrow \chi(b,w,w')=1$
  defines a strict total order on $W$. Call $w_1$ and $w_2$ the minimal and maximal
  elements of $W$ for this order relation.
  For all $w \in W\setminus \{w_1,w_2\}$, we have $w_1 <_b w <_b w_2$, i.e.~
  \[\chi(b,w_1,w)= \chi(b,w,w_2)=1.\]
  In particular, \cref{eq:extreme_weak} holds for $a=b$.
  It also for other elements $a \in X\setminus W$, using the fact that $W$ is a quasi-module.
  This proves the existence of the pair $(w_1,w_2)$.

  The uniqueness also follows from the fact that $<_b$ is a total order.
  Only the pair constituted by the two extrema of this total order
  can satisfy \cref{eq:extreme_weak} for $a=b$.
\end{proof}

\noindent
We refer to the two elements $\{w_1,w_2\}$ given by
Lemma~\ref{lem:extreme de weak} as the \emph{antipodal elements} of the
quasi-module $W$.
For example, the antipodal elements of the weak module $\{a,b,c,d\}$
in \Cref{f:weakmod} are $b$ and $d$.

\subsection{Substituting a realizable chirotope along a segment}\label{sec:subs along a segment}

Let $\chi$ be a realizable chirotope on a set $X$ and let $a \in X$ be extreme
for $\chi$.
As we saw in \cref{lem:order}, and immediately after,
we can define $a_+$ (resp. $a_-$) as the unique element of $X$ such that $\chi(a,a_+,x)=1$ for every $x \in X \setminus \{a,a_+\}$
(resp. such that $\chi(a_-,a,x)=1$ for every $x \in X \setminus
\{a,a_-\}$).
(Geometrically, $a_-$, $a$ and $a_+$ are consecutive points on
the convex hull of $\chi$.)
We refer to $a_+$ (resp.~$a_-$) as the {\em successor}
(resp.~the {\em predecessor}) of $a$ on the convex hull of $\chi$.

\begin{definition}\label{def:quasi-module}
  Let $\chi$ be a realizable chirotope on a set $X$ and let $a$ be an element of
  $X$ extreme for~$\chi$. Let $a_+$ and ${a_-}$ be, respectively, the
  successor and predecessor of $a$ on the convex hull of $\chi$.  A
  $\delta$-squeezing of $(\chi,a)$ is a realization $\cal P$ of $\chi$
  such that:
  \begin{itemize}
    \item $\mathfrak p_a = (0,1)$, $\mathfrak p_{a_+}=(-1,0)$ and $\mathfrak p_{a_-}=(1,0)$,
    \item the $y$-coordinates of every other point of $\cal P$ is in $[-\delta,\delta]$, and
    \item every line going through two points of $\cal P \setminus \{\mathfrak p_a\}$ has slope in $[-\delta,\delta]$.
  \end{itemize}
\end{definition}

\noindent
This definition is illustrated in \Cref{f:delta_squeezing}.
\begin{figure}[h]
\[\includegraphics{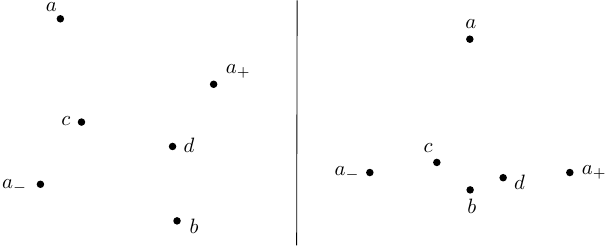}\]
\caption{Two realization of the same chirotope $\chi$, the second one being a 
 $\delta$-squeezing of $(\chi,a)$ for $\delta \approx .5$ (the slope of the line $bc$ is close to $.5$).\label{f:delta_squeezing}}
\end{figure}

\begin{lemma}
  \label{lem:existence_squeezing}
Let $\chi$ be a realizable chirotope on a set $X$ and let $a$ be an element of
$X$ extreme for~$\chi$. For every $\delta>0$, the pair $(\chi,a)$
admits a $\delta$-squeezing.
\end{lemma}
\begin{proof}
  Let $\cal Q$ be a realization of $\chi$ in which $\mathfrak q_a$ is in an
  unbounded cell $\Sigma$ of the line arrangement $\big\{(\mathfrak q_y,\mathfrak q_z),
  z,y \in X \setminus \{a\} \big\}$ (such a realization exists by
  Lemma~\ref{lem:extreme_in_unbounded_cell}). Applying a translation,
  a rotation and a scaling, we can further assume that $\mathfrak q_{a_+}=(-1,0)$
  and $\mathfrak q_{a_-}=(1,0)$ while keeping $\mathfrak q_a$ in an unbounded cell $\Sigma$;
  note that $\mathfrak q_a$ (and thus $\Sigma$) lies above the $x$-axis. Let
  $\vec u = (u_x,u_y)$ be a direction such that $\mathfrak q_a + \R_{\ge 0} \vec
  u \subset \Sigma$. Note that $u_y$ can be chosen to be positive.

  \bigskip

  For any $\alpha \in \R$ and $\beta >0$, the transform
  $f_{\alpha,\beta}: (x,y) \mapsto (x+\alpha y, \beta y)$ preserves
  orientations and leaves the $x$-axis fixed. Moreover, for every
  point $(x_0,y_0)$ with $y_0>0$, there exists $\alpha \in \R$ such
  that $f_{\alpha,\frac1{y_0}}(x_0,y_0) = (0,1)$ (namely, $\alpha= - \tfrac{x_0}{y_0}$).

  \bigskip

  Now, let $t \in \R_{\ge 0}$ and consider the point $\mathfrak p(t) \eqdef \mathfrak q_a
  + t \vec u$. Let $\alpha_t$ and $\beta_t$ be such that
  $f_{\alpha_t,\beta_t}(\mathfrak p(t)) = (0,1)$ and consider the point set
  \[ \cal
  P(t) \eqdef \{f_{\alpha_t,\beta_t}(\mathfrak q) \colon \mathfrak q \in  \{\mathfrak p(t)\} \cup \cal Q \setminus \{\mathfrak q_a\}\}\]
  still labeled by $X$ (with $\mathfrak p(t)$ labeled by $a$ and all other labels unchanged). Since $\mathfrak p(t)$ and $\mathfrak q_a$ are both in $\Sigma$,
  $\{\mathfrak p(t)\} \cup \cal Q \setminus \{\mathfrak q_a\}$ has the same chirotope as
  $\cal Q$, that is $\chi$. Since
  $f_{\alpha_t,\beta_t}$ preserves orientations, $\cal P(t)$ is also a
  geometric realization of $\chi$.

  \bigskip

  It is easy to check that $(\alpha_t,\beta_t) \to (-u_x/u_y,0)$ as $t
  \to \infty$. Consequently, the $y$-coordinate of every point of
  $\cal P(t) \setminus \{(0,1)\}$ goes to $0$ as $t\to \infty$. Also, the
  slope of every line going through two points of $\cal P(t) \setminus
  \{(0,1)\}$ goes to $0$ as $t\to \infty$. Hence, for every
  $\delta>0$, for $t$ large enough, $\cal P(t)$ is a
  $\delta$-squeezing of $(\chi,a)$.
\end{proof}
\bigskip

We now use $\delta$-squeezings to construct larger realizable chirotopes out of
smaller ones.
\begin{enumerate}
\item We start with two realizable chirotopes $\chi$ on $X\cup\{x_1^*,x_2^*\}$ and $\xi$ on $Y \cup \{y^*\}$, where $X$ and $Y$ are disjoint and $x_1^*x_2^*$ is not an extreme edge\footnote{An extreme edge is an edge between two consecutive points
on the convex hull.} of $\chi$, while $y^*$ is extreme in $\xi$.
\item We pick some realization $\cal P$ of $\chi$ with $\mathfrak p_{x_1^*} = (-1,0)$ and $\mathfrak p_{x_2^*} = (1,0)$, and some $\delta$-squeezing $\cal Q$ of $(\xi,y^*)$, for some $\delta>0$ small enough.
\item We let $\kappa$ denote the chirotope on $X \cup Y$ realized by the point set $({\cal P} \setminus \{\mathfrak p_{x_1^*},\mathfrak p_{x_1^*}\})\cup ({\cal Q} \setminus \{{q}_{y^*}\})$.
\end{enumerate}
For $\delta>0$ small enough, $Y$ is a quasi-module for $\kappa$ and
(the labels of) $(-1,0)$ and $(1,0)$ are its antipodal elements. We
assume this is the case from now on. Let us stress that the chirotope
$\kappa$ depends not only on the chirotopes $\chi$ and $\xi$, but on
the actual realizations $\cal P$ and $\cal Q$ (see
Figure~\ref{fig:squeeze et realization}). We refer to the above
construction as {\em substituting $\xi$ into $\chi$ along the segment
  $x_1^*x_2^*$ squeezing $\xi$ orthogonally to $y^*$}.

\begin{figure}
  \begin{center}
    \includegraphics{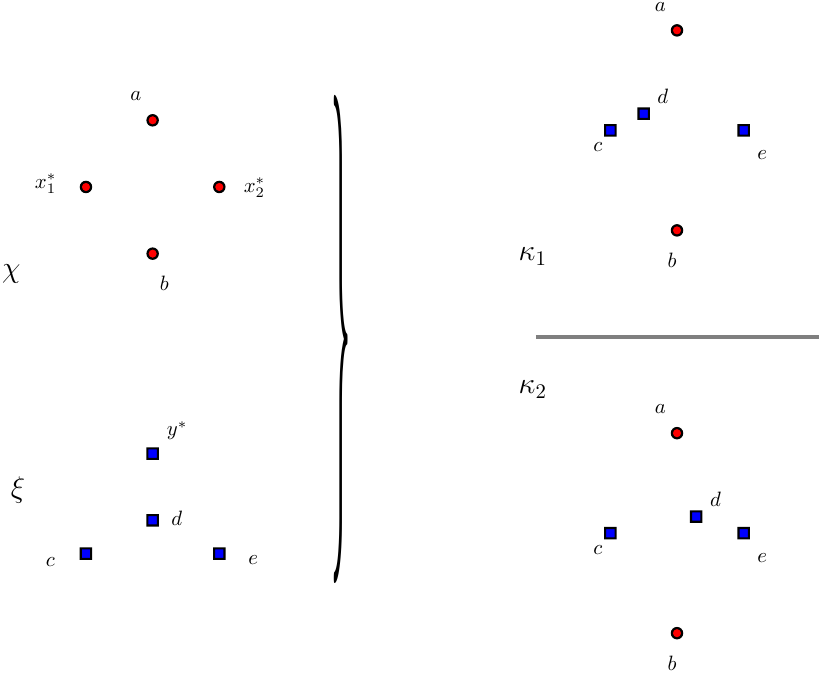}
\end{center}
\caption{Two chirotopes $\kappa_1$ and $\kappa_2$ obtained from the same substitution (substituting $\xi$ into $\chi$ along the segment $x_1^*x_2^*$ squeezing $\xi$ orthogonally to $y^*$), but using different realizations. Observe that $d$ is not on the same side of the line $ab$ in $\kappa_1$ and $\kappa_2$. \label{fig:squeeze et realization}}
\end{figure}

\begin{lemma}\label{lem:at most two}
  For any realizable chirotope $\kappa$ on $Z$ and any quasi-module $Y$ of
  $\kappa$, there exist a single chirotope $\chi$ on $(Z\setminus
  Y)\cup\{x_1^*,x_2^*\}$ and at most two chirotopes $\xi$ on $Y \cup
  \{y^*\}$ such that $\kappa$ can be obtained by the above
  construction (up to renaming the symbols $x_1^*$, $x_2^*$ and
  $y^*$).
\end{lemma}
\begin{proof}
  Suppose $\kappa$ and $Y$ are obtained by the above construction. Let
  $\{x_1^*, x_2^*\}$ denote the antipodal elements of $Y$, defined through \cref{lem:extreme de weak}.
  Then, $\chi$
  must coincide with the restriction of $\kappa$ to $(Z \setminus Y)
  \cup \{x_1^*,x_2^*\}$, so it is unique (up to renaming the symbols
  $x_1^*$ and $x_2^*$).

  Since $x_1^*,x_2^* \in Y$ and $Y$ is a quasi-module of $\kappa$, this
  pair defines a nontrivial bipartition, so there exist $y^*_+$ and
  $y^*_-$ in $(Z \setminus Y)$ such that $\chi(x_1^*,x_2^*,y^*_-) =
  -1$ and $\chi(x_1^*,x_2^*,y^*_+) = +1$. By construction, $\xi$ must
  coincide with the restriction of $\kappa$ to $Y \cup \{y^*_-\}$ or
  to $Y \cup \{y^*_+\}$. There is therefore at most two choices for
  the chirotope $\xi$ on $Y \cup \{y^*\}$ (up to renaming the symbol~$y^*$).
\end{proof}

\bigskip

We now use this construction for counting purposes.  Let $w_{n,k}$
denote the number of pairs $(\chi, W)$ where $\chi$ is a realizable chirotope on
$[n]$ and $W$ is a quasi-module of $\chi$ of size $k$. Recall that
$t_n^*$ denotes the number of realizable chirotopes on $[n]$ in which $n$ is
extreme. Let $t_n^{\div}$ denote the number of realizable chirotopes on $[n]$ in
which $n$ and $n-1$ do not form an extreme edge.

\begin{lemma}\label{lem:wtdiv}
  For any integers $n_1 \ge 2$, $n_2 \ge 2$, we have $w_{n_1+n_2,n_2} \ge \frac12 \binom{n_1+n_2}{n_1} t_{n_1+2}^{\div} t_{n_2+1}^*$.
\end{lemma}
\begin{proof}
  Let $n = n_1 + n_2$. Start with a triple $(I,\chi,\xi)$ where $I
  \subset [n]$ has size $n_1$, $\chi$ is a realizable chirotope on $[n_1+2]$ such
  that $n_1+1$ and $n_1+2$ do not form an extreme edge, and $\xi$ is a realizable
  chirotope on $[n_2+1]$ in which $n_2+1$ is extreme. With our
  notations, the number of such triples is $\binom{n_1+n_2}{n_1}
  t_{n_1+2}^{\div} t_{n_2+1}^*$.

  \bigskip

  Let $x_1^*, x_2^*$ and $y^*$ be formal symbols not in $[n]$. For
  each such triple $(I,\chi,\xi)$, we do the following operation:
  \begin{itemize}
  \item We let $\tilde \chi$ denote the chirotope on $I \cup \{x_1^*,x_2^*\}$ obtained from $\chi$ by relabeling $[n_1]$ increasingly into $I$ and mapping $n_1+1 \mapsto x_1^*$ and $n_1+2 \mapsto x_2^*$.
  \item We let $\tilde \xi$ denote the chirotope on $([n] \setminus I) \cup \{y^*\}$ obtained from $\xi$ by relabeling $[n_2]$ increasingly into $[n]\setminus I$ and mapping $n_2+1 \mapsto y^*$.
  \item We substitute $\tilde \xi$ into $\tilde \chi$ along the segment $x_1^*x_2^*$ squeezing $\tilde \xi$ orthogonally to $y^*$.
  \end{itemize}
  The result of this operation is a pair $(\kappa,W)$ where $\kappa$
  is a realizable chirotope on $[n]$ and $W$ is a quasi-module. By
  Lemma~\ref{lem:at most two}, each pair $(\kappa,W)$ is the image of
  at most two triples $(I,\chi,\xi)$, which concludes the proof.
\end{proof}

We can now complete the proof of Lemma~\ref{lem:bounding-d-large}, that is, that for $n \ge \frac3{\varepsilon}$,
\[ t_n\geq \frac{3K\varepsilon^7}{4}\, n^3 \, d^{\rm large}_n \qquad \text{where} \qquad  d^{\rm large}_n=\sum_{\substack{n_1+n_2=n \\ \varepsilon n \le n_1,n_2 \le (1-\varepsilon) n}} \binom{n}{n_1}\, \frac{t_{n_1+1}\, t_{n_2+1}}{(n_1+1)\,(n_2+1)}. \]

\begin{proof}[Proof of Lemma~\ref{lem:bounding-d-large}]
  Lemma~\ref{l:tn et tn*} gives that $t_n^* \ge \frac3n
  t_n$. Furthermore, adapting the proof of Lemma~\ref{l:tn et tn*}
  yields
  \[ \frac{t_n^{\div}}{t_n} = \frac2{n(n-1)} \Ex[\nee(\chi)]\]
  where $\nee(\chi)$ denotes the number of pairs of points of $\chi$
  which are not extreme edges.
 For $n \ge 5$, we have
  \[ \Ex[\nee(\chi)] \ge \binom{n}2 - n \ge \frac12\binom{n}2\]
  and hence $t_n^{\div} \ge \frac12 t_n$. Now, by Lemmas~\ref{lem:counting_weak_modules} and~\ref{lem:wtdiv} we have
  \[\begin{aligned}
  \frac{n^2}{\varepsilon} t_n \ge \sum_{\substack{n_1+n_2=n \\ \varepsilon n \le n_1,n_2 \le (1-\varepsilon) n}} w_{n,n_2} & \ge  \sum_{\substack{n_1+n_2=n \\ \varepsilon n \le n_1,n_2 \le (1-\varepsilon) n}}\frac12 \binom{n}{n_1} t_{n_1+2}^{\div} t_{n_2+1}^*\\
  &\ge \sum_{\substack{n_1+n_2=n \\ \varepsilon n \le n_1,n_2 \le (1-\varepsilon) n}}\frac3{4(n_2+1)} \binom{n}{n_1} t_{n_1+2} t_{n_2+1}.
  \end{aligned}
  \]
  (The inequality $t_{n_1+2}^\div \ge \frac12t_{n_1+2}$ requires $n_1
  + 2 \ge 5$, which is ensured by the assumption $n \ge
  \frac{3}{\varepsilon}$.) By Lemma~\ref{lem:bound-extensions}, there
  exists a universal constant $K>0$ such that, for all $n_1 \ge 2$, we
  have $t_{n_1+2} \ge K \, (n_1+1)^4 \, t_{n_1+1}$. Substituting into
  the previous inequality and using that $n_1 +1 \ge \varepsilon n$, we obtain as desired
  \begin{align*}
  t_n & \ge \frac{\varepsilon}{n^2} \frac{3K}4 \sum_{\substack{n_1+n_2=n \\ \varepsilon n \le n_1,n_2 \le (1-\varepsilon) n}} \binom{n}{n_1} (n_1+1)^5\frac{t_{n_1+1} t_{n_2+1}}{(n_1+1)(n_2+1)}\\
  & \ge \varepsilon^6n^3 \frac{3K}4 \sum_{\substack{n_1+n_2=n \\ \varepsilon n \le n_1,n_2 \le (1-\varepsilon) n}} \binom{n}{n_1} \frac{t_{n_1+1} t_{n_2+1}}{(n_1+1)(n_2+1)}.
  \qedhere  \end{align*}
\end{proof}

\section{A bijection to decompose triangulations through one bowtie product}
\label{sec:triangulations}

In this section, we consider a decomposable chirotope
$\kappa=\chi\join{x^*}{y^*}\xi$, and relate the triangulations of
$\kappa$ to those of $\chi$ and $\xi$. We show that each triangulation
of $\kappa$ projects into triangulations of $\chi$ and $\xi$, and that
conversely every pair of triangulations of $\chi$ and $\xi$ can be
obtained as projections of (possibly many) triangulations of $\kappa$.
We establish a bijection between these pre-images and
maximal\footnote{In what follows, inclusion-maximality is equivalent
to being of maximal cardinality.}  sets of noncrossing edges between
appropriate sets.

{\em Note:} We recall from the introduction that the notion of crossing edges
can be read on the chirotope, giving meaning to the notion that $xy$ and $zt$ cross in a 
(possibly abstract) chirotope $\chi$, if $x$,$y$,$z$ and $t$ are elements of the ground set of $\chi$.
If there is no ambiguity on which chirotope $\chi$ we are considering, we shall simply say that $xy$ and $zt$ cross.

\subsection{Statement of the bijection}

In the whole section, we let $\kappa=\chi\join{x^*}{y^*}\xi$ where
$\chi$ and $\xi$ are chirotopes on $X\cup\{x^*\}$ and $Y\cup\{y^*\}$,
with $X$ and $Y$ disjoint, $x^* \notin X$ being extreme in $\chi$ and
$y^* \notin Y$ being extreme in $\xi$. As a preamble, note that for
$x_1,x_2,x_3,x_4$ in $X$, the pairs $x_1x_2$ and $x_3x_4$ cross in
$\chi$ if and only if they cross in $\kappa$, so we shall simply say
that $x_1x_2$ and $x_3x_4$ cross, without making the underlying
chirotope explicit.

\medskip

Let $\pi_{Y \to x^*}$ denote the map that sends every element of $Y$ to
$x^*$. We extend this map to sets of edges by putting, for any
triangulation $T$ of $\kappa$,
\[ \pi_{Y \to x^*}(T) \eqdef \{\{\pi_{Y  \to x^*}(a),\pi_{Y \to x^*}(b)\} \colon \{a,b\} \in T\} \setminus \{\{x^*,x^*\}\}.\]
We define $\pi_{X \to y^*}$ analogously. Given a set $T$ of edges of
$\chi\join{x^*}{y^*}\xi$, not necessarily a triangulation, and two
sets $A$ and $B$, we let $T_{AB}$ denote the set of edges of $T$ with
one element in $A$ and the other in $B$; {we say that such an edge is
  {\em between $A$ and $B$}}. As illustrated in
Figure~\ref{fig:decomposition_triangulation}, any triangulation $T$ of
$\kappa$ decomposes into triangulations of $\chi$ and $\xi$ and a set
$T_{XY}$ of noncrossing edges. The aim of this section is to prove
that this decomposition is a bijection.
\begin{figure}[h]
\includegraphics[scale=.75]{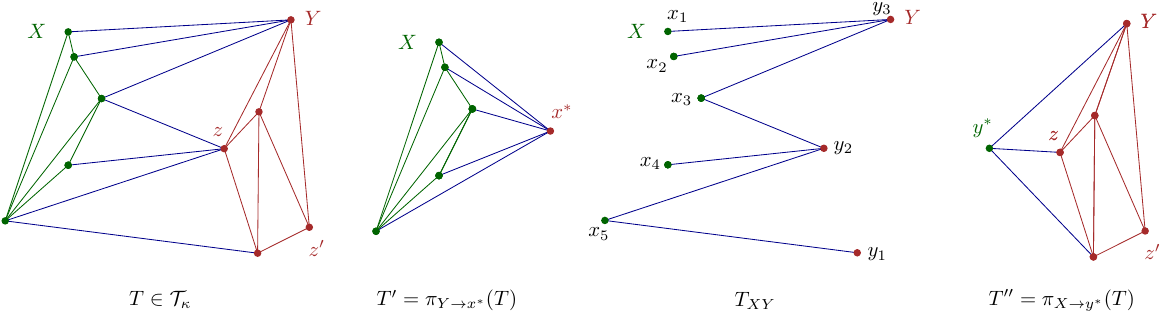}
\caption{A triangulation $T$ of $\kappa=\chi\join{x^*}{y^*}\xi$, the
  projection of $T$ on $X\cup\{x^*\}$, the set $T_{XY}$, and the
  projection of $T$ on
  $Y\cup\{y^*\}$.\label{fig:decomposition_triangulation}}
\end{figure}

Recall that ${\cal T}_\alpha$ denotes the set of triangulations of a
chirotope $\alpha$. Given $T \in {\cal T}_\alpha$ and an element $x$
of the ground set of $\alpha$, we write $N_T(x)$ for the neighborhood
of $x$ in $T$ and $\deg_{T}(x)$ for the degree of $x$ in $T$.  We
extend the notation $N_T(x)$ and $\deg_{T}(x)$ in the obvious manner
when $T$ is a (not necessarily maximal) set of noncrossing edges.

\begin{proposition}\label{p:bij}
  For any triangulation $T$ of $\kappa$, the sets $T' \eqdef \pi_{Y
    \to x^*}(T)$ and $T'' \eqdef \pi_{X \to y^*}(T)$ are
  triangulations of $\chi$ and $\xi$, respectively, and $T_{XY}$ is a
  {maximal} set of noncrossing edges between $N_{T'}(x^*)$ and
  $N_{T''}(y^*)$.

  \smallskip

  Conversely, for any triangulations $T'$ of $\chi$ and $T''$ of $\xi$
  and any maximal set $H$ of noncrossing edges between $N_{T'}(x^*)$
  and $N_{T''}(y^*)$, there exists a unique triangulation $T \in
  {\T}_{\kappa}$ such that $T' = \pi_{Y \to x^*}(T)$, $T'' = \pi_{X
    \to y^*}(T)$ and $H = T_{XY}$.
\end{proposition}

The proof of the bijection studies in particular the geometric
structure of the sets of noncrossing edges $T_{XY}$ for any
triangulation of $\kappa$; from this structure we derive the following
combinatorial lemma that will be useful for counting
triangulations. (We use the convention $\binom{n}{-1}=\delta_{n,-1}$
for any integer $n$, where $\delta$ is the usual Kronecker symbol.)

\begin{lemma}\label{lem:concave_k}
  Let $X'\subseteq X$, $Y'\subseteq Y$ be two nonempty subsets of
  elements of $\kappa$, let $0 \leq k \leq |Y'|$, and let
  $z_1,z_2,\ldots, z_k$ be distinct elements of $Y'$. Then for
  $i_1,\ldots, i_k\geq 1$ there are
  exactly \[\binom{|X'|+|Y'|-(i_1+\ldots+i_k)-2}{|Y'|-(k+1)}\] maximal
  sets $H$ of noncrossing edges {between $X'$ and $Y'$} and satisfying
  $\deg_{H}(z_j)=i_j$ for every $j\in[k]$.
\end{lemma}

In particular, for $k=0$, we obtain that for any nonempty subsets  $X'\subseteq X$ and $Y'\subseteq Y$ there are exactly  $\binom{|X'|+|Y'|-2}{|Y'|-1}$ {maximal} sets $H$ of noncrossing edges in $\kappa$.

\medskip

The rest of this section is devoted to proving \cref{lem:concave_k} and \cref{p:bij}.

\subsection{Set of noncrossing edges between $X$ and $Y$}
\label{s:radial}

For a chirotope $\alpha$ on a set $A$, we write $\ext(\alpha)$
for the set of elements of $A$ that are extreme in $\alpha$. 
As we saw in \cref{lem:order}, the relation defined by 
\[ \text{for } a,b,c \in A, \ b<_{\alpha,a} c \Leftrightarrow \alpha(a,b,c)=1\]
defines a strict total order on $X\setminus\{a\}$ if $a$ is extreme for $\alpha$.
This order allows to characterize crossings of edges between 
the modules $X$ and $Y$ of $\kappa$.

\begin{lemma}\label{l:cross}
 For any $x_1,x_2 \in X$ and $y_1,y_2 \in Y$ with $y_1
    <_{\xi,y^*} y_2$, the edges $x_1y_1$ and $x_2y_2$ cross in
    $\chi\join{x^*}{y^*}\xi$ if and only if $x_1 <_{\chi,x^*}
    x_2$.
\end{lemma}
\begin{proof}
 By definition, we have
  $\chi(x^*,x_1,x_2)=\kappa(x_2,y_2,x_1)=-\kappa(x_1,y_1,x_2)$ and
  similarly, we have
  $\xi(y^*,y_1,y_2)=-\kappa(x_2,y_2,y_1)=\kappa(x_1,y_1,y_2)$. 
  We also have $\xi(y^*,y_1,y_2)=1$ by hypothesis. 
  Hence if $x_1 <_{\chi,x^*} x_2$,
  then $\kappa(x_2,y_2,x_1)=1=-\kappa(x_2,y_2,y_1)$ and
  $\kappa(x_1,y_1,x_2)=-1=-\kappa(x_1,y_1,y_2)$, so that $x_1y_1$ and
  $x_2y_2$ cross in $\kappa$. Conversely, if $x_1y_1$ and $x_2y_2$
  cross, in particular $\kappa(x_1,y_1,y_2)=-\kappa(x_1,y_1,x_2)$,
  hence $1=\xi(y^*,y_1,y_2)= -\chi(x_1,x^*,x_2)=\chi(x^*,x_1,x_2)$ and
  $x_1 <_{\chi,x^*} x_2$. This proves the lemma.
  \end{proof}

\begin{lemma}\label{prop:bijection_triang_concave}
  Let $X'\subseteq X$, $Y'\subseteq Y$ be two nonempty subsets of
  elements of $\kappa$. Let us write $a=|X'|$, $b=|Y'|$, and
  $y_1<_{\xi,y^*}\ldots <_{\xi,y^*} y_b$ the elements of $Y'$. The map
  sending a set $H$ of edges to $(\deg_{H}(y_1),\ldots, \deg_{H}(y_b))$ is a bijection
  between maximal sets of noncrossing edges between $X'$ and $Y'$
 and sequences $(d_1,\ldots,d_b) \in
  \pth{{\N}_{\ge 1}}^b$ satisfying $d_1+\ldots +d_b=a+b-1$.%
\end{lemma}

For example the set of noncrossing edges of the third picture in \Cref{fig:decomposition_triangulation} is associated with the triple $(1,3,3)$.

\begin{proof}
  Let us write $x_1<_{\chi,x^*}\ldots <_{\chi,x^*} x_a$ the elements
  of $X'$, ordered by $<_{\chi,x^*}$. Let $\Lambda_{a,b}$ denote the
  set of sequences $(d_1,\ldots,d_b) \in \pth{{\N}_{\ge 1}}^b$
  satisfying $d_1+\ldots +d_b=a+b-1$.
  \medskip

  Let $H$ be a maximal set of noncrossing edges between $X'$ and $Y'$. For $j\in [b]$, we let ${\cal I}^H_j$
  denote the set of indices $i$ such that $x_iy_j$ is an edge of
  $H$.  
  By \cref{l:cross}, for $j<j'$, elements of ${\cal I}^H_{j'}$
  are smaller than or equal to elements of ${\cal I}^H_{j}$.
   But, using again that $H$ is maximal,
   we know that every element of $X'$ or $Y'$ is the endpoint of an edge,
   so that each ${\cal I}^H_{j}$ is a non-empty interval, say $[m_j,M_j]$.
   Moreover, consecutive ${\cal I}^H_{j}$ must have a
  common endpoint (otherwise $H$ would not be maximal).
 In other words, we have
  \[ 1 = m_b \le M_b = m_{b-1} \le M_{b-1} = \ldots = m_2 \le M_2 =
  m_1 \le M_1 = a.\]
  Since $\deg_{H}(y_j) = M_j-m_j+1 = m_{j-1}-m_j+1$ (with the convention $m_0=a$), it follows that
  the map sending $H$ to $(\deg_{H}(y_1),\ldots, \deg_{H}(y_b))$ 
  takes value in $\Lambda_{a,b}$.

Conversely, it is clear that the values $\deg_{H}(y_j) = M_j-m_j+1$
  uniquely determine the tuple $(m_1,\cdots,m_b)$, which in turn uniquely determines
  $({\cal I}^H_1,\cdots,{\cal I}^H_b)$ (recall that $M_1=a$ and $M_{j} = m_{j-1}$ for $2 \le j \le b$),
  which determines $H$.
  Thus the map $H \mapsto (\deg_{H}(y_1),\ldots, \deg_{H}(y_b))$ is a bijection.
\end{proof}

We can now prove \cref{lem:concave_k}.

  \begin{proof}[Proof of \cref{lem:concave_k}]
As before, let $a=|X'|$, $b=|Y'|$, and let $y_1<_{\xi,y^*}\ldots
<_{\xi,y^*} y_b$ be the elements of $Y'$.  We write $z_j=y_{h_j}$ for
some index $h_j$.  By \cref{prop:bijection_triang_concave}, we need to
count the number of sequences $(d_1,\ldots,d_b)$ satisfying
$d_1+\ldots +d_b=a+b-1$ with $d_i\in\mathbb Z_{\geq 1}$, and
$d_{h_j}=i_j$ for $j \le k$ {(we allow $k=0$)}.  These sequences are
in bijection with the sequences $(n_1,\ldots,n_{b-k})$ of $b-k$
positive integers summing to $a+b-1-(i_1+\ldots+i_k)$.  They are
sometimes called compositions of $a+b-1-(i_1+\ldots+i_k)$ into $b-k$
parts, and are known to be counted by
$\binom{a+b-(i_1+\ldots+i_k)-2}{b-k-1}$. (With the convention
$\binom{n}{-1}=\delta_{n,-1}$, the formula is also correct when
$k=b$.)
\end{proof}

\subsection{Proof of the bijection}\label{sec:proof71}

The goal of this section is to prove \cref{p:bij}.  We use the
notation in the first part of its statement, namely: $T$ is a
triangulation of $\kappa=\chi\join{x^*}{y^*}\xi$, where $\chi$ and
$\xi$ are chirotopes on $X\cup\{x^*\}$ and $Y\cup\{y^*\}$, with $X$
and $Y$ disjoint, $x^* \notin X$ being extreme in $\chi$ and $y^*
\notin Y$ being extreme in $\xi$.  We also recall that $T' = \pi_{Y
  \to x^*}(T)$ and $T'' = \pi_{X \to y^*}(T)$.  Additionally we set
$X'=N_{T'}(x^*)$ and $Y'= N_{T''}(y^*)$.  The set $T_{XY}$ is a subset
of $T$ and hence noncrossing by definition.  We start by showing that
the same is true for $\pi_{Y \to x^*}(T)$ (and symmetrically $\pi_{X
  \to y^*}(T)$).
\begin{lemma}\label{lem:bij-easy}
 $\pi_{Y \to x^*}(T)$ (resp.~$\pi_{X \to y^*}(T)$) is a set of noncrossing edges.
\end{lemma}
\begin{proof}
 Recall that whether two edges cross
  or not in a chirotope $\alpha$ depends only on the restriction of the chirotope
  $\alpha$ to the four endpoints of these two edges. As $\chi$ and $\kappa$ coincide on $X$, no two edges of
  $T_{XX}$ can cross in $T'$. Furthermore, by definition of $\pi_{Y\to
    x^*}$ and $\kappa$, the existence of two crossing edges $x^*x_1$
  and $x_2x_3$ of $T'$ would lead to the existence of an
  edge $y_1x_1\in T$ crossing the edge $x_2x_3\in T$,
  for some $y_1\in Y$. This is impossible since $T$ is a set of
  noncrossing edges in $\kappa$. This proves the lemma, the case of $\pi_{X \to y^*}(T)$
  being symmetrical.
\end{proof}

We still need to prove that $T_{XY}$, $\pi_{Y \to x^*}(T)$ and $\pi_{Y \to x^*}(T)$ 
are maximal. For this, we use cardinality arguments.
As a consequence of \cref{prop:bijection_triang_concave}, any maximal set 
of noncrossing edges between $X'$ and $Y'$ has size $|X'|+|Y'|-1$.
In particular
\begin{equation}\label{eq:card_TXY}
|T_{XY}| \le |X'|+|Y'|-1\,.
\end{equation}

Besides, a classical consequence of Euler's formula, via double counting, is
that any triangulation $\Delta$ of a realizable chirotope of size $n$ with
$h$ extreme elements has exactly
\begin{equation}\label{eq:const_edges_triangulation}
	|\Delta|=3n-h-3
\end{equation}
edges. In particular, any family of noncrossing edges of a realizable
chirotope of size $n$ with $h$ extreme elements has at most $3n-h-3$
edges. \cref{eq:const_edges_triangulation}
also holds for abstract chirotopes since every
abstract chirotope can be realized by a {\em generalized
  configuration}, that is a set of points in the plane together with
an arrangement of pseudolines, each one containing exactly two
points~\cite[Theorems~5.2.6]{felsner2017pseudoline}.

Let us come back to the setting of \cref{p:bij},
and let $h_X$ (resp. $h_Y$) denote the number of elements of $X$
(resp. $Y$) that are extreme in $\kappa$. As $T$ is a triangulation of
$\kappa$, we have
\[ |T| = 3(|X|+|Y|)-(h_X+h_Y)-3.\]
By Lemma~\ref{lem:bij-easy}, $T'=\pi_{Y \to x^*}(T)$ and $T''=\pi_{X \to y^*}(T)$ are families of noncrossing
edges of $\chi$ and $\xi$, hence they have at most as many edges as triangulations of $\chi$ and $\xi$, i.e., using
\cref{eq:const_edges_triangulation}, we have
\[ |T'| \le 3(|X|+1)-(h_X+1) - 3 = 3|X|-h_X-1 \qquad \text{and} \qquad |T''| \le 3|Y|-h_Y-1.\]

By definition of $T'$, $T''$, $X'$ and $Y'$, we further have
\[|T'| = |T_{XX}| + \deg_{T'}(x^*) = |T_{XX}| + |X'|,\quad  |T''| = |T_{YY}| + |Y'|\,.\]
Hence, recalling \cref{eq:card_TXY}, we have the following inequality
\[|T| = |T_{XX}|+|T_{YY}|+|T_{XY}| \le |T_{XX}|+|T_{YY}|+|X'|+|Y'| -1  = |T'| + |T''| -1.\]
so that
\[ |T|\le |T'|+|T''|-1 \le  3(|X|+|Y|)-(h_X+ h_Y)-3 = |T|.\]
This implies that every inequality we have proved so far is in fact an
equality. In particular, $T'$ (resp. $T''$) contains the maximal
possible number of noncrossing edges of $\chi$ (resp. $\xi$), ensuring
that it is a triangulation of $\chi$ (resp. $\xi$); and $T_{XY}$
is a {\em maximal} set of noncrossing edges between $X'$ and $Y'$.
This proves the first part of  \cref{p:bij}.
\medskip

It remains to prove that for any triangulations $T'$ and $T''$ of,
respectively, $\chi$ and $\xi$, and any maximal set $H$ of noncrossing
edges between $X'$ and $Y'$, there exists a unique triangulation $T
\in {\T}_{\kappa}$ such that $T' = \pi_{Y \to x^*}(T)$, $T'' = \pi_{X
  \to y^*}(T)$ and $H = T_{XY}$.

Let us first argue that such a triangulation $T$ is unique. Indeed, if
$\pi_{Y \to x^*}(T) = T$ then $T_{XX}=T'_{XX}$ and the edges of $T$
between $X$ and $X$ are determined. Similarly, the edges of $T$
between $Y$ and $Y$ must be determined by $T_{YY}=T''_{YY}$, and the edges
between $X$ and $Y$ must be determined by $T_{XY}=H$.

To prove the existence of $T$, we consider the set of edges $E\eqdef
T'_{XX}\cup T''_{YY}\, \cup\, H$. By construction, $E$ satisfies
$\pi_{Y \to x^*}(E)= T'$, $\pi_{X \to y^*}(E)=T''$ and
$E_{XY}=H$, so it suffices to prove that $E$ is a
triangulation. We first argue that the size of $E$ satisfies the
\cref{eq:const_edges_triangulation}. Indeed, we have
\begin{align*}
  |E| & =  |T'_{XX}|+|T''_{YY}|+ |H|\\
  & = (|T'|-\deg_{T'}(x^*)) + (|T''|-\deg_{T''}(y^*)) + (|X'|+|Y'|-1) & \text{by~\cref{eq:card_TXY}}\\
  & =  (|T'|-|X'|) + (|T''|-|Y'|) + (|X'|+|Y'|-1)\\
  & = |T'| + |T''|-1. 
\end{align*}
Let $h_X$ (resp. $h_Y$) denote the number of elements of $X$
(resp. $Y$) that are extreme in $\kappa$. Since $T'$ and $T''$ are
triangulations of respectively $\chi$ and $\xi$, by
\cref{eq:const_edges_triangulation},
\[ |T'| = 3(|X|+1)-(h_X+1) - 3 = 3|X|-h_X-1 \qquad \text{and} \qquad |T''| = 3|Y|-h_Y-1\,.\]
We also have $|T'|=|T'_{XX}|+\deg_{T'}(x^*)=|T'_{XX}|+|X'|$ and
$|T''|=|T''_{YY}|+\deg_{T''}(y^*)=|T''_{YY}|+|Y'|$.  Therefore, 
\[|E|=3|X|-h_X-1+3|Y|-h_Y-1-1=3(|X|+|Y|)-(h_X+h_Y)-3\,.\]
We recognize \cref{eq:const_edges_triangulation}. It now remains to
argue that the edges in $E$ are pairwise noncrossing. By hypothesis,
the edges in $T'_{XX}\subseteq T'$ are pairwise noncrossing, and
similarly for the edges in $T''_{YY}$, or in $H$. An edge in $T'_{XX}$
cannot cross an edge in $T''_{YY}$ as $X$ and $Y$ are mutually
avoiding sets in $\kappa$. Finally, an edge of $H$ cannot cross in
$\kappa$ an edge of $T'_{XX}$ (resp. $T''_{YY}$ by symmetry): suppose
by contradiction that $x_1y_1\in H$ is crossing $x_2x_3\in T'_{XX}$ in
$\kappa$. Then it would imply that the edges $x_1x^*$ and $x_2x_3$ are
two crossing edges in $T'$, which is impossible as $T'$ is a
triangulation of $\chi$. Altogether, $E$ is a triangulation of
$\kappa$. This ends the proof of \cref{p:bij}.

\section{Counting triangulations}\label{sec:counting_triangulations}

We now explain how to count the triangulations of a chirotope
  presented by a chirotope tree from information on the triangulations
  of the chirotopes decorating the nodes.

\subsection{Counting the triangulations of a bowtie}

Let us first consider the case of a single bowtie $\kappa = \chi\join{x^*}{y^*}\xi$. 
\cref{p:bij} describes the {\em set} of triangulations of $\kappa$ in terms of the {\em sets} of triangulations of $\chi$ and $\xi$. With \cref{lem:concave_k}, however, we can {\em count} the former knowing only the generating polynomials of the latters marking the degrees of the proxies $x^*$ and $y^*$.

\medskip

Let $P_{\chi,x^*}(s)$ and $P_{\xi,y^*}(t)$ denote the generating
polynomials of the triangulations of $\chi$ and $\xi$ recording,
respectively, the degrees of the elements $x^*$ and $y^*$:
\[P_{\chi,x^*}(s)=\sum_{T'\in \mathcal{T}_{\chi}}s^{\deg_{T'}(x^*)} \qquad \text{and} \qquad P_{\xi,y^*}(t)=\sum_{T''\in \mathcal{T}_{\xi}}t^{\deg_{T''}(y^*)}.\]
Here are examples with $P_{\chi,x^*}(s)=s^3(s+1)$
and $P_{\xi,y^*}(t)=t^2(1+t+t^2)$.

\begin{center}
\tikzset{
  fourTri/.pic={
    \node (n1) at (0,.75) {$\bullet$};
    \node (n1bis) at (.15,0) {$\bullet$};
    \node (n2) at (-.15,0) {$\bullet$};
    \node (n3) at (-.6,-.6) {$\bullet$};
    \node[red] (n4) at (.6,-.6) {$\bullet$};
    \node[red] at ($(n4)+(0:0.3)$) {$x^*$};
}}
\tikzset{
  fourTrie/.pic={
    \pic {fourTri};
    \begin{scope}[on background layer]
      \draw[thin] (n2.center)--(n1.center);
      \draw[thin] (n2.center)--(n1bis.center);
      \draw[thin] (n2.center)--(n3.center);
      \draw[thin] (n1bis.center)--(n1.center);
      \draw[thin] (n1bis.center)--(n4.center);
      \draw[thin] (n4.center)--(n3.center);
      \draw[thin] (n1.center)--(n3.center);
      \draw[thin] (n4.center)--(n1.center);
    \end{scope}
  }
}
\tikzset{
  littlefourT/.pic={
    \node (n1) at (0,0.3) {$\bullet$};
    \node (n1bis) at (.5,.5) {$\bullet$};
    \node (n2) at (-.5,.5) {$\bullet$};
    \node (n3) at (-.5,-.5) {$\bullet$};
    \node[red] (n4) at (.5,-.5) {$\bullet$};
    \node[red] at ($(n4)+(0:0.3)$) {$y^*$};
}}
\tikzset{
  littlefourTe/.pic={
    \pic {littlefourT};
    \begin{scope}[on background layer]
      \draw[thin] (n2.center)--(n1.center);
      \draw[thin] (n2.center)--(n1bis.center);
      \draw[thin] (n2.center)--(n3.center);
      \draw[thin] (n1bis.center)--(n1.center);
      \draw[thin] (n1bis.center)--(n4.center);
      \draw[thin] (n4.center)--(n3.center);
    \end{scope}
  }
}
\begin{tikzpicture}[baseline={(0,0)}, scale=0.75, every node/.style={transform shape}]
  \pic (chi) {fourTri};
  \pic[xshift=2.5cm] (chi) {fourTrie};
  \pic[xshift=5cm] (chi2) {fourTrie};
  \node at (2.5,-1) {$s^3$};
  \node at (5,-1) {$s^4$};
  \node at (0,-1) {$\chi$};
  \begin{scope}[on background layer]
    \draw[thin] (chin3.center)--(chin1bis.center);
    \draw[thin] (chi2n4.center)--(chi2n2.center);
  \end{scope}
  \begin{scope}[xshift=8.5cm]
    \pic (chi0) {littlefourT};
    \pic[xshift=2cm] (xi) {littlefourTe};
    \pic[xshift=4cm] (xi2) {littlefourTe};
    \pic[xshift=6cm] (xi4) {littlefourTe};
    \node at (0,-1) {$\xi$};
    \node at (2,-1) {$t^2$};
    \node at (4,-1) {$t^4$};
    \node at (6,-1) {$t^3$};
      \begin{scope}[on background layer]
        \draw[thin] (xin3.center)--(xin1bis.center);
        \draw[thin] (xin3.center)--(xin1.center);
        \draw[thin] (xi2n4.center)--(xi2n1.center);
        \draw[thin] (xi2n4.center)--(xi2n2.center);
        \draw[thin] (xi4n4.center)--(xi4n1.center);
        \draw[thin] (xi4n3.center)--(xi4n1.center);
      \end{scope}
  \end{scope}
\end{tikzpicture}
\end{center}
%

\noindent
For $\kappa = \chi\join{x^*}{y^*}\xi$, 
using Proposition~\ref{p:bij} and  \cref{lem:concave_k}, $|{\cal T}_{\kappa}|$ 
can be expressed in terms of $P_{\chi,x^*}(s)$ and $P_{\xi,y^*}(t)$,
as explained in the following proposition.
\begin{proposition}\label{prop:triangulation_bowtie}
  Let $\kappa=\chi\join{x^*}{y^*}\xi$ where $\chi$ and $\xi$ are two chirotopes
  on disjoint ground sets $X\cup\{x^*\}$ and $Y\cup\{y^*\}$,
  respectively. The number of triangulations of $\kappa$ is
\begin{equation}\tag{\ref{eq:counting bowtie}}
    |\mathcal{T}_{\kappa}|=\sum_{\substack{T'\in \mathcal{T}_{\chi}\\T''\in \mathcal{T}_{\xi}}}\binom{\deg_{T'}(x^*)+\deg_{T''}(y^*)-2}{\deg_{T'}(x^*)-1}    = \sum_{a,b \ge 2} \binom{a+b-2}{a-1} [s^a]P_{\chi,x^*}(s) \  [t^b]P_{\xi,y^*}(t).
    \end{equation}
\end{proposition}
\begin{proof}
Let $\T_{\join{}{}}$ denote the set of triples $(T',T'', H)$ such that
$T'\in \T_{\chi}$, $T''\in \T_{\xi}$ and $H$ is a maximal set of
noncrossing edges between $X'\eqdef N_{T'}(x^*)$ and $Y'\eqdef N_{T''}(y^*)$. 
 By \cref{lem:concave_k} (case $k=0$), for
fixed $T'\in \T_{\chi}$ and $T''\in \T_{\xi}$, there are
$\binom{|X'|+|Y'|-2}{|X'|-1}$
sets $H$ such that $(T',T'',H)\in \T_{\join{}{}}$. 
Hence, recalling that $|X'|=\deg_{T'}(x^*)$ and $|Y'|=\deg_{T''}(y^*)$, we have
\[ |\T_{\join{}{}}|=\sum_{\substack{T'\in \mathcal{T}_{\chi}\\T''\in \mathcal{T}_{\xi}}}\binom{\deg_{T'}(x^*)+\deg_{T''}(y^*)-2}{\deg_{T'}(x^*)-1}.\]
But, from
Proposition~\ref{p:bij}, we know that $|\T_{\kappa}|=|\T_{\join{}{}}|$.
This proves the first equality in the proposition.
For the second equality, it suffices to group triangulations according to the
 degrees of $x^*$ and $y^*$ (which are at least $2$ in any triangulation).
\end{proof}

\subsection{Generating polynomial of the triangulations of a bowtie}
\label{ssec:refined_count_triangulations_bowtie}
We just used the bijection of Proposition~\ref{p:bij} to
express $|{\cal T}_{\kappa}|$ for $\kappa \eqdef
\chi\join{x^*}{y^*}\xi$ from $P_{\chi,x^*}(s)$ and $P_{\xi,y^*}(t)$
(see ~\cref{eq:counting bowtie}). To iterate this computation
through several bowtie products, we need to determine not only $|{\cal
  T}_{\kappa}|$, but $P_{\kappa, z^*}(t)$ for any extreme element $z^*$
of $\kappa$, to be used as a proxy element in a subsequent bowtie
product.

\medskip

Note that at this point, $\chi$ and $\xi$ do not play a
symmetric role anymore, as $z^*$ is an element of the ground sets of
only one of them. Without loss of generality, we suppose that $z^* \in
Y$, and we replace $P_{\xi,y^*}(t)$ by two bivariate polynomials:
\[Q^\in_{\xi,y^*,z^*}(t,u)=\sum_{\substack{T\in \mathcal{T}_{\xi}\\\{y^*,z^*\}\in T }}t^{\deg_{T}(y^*)}u^{\deg_{T}(z^*)} \quad\text{and}\quad Q^{\notin}_{\xi,y^*,z^*}(t,u)=\sum_{\substack{T\in \mathcal{T}_{\xi}\\\{y^*,z^*\}\notin T }}t^{\deg_{T}(y^*)}u^{\deg_{T}(z^*)}.\]
Observe that these polynomials refine the polynomials $P_{\xi,y^*}(t)$ defined earlier, in the sense that 
\begin{equation}\label{eq:specilizationQintoP}
 P_{\xi,y^*}(t) = Q^\in_{\xi,y^*,z^*}(t,1) + Q^{\notin}_{\xi,y^*,z^*}(t,1). 
\end{equation}
Coming back to our previous example for $(\xi,y^*)$, it has the following triangulations
\medskip
\begin{center}
\tikzset{
    fourT/.pic={
        \node (n1) at (0,0.3) {$\bullet$};
      \node (n1bis) at (.5,.5) {$\bullet$};
      \node[red] (n2) at (-.5,.5) {$\bullet$};
      \node (n3) at (-.5,-.5) {$\bullet$};
      \node[red] (n4) at (.5,-.5) {$\bullet$};
      \node[red] at ($(n2)-(0.3,0)$) {$z^*$};
      \node[red] at ($(n4)+(0:0.3)$) {$y^*$};
    }}
\tikzset{
    fourTe/.pic={
   	\pic {fourT};
    	\begin{scope}[on background layer]
      \draw[thin] (n2.center)--(n1.center);
      \draw[thin] (n2.center)--(n1bis.center);
      \draw[thin] (n2.center)--(n3.center);
      \draw[thin] (n1bis.center)--(n1.center);
      \draw[thin] (n1bis.center)--(n4.center);
      \draw[thin] (n4.center)--(n3.center);
      \end{scope}
    }
}
    \begin{tikzpicture}[baseline={(0,0)}, scale=1, every node/.style={transform shape}]
     \pic (chi0) {fourT};
      \pic[xshift=2.3cm] (chi) {fourTe};
      \pic[xshift=4.6cm] (chi2) {fourTe};
      \pic[xshift=6.9cm] (chi4) {fourTe};
\node at (0,-1) {$\xi$};
      \node at (2.3,-1) {$t^2u^3$};
      \node at (2.3,-1.5) {\footnotesize$\{y^*,z^*\}\notin T$};
      \node at (4.6,-1) {$t^4u^4$};
      \node at (4.6,-1.5) {\footnotesize$\{y^*,z^*\}\in T$};
       \node at (6.9,-1) {$t^3u^3$};
      \node at (6.9,-1.5) {\footnotesize$\{y^*,z^*\}\notin T$};
      \begin{scope}[on background layer]
      \draw[thin] (chin3.center)--(chin1bis.center);
      \draw[thin] (chin3.center)--(chin1.center);
      \draw[thin] (chi2n4.center)--(chi2n1.center);
      \draw[thin] (chi2n4.center)--(chi2n2.center);
      \draw[thin] (chi4n4.center)--(chi4n1.center);
      \draw[thin] (chi4n3.center)--(chi4n1.center);
      \end{scope}
    \end{tikzpicture}.
\end{center}


\noindent
Thus $Q^{\in}_{\xi,y^*,z^*}(t,u)=t^4u^4$ and
$Q^{\notin}_{\xi,y^*,z^*}(t,u)=t^2u^3(1+t)$. In the following
proposition, we express $P_{\kappa,z^*}(u)$ in terms of
$Q^{\in}_{\xi,y^*,z^*}(t,u)$, $Q^{\notin}_{\xi,y^*,z^*}(t,u)$ and
$P_{\chi,x^*}(s)$.\footnote{Recalling also
\cref{eq:specilizationQintoP}, the \cref{prop:triangulation_chaine}
suggests that, by iteration, knowing the polynomials $Q$ at each node
in a chirotope tree will be enough to compute the polynomial $P$ for
the chirotope described by the whole tree. This is actually true only
in the special case where the tree we consider is a path (as discussed
in \cref{sec:kernel-method}). Further refinements are necessary for
considering general trees, and these will be presented later, in
\cref{p:fullpol_recurrence}. }

\begin{proposition}\label{prop:triangulation_chaine}
  Let $\kappa=\chi\join{x^*}{y^*}\xi$ where $\chi$ and $\xi$ are
  chirotopes on $X\cup\{x^*\}$ and $Y\cup\{y^*\}$, with $X$ and $Y$
  disjoint, $x^* \notin X$ and $y^* \notin Y$. For any $z^* \in Y$
  extreme in $\xi$ we have
  \begin{align*}
  P_{\kappa,z^*}(u) = &\sum_{a,b \ge 2} \tbinom{a+b-2}{a-1} [s^a]P_{\chi,x^*}(s) [t^b]Q^{\notin}_{\xi,y^*,z^*}(t,u)\\
  & + \sum_{a,b \ge 2}  R_{a,b}(u)[s^a]P_{\chi,x^*}(s)\ [t^b]Q^{\in}_{\xi,y^*,z^*}(t,u), \text{ with }  R_{a,b}(u)=\sum_{i=0}^{a-1}\tbinom{a+b-i-3}{b-2}u^i.
  \end{align*}
\end{proposition}
\begin{proof} 
  The proof relies again on the bijection of Proposition~\ref{p:bij},
  and the reader is invited to refer to
  \Cref{fig:decomposition_triangulation}. As before, for a
  triangulation $T\in \T_\kappa$, we let $T'=\pi_{Y \to x^*}(T)$, $T''
  = \pi_{X \to y^*}(T)$ and $T_{XY}$ be the set of edges of $T$
  between $X$ and $Y$.  We let $\T_\kappa^\in$
  (resp.$\T_\kappa^{\notin}$) denote the set of triangulations of
  $\T_\kappa$ that contain (resp. do not contain) at least one edge
  between $z^*$ and an element of $X$.  Note that $T$ is in
  $\T_\kappa^\in$ if and only if the edge $\{y^*,z^*\}$ belongs to
  $T''$.
  
  \smallskip

  For any $T\in \T^{\notin}$, we have $\deg_{T}(z^*) =
  \deg_{T''}(z^*)$ as $\{y^*,z^*\} \notin T''$ so that the
  contribution of all triangulations of $\T^{\notin}$ to
  $P_{\kappa,z^*}(u)$ can then be computed by a straightforward
  adaptation of Proposition~\ref{prop:triangulation_bowtie}:
 \[\sum_{T\in \T^{\notin}_{\kappa}}u^{\deg_{T}(z^*)}=\sum_{a,b \ge 2} \binom{a+b-2}{a-1}  [s^a]P_{\chi,x^*}(s)\ [t^b]Q^{\notin}_{\xi,y^*,z^*}(t,u) \,.\]

 Now, assume that $T\in \T^{\in}$, so that $\{y^*,z^*\}\in T''$. The edges
 incident to $z^*$ in $T$ are exactly the union of those incident to
 $z^*$ in $T_{XY}$ and those incident to $z^*$ in $T''$, minus the
 edge $\{y^*,z^*\}$. Hence
 $\deg_T(z^*)=\deg_{T_{XY}}(z^*)+\deg_{T''}(z^*)-1$. Using \cref{lem:concave_k} (case $k=1$)  after regrouping the terms according to the
 value of $i=\deg_{T_{XY}}(z^*)$, we obtain:
 \[\sum_{T\in \T^{\in}_{\kappa}}u^{\deg_{T}(z^*)}=\sum_{\substack{T'\in \mathcal{T}_{\chi}\\T''\in \mathcal{T}_{\xi},\, \{y^*,z^*\}\in T''}}u^{\deg_{T''}(z^*)}\underbrace{\sum_{i=1}^{\deg_{T'}(x^*)}\binom{\deg_{T'}(x^*)+\deg_{T''}(y^*)-i-2}{\deg_{T''}(y^*)-2}  u^{i-1}}_{= R_{\deg_{T'}(x^*),\deg_{T''}(y^*)}(u)} \,.\]
 Grouping the terms by the values of $a=\deg_{T'}(x^*)\geq 2$ and
 $b=\deg_{T''}(y^*)\geq 2$, we finally obtain
 \[\sum_{T\in \T^{\in}_{\kappa}}u^{\deg_{T}(z^*)}=\sum_{a,b \ge 2}  R_{a,b}(u)[s^a]P_{\chi,x^*}(s)\ [t^b]Q^{\in}_{\xi,y^*,z^*}(t,u) \,,\]
 concluding the proof.
\end{proof}

\subsection{Example of application: triangulations of a family of chirotope paths}
\label{sec:kernel-method}

As claimed earlier, iterating \cref{prop:triangulation_chaine}, 
it is possible to compute the polynomial $P$ associated to a chirotope described by a chirotope tree that is a path, knowing the polynomials $Q$ at each node in this tree.
In this section, we show an application of this fact. 
Specifically, we study a family of chirotope trees,
where the underlying trees are simply paths (we call them {\em chirotope paths} for short),
and where the decorations are chosen among two chirotopes, symmetric of each other.
These leads, for each $k$, to $2^k$ different chirotopes (some but not all being equivalent by symmetry) with the same number of triangulations.
This number of triangulations can be evaluated explicitly,
combining \cref{prop:triangulation_chaine} with the so-called {\em kernel method}
 from analytic combinatorics (see, \emph{e.g.},~\cite{kernelmethod}).
\medskip

Consider the two following chirotopes on $\{x^*,y^*,z,c\}$:

\begin{center}
$\chi^{(0)}=$
\begin{tikzpicture}[baseline={(0,0)}, scale=0.8, every node/.style={transform shape}]
      \node (n1) at (0,0) {$\bullet$};
      \node (n2) at (0,.7) {$\bullet$};
      \node[red] (n3) at ($(n1)+(210:.7)$) {$\bullet$};
      \node[red] (n4) at ($(n1)+(-30:.7)$) {$\bullet$};
      \node at ($(n1)+(0.23,0)$) {$c$};
      \node at ($(n2)+(0.23,0)$) {$z$};
      \node[red] at ($(n3)+(180:0.3)$) {$x^*$};
      \node[red] at ($(n4)+(0:0.3)$) {$y^*$};
    \end{tikzpicture}
   \qquad \text{and} \qquad  $\chi^{(1)}=$
\begin{tikzpicture}[baseline={(0,0)}, scale=0.8, every node/.style={transform shape}]
	\node (n1) at (0,0) {$\bullet$};
      \node (n2) at (0,-.7) {$\bullet$};
      \node[red] (n3) at ($(n1)+(-210:.7)$) {$\bullet$};
      \node[red] (n4) at ($(n1)+(30:.7)$) {$\bullet$};
      \node at ($(n1)+(0.23,0)$) {$c$};
      \node at ($(n2)+(0.23,0)$) {$z$};
      \node[red] at ($(n3)+(180:0.3)$) {$x^*$};
      \node[red] at ($(n4)+(0:0.3)$) {$y^*$};
\end{tikzpicture}
\end{center}

\noindent
Each of $\chi^{(0)}$ and $\chi^{(1)}$ has a unique triangulation,
which contains the edge $\{x^*,y^*\}$ and where each vertex has
degree~$3$. Thus,
\[\begin{array}{ccccl}
P_{\chi^{(0)},y^*}(s) &=& P_{\chi^{(1)},y^*}(s) &=& s^3,\\
Q_{\chi^{(0)},x^*,y^*}^\in(t,u)&=&Q_{\chi^{(1)},x^*,y^*}^\in(t,u)&=&t^3u^3,\\
Q_{\chi^{(0)},x^*,y^*}^{\notin}(t,u)&=&Q_{\chi^{(1)},x^*,y^*}^{\notin}(t,u)&=&0.
\end{array}\]

\medskip

For $i \in \N$, let $\chi_i^{(0)}$ denote a copy of $\chi^{(0)}$
relabeled by $x^* \mapsto x_i^*$, $y^* \mapsto y_i^*$, $z \mapsto
z_i$, $c \mapsto c_i$. We similarly define $\chi_i^{(1)}$ to be a copy of
$\chi^{(1)}$ with the same relabeling.  For any word $\sigma =
\sigma_1\sigma_2\ldots \sigma_k \in \{0,1\}^k$ we put $\chi^{(\sigma)}
\eqdef \chi_1^{(\sigma_1)} \join{y_1^*}{x_2^*} \chi_2^{(\sigma_2)}
\join{y_2^*}{x_3^*} \ldots \join{y_{k-1}^*}{x_k^*}
\chi_k^{(\sigma_k)}$. This expression translates into a chirotope
tree, actually a path, for example for $\chi^{(01101)}$:

\begin{center}
\tikzset{
brick_chi1/.pic={
	\draw (0,0) circle(1);
	\begin{scope}[yshift=-0.35cm]
	\node (n1) at (0,0) {$\bullet$};
      \node (n2) at (0,-.45) {$\bullet$};
      \node[red] (nx) at ($(n1)+(-210:.7)$) {$\bullet$};
      \node[red] (ny) at ($(n1)+(30:.7)$) {$\bullet$};
      \node at ($(n1)+(0.23,0)$) {$c_{#1}$};
      \node at ($(n2)+(0.23,0)$) {$z_{#1}$};
      \node[red] at ($(ny)+(90:0.35)$) {$y_{#1}^*$};
      \node[red] at ($(nx)+(90:0.35)$) {$x_{#1}^*$};
      \end{scope}
}
}
\tikzset{
brick_chi0/.pic={
	\draw (0,0) circle(1);
\begin{scope}[yshift=0.35cm]
	\node (n1) at (0,0) {$\bullet$};
      \node (n2) at (0,.45) {$\bullet$};
      \node[red] (nx) at ($(n1)+(210:.7)$) {$\bullet$};
      \node[red] (ny) at ($(n1)+(-30:.7)$) {$\bullet$};
      \node at ($(n1)+(0.23,0)$) {$c_{#1}$};
      \node at ($(n2)+(0.23,0)$) {$z_{#1}$};
      \node[red] at ($(nx)+(-90:0.35)$) {$x_{#1}^*$};
      \node[red] at ($(ny)+(-90:0.35)$) {$y_{#1}^*$};
      \end{scope}
}}
\begin{tikzpicture}[scale=0.7,every node/.style={transform shape}]
	\pic (0) {brick_chi0=1};
	\pic [xshift=2.6cm] (1) {brick_chi1=2};
	\pic[xshift=5.2cm] (2) {brick_chi1=3};
	\pic[xshift=7.8cm] (3) {brick_chi0=4};
	\pic[xshift=10.4cm] (4) {brick_chi1=5};
	\draw[red] (0ny.center) -- (1nx.center);
	\draw[red] (1ny.center) -- (2nx.center);
	\draw[red] (2ny.center) -- (3nx.center);
	\draw[red] (3ny.center) -- (4nx.center);
\end{tikzpicture}
\end{center}

\noindent
The chirotopes $\chi^{(0)}$ and $\chi^{(1)}$ are indecomposable and nonconvex, so
these trees are canonical and Theorem~\ref{thm:uniqueness_canonical}
ensures that $\chi^{(\sigma)} = \chi^{(\sigma')}$ if and only if
$\sigma = \sigma'$.  Moreover, $\chi^{(0)}$ and $\chi^{(1)}$ are
realizable so $\chi^{(\sigma)}$ is realizable by
Proposition~\ref{p:whenischirotope}~(ii) and is the chirotope of a set
of $2k+2$ elements, $k$ of which are interior. It turns out that the
number of triangulations of $\chi^{(\sigma)}$ depends only on the
length of~$\sigma$, and there is a closed formula for this number.

\begin{proposition}\label{p:counting-chains}
  For every $\sigma \in \{0,1\}^k$ we have
  \[ \left|{\cal T}_{\chi^{(\sigma)}}\right| = \frac{3(2k+1){\binom{4k+2}{2k+1}
      }}{(2k+2)(4k+1)}-\frac {{4}^{k}{\binom{2k+2}{k+1}}}{2k+1}
  \sim_{k \to
    \infty}\frac{3-2\sqrt{2}}{\sqrt{2\pi}} \frac{ 16^{k}}{
    k^{3/2}}.\]
\end{proposition}
\noindent
Recall that for $\sigma \in \{0,1\}^k$, the chirotope $\chi^{(\sigma)}$ has $2k+2$ elements. In particular, up to a multiplicative constant, $\left|{\cal T}_{\chi^{(\sigma)}}\right|$ has the same asymptotics as the $2k$-th Catalan number, which is the number of triangulations of $2k+2$ points in convex position. 
\begin{proof}
Let $\tau \in \{0,1\}^\N$ be an infinite binary
word, let $\tau_k$ be its $k$th letter, and let $\tau^{[k]}$ denote
its prefix of length $k$. We introduce the shorthands
\[ P_k(u) \eqdef P_{\chi^{(\tau^{[k]})},y_k^*}(u) \quad \text{and} %
\quad Q(t,u) \eqdef Q^{\in}_{\chi_k^{(\tau_k)},x_k^*,y_k^*}(t,u) = t^3u^3 \]
and recall that $Q^{\notin}_{\chi_k^{(\tau_k)},x_k^*,y_k^*}(t,u) =0$.
Since $\chi^{(\tau^{[k+1]})} = \chi^{(\tau^{[k]})}
\join{y_{k}^*}{x_{k+1}^*} \chi_{k+1}^{(\tau_{k+1})}$,
we can apply
Proposition~\ref{prop:triangulation_chaine} with $z^*=y_{k+1}^*$
to get
\[  P_{k+1}(u) = \sum_{a,b \ge 2}  R_{a,b}(u) \, \ [s^a]P_{k}(s)\, \  [t^b]Q(t,u) =  u^3\sum_{a=0}^{2k+1}  R_{a,3}(u) \ [s^a]P_{k}(s).\]
(Note that $a$ can only range from $2$ to $2k+1$ as it is the degree
of a point in a triangulation of a set of $2k+2$ points.) Using the identity $\displaystyle
R_{a,3}(u)=\sum_{i=0}^{a-1}\binom{a-i}{1}u^i=\frac{u}{(1-u)^2}(u^a-1)+\frac{a}{1-u}$, we get
\[
P_{k+1}(u)=\frac{u^4}{(1-u)^2}\underbrace{\sum_{a=0}^{2k+1} [s^a]P_{k}(s)\, u^a}_{=\, P_k(u)}-\frac{u^4}{(1-u)^2}\underbrace{\sum_{a=0}^{2k+1}  [s^a]P_{k}(s)}_{=\, P_k(1)} +\frac{u^3}{(1-u)}\underbrace{\sum_{a=0}^{2k+1} a \, [s^a]P_{k}(s)}_{=\, P_k'(1)},\]
which rewrites as
\begin{equation}\label{eq:rec_chain_triangles}
P_{k+1}(s)-\frac{s^4}{(1-s)^2}P_{k}(s)=-\frac{s^4}{(1-s)^2}P_{k}(1)+\frac{s^3}{1-s}P'_{k}(1)\,.
\end{equation}

Let us introduce the formal series $F(s,u) \eqdef \sum_{k\geq
  1}P_{k}(s)u^k$. Note that by definition of $P_{k}(s)$,  its specialization 
  $F(1,u)= \sum_{k\geq 1} P_{k}(1)u^{k}$ is
the generating series of the sequence recording, for each $k$, the number of triangulations of $\chi^{(\tau^{[k]})}$.
Since $\chi^{(\tau^{[k]})}$ is realizable, we have $P_{k}(1) \le 30^{2k+2}$~\cite{countingT}.
Moreover, each $P_k$ has degree at most $2k+1$ and nonnegative coefficients,
so that, for $|s|\le 2$, we have $|P_k(s)| \le P_k(2) \le P_k(1)\cdot  2^{2k+1} \le (60)^{2k+2}$.
This implies that $F$ is analytic in the two variables $s$ and $u$
on the domain $$D=\{(s,u): |s| < 2, |u| < \tfrac1{60^2}\},$$ which contains in particular the point $(1,0)$.
 Now, multiplying \cref{eq:rec_chain_triangles} by
$u^{k+1}$ and summing up for $k \ge 1$, we obtain the following functional
equation on $D$:
\[\big(F(s,u)-us^3 \big) -\frac{s^4}{(1-s)^2} u F(s,u) = -\frac{s^4}{(1-s)^2}uF(1,u)+\frac{s^3}{1-s}u \partial_s
  F(1,u),\]
which rewrites as
\begin{equation}\label{eq:kernel_equation_chain}
  \left(1-\frac{us^4}{(1-s)^2}\right)F(s,u)=us^3\left(1-\frac{s}{(1-s)^2}F(1,u)+\frac{1}{1-s}\partial_s
  F(1,u)\right),
\end{equation}%
where, as usual,  $\partial_s F(s,u)$ denotes the derivative of $F$ according to
the first variable $s$.

\medskip

The right-hand side of \cref{eq:kernel_equation_chain} is
linear in $F(1,u)$ and $\partial_s F(1,u)$. Let us apply the standard {\em
  kernel method} (see, e.g., \cite{kernelmethod}): since $F(s,u)$ is analytic near $(1,0)$, we can
evaluate \cref{eq:kernel_equation_chain} for some $s(u)$
that cancels the coefficient $1-\frac{us^4}{(1-s)^2}$, called the {\em
  kernel}, and such that $s(u)$ is close to $1$ when $u$ is close to
$0$. Two functions satisfy these conditions:
\[ s_{1}(u)=\frac{-1+ \sqrt{1+4\sqrt{u}} }{2\sqrt{u}} \quad \text{and} \quad s_{2}(u)=\frac{1- \sqrt{1-4\sqrt{u}} }{2\sqrt{u}}.\]
Substituting into \cref{eq:kernel_equation_chain} yields the following linear
system:
\begin{equation}
	\left\{\begin{aligned}
		\frac{s_1(u)}{(1-s_1(u))^2}F(1,u)- \frac{1}{1-s_1(u)}\partial_s F(1,u) &=1\\
		\frac{s_2(u)}{(1-s_2(u))^2}F(1,u)- \frac{1}{1-s_2(u)}\partial_s F(1,u) &=1
	\end{aligned}\right..
\end{equation}
Solving this linear system gives us the closed form
\[F(1,u)=s_1(u)+s_2(u)-s_1(u)s_2(u)-1.\]
The rest of the proof is computational.
Since $\sqrt{1-4x}=-\sum_{j\geq 0}\frac{1}{2j-1}\binom{2j}{j}x^j$, we have the following expansions:
\begin{align*}
	\sqrt{1+4x}-\sqrt{1-4x}&=\sum_{j\geq 0}\frac{(-1)^{j+1}+1}{2j-1}\binom{2j}{j}x^j=2x\sum_{k\geq 0}\frac{1}{4k+1}\binom{4k+2}{2k+1}x^{2k}\,,\\
	\sqrt{1+4x}+\sqrt{1-4x}&=\sum_{j\geq 0}\frac{(-1)^{j+1}-1}{2j-1}\binom{2j}{j}x^j=-2\sum_{k\geq 0}\frac{1}{4k-1}\binom{4k}{2k}x^{2k}\,,\\
	\sqrt{1-16u}&=-\sum_{k\geq 0}\frac{1}{2k-1}\binom{2k}{k}4^ku^k\,.
\end{align*}
Hence
\begin{align*}
s_1(u)+s_2(u)&=\frac{\sqrt{1+4\sqrt{u}}-\sqrt{1-4\sqrt{u}}}{2\sqrt u}=\sum_{k\geq 0}\frac{1}{4k+1}\binom{4k+2}{2k+1}u^{k}\,,\\
4u s_1(u)s_2(u)&=-1+\sqrt{1+4\sqrt{u}}+\sqrt{1-4\sqrt{u}}-\sqrt{1-16{u}}\,,\\
&=-1-{2}\sum_{k\geq 0}\frac{1}{4k-1}\binom{4k}{2k}u^{k}+\sum_{k\geq 0}\frac{1}{2k-1}\binom{2k}{k}4^{k}u^{k}.
  \end{align*}
  The constant terms in the last equation cancel out, so we can divide by $4u$ and we obtain
\[F(1,u)=\sum_{k\geq 0}\frac{1}{4k+1}\binom{4k+2}{2k+1}u^{k}+\frac1{2}\sum_{k\geq 1}\frac{1}{4k-1}\binom{4k}{2k}u^{k-1}-\sum_{k\geq 1}\frac{1}{2k-1}\binom{2k}{k}4^{k-1}u^{k-1}-1.\]
Grouping the terms with same degree yields $P_k(1)=\frac{1}{4k+1}\binom{4k+2}{2k+1}+\frac1{2}\frac{1}{4k+3}\binom{4k+4}{2k+2}-\binom{2k+2}{k+1}\frac{4^k}{2k+1}$ for $k\geq 1$ (recalling that $P_k(1)$ is the coefficient of $u^k$ in $F(1,u)$, which is also the number of triangulations of $\chi^{(\tau^{[k]})}$). As $\binom{4k+4}{2k+2}/\binom{4k+2}{2k+1}=\frac{2(4k+3)}{2k+2}$, this simplifies to \[P_k(1)=\Big(\frac1{4k+1}+\frac1{2k+2}\Big)\binom{4k+2}{2k+1}-\binom{2k+2}{k+1}\frac{4^k}{2k+1},\] which gives the announced closed formula.
Finally, using Catalan numbers $C_k=\frac{1}{k+1}\binom{2k}{k}$, and their asymptotics $C_k=\frac{4^k}{k^{3/2}\sqrt{\pi}}(1+o(1))$, we can easily compute:
\[P_k(1)=\frac{3(2k+1)}{4k+1}C_{2k+1}-C_{k+1}\frac{4^k(k+2)}{2k+1}\sim_{k\to \infty}\frac{3-2\sqrt 2}{\sqrt{2\pi}}\frac{16^k}{k^{3/2}}\,.\qedhere\]
\end{proof}

\section{General algorithm for counting triangulations, and its complexity}\label{sec:general_triangulations}

Sections~\ref{sec:triangulations}
and~\ref{sec:counting_triangulations} presented a method to count the
triangulations of {chirotope paths}. We now generalize this method to
arbitrary chirotope trees.

\subsection{Chirotope with several marked elements and their triangulation polynomials}

Let $\xi$ be a chirotope with ground set $Y\cup \{y_1^*,\ldots,y_k^*\}$,
where $y_1^*$, \dots, $y_k^*$ are marked {\em extreme} elements of $\xi$ which do not belong to $Y$.
We have seen in \cref{ssec:refined_count_triangulations_bowtie}
that, with two marked elements ($k=2$), it is useful to introduce two 
bivariate polynomials $Q^\in$ and $Q^{\notin}$, separating the triangulations joining or not joining the marked elements.
For a general $k$, we shall need one $k$-variate polynomial for each possible restriction
of a triangulation to $\{y_1^*,\dots,y_k^*\}$.
Such a restriction is in general a collection of noncrossing edges with extremities in $\{y_1^*,\ldots,y_k^*\}$,
not necessarily maximal.
We will denote the set of such configurations by $\U_{\{y_1^*,\dots,y_k^*\}}$, or simply $\U_k$ when there is no ambiguity.
\medskip

Since $y_1^*$, \dots, $y_k^*$  are {\em extreme} elements, the set $\{y_1^*,\ldots,y_k^*\}$
is in convex position, and the asymptotic behaviour of $|\U_k|$ is known.
\begin{lemma}
  For any $k$ and $y_1^*$, \dots, $y_k^*$ in convex position, we have $|\U_k| = \O(\A^k)$, where $\A=6+4 \sqrt 2$.
  \label{lem:control_U}
\end{lemma}
\begin{proof}
  Let $U$ in $\U_k$ be a set of noncrossing edges with extremities in $\{y_1^*,\ldots,y_k^*\}$.
  We first note that each edge of the convex hull can be included in $U$ or not, independently of all other edges.
  We therefore get that $|\U_k|=2^k |\D_k|$, where $\D_k$ is the set of 
  collections of noncrossing edges with extremities in $\{y_1^*,\ldots,y_k^*\}$,
  and containing all edges of the convex hull.
  The elements of $\D_k$ are called {\em dissections} of the $k$-gon 
  with vertices $\{y_1^*,\ldots,y_k^*\}$.
  In particular, it is proved in \cite[Theorem 4]{flajolet1999noncrossing}
  that $|\D_k| = \O( (3+2\sqrt 2)^k )$. The statement follows.
\end{proof}

The collection of triangulation polynomials associated with a
chirotope $\xi$ having $k$ marked elements $Y^*=\{y_1^*,\ldots,
y_k^*\}$ is defined as follows. For a set $U \in \mathcal U_k$ of
noncrossing edges, we set
\[Q^U_{\xi}(y_1,\ldots,y_k)=
\sum_{T\in \mathcal{T}_{\xi} \atop T_{Y^* Y^*}=U} \prod_{i=1}^k
y_i^{\deg_{T}(y_i^*)}, \] where we recall from
\cref{sec:counting_triangulations} that $T_{Y^* Y^*}$ is the subset of
edges of $T$ with both extremities in $Y^*$. (To simplify notation, we
do not include the marked elements $\{y_1^*,\ldots, y_k^*\}$, in the
index of $Q$.) Note that the degree of each $Q^U_\xi$ in a variable
$y_i$ is less than the number of elements of $\xi$, that is $|Y|+k$.
If $U$ is not in $\mathcal U_k$, we use the convention $Q^U_\xi=0$.

\begin{remark}\label{rem:PandQ}
We already encountered specializations of these polynomials $Q^U_\xi$ in cases with few marked elements:
\begin{itemize}
 \item when there is no marked element, then $Q^U_\xi$ is just a number, and is equal to $|\mathcal{T}_\xi|$ (here, $U$ needs to be empty); 
 \item when there is one marked element $y*$, then $Q^U_\xi(y) = P_{\xi,y^*}(y)$ (here again, $U$ needs to be empty);
 \item and when there are two marked elements $y_1^*$ and $y_2^*$, 
 then $Q^U_\xi(y_1,y_2) = Q^\in_{\xi,y_1^*,y_2^*}(y_1,y_2)$ 
 (resp. $Q^{\notin}_{\xi,y_1^*,y_2^*}(y_1,y_2)$) if $U$ is the singleton containing the edge $\{y_1^*,y_2^*\}$ (resp. if $U$ is empty). 
\end{itemize}
\end{remark}

\medskip

The next step is to explain how to compute recursively $Q^U_\xi$
through a bowtie operation, just like \cref{prop:triangulation_chaine}
does for chirotopes with only two marked elements. To this end, for any set
$\{x_1,\ldots, x_r\}$ of variables {and integers $a,b\geq 2$}
satisfying $b \ge r$, we define
\[
	R_{a,b}(\{x_1,\ldots, x_r\})=
    \sum_{i_1,\ldots, i_r\geq 1}\binom{a+b-(i_1+\ldots +i_r)-2}{b-(r+1)}x_1^{i_1-1}\cdots x_r^{i_r-1}.
\]
Note that these polynomials are symmetric, justifying that we write our list of variables as a set.
We recall the convention $\binom{n}{-1}=\delta_{-1}$, leading to the following simpler formula if $b=r$:
\[
R_{a,r}(\{x_1,\ldots, x_r\})= \sum_{i_1,\ldots, i_r\geq 1 \atop
  i_1+\cdots +i_r= a+r-1}x_1^{i_1-1}\cdots x_r^{i_r-1}.\] Other
particular cases are given by $R_{a,b}(\emptyset)=\binom{a+b-2}{b-1}$
and $R_{a,b}(\{u\})$ being the polynomial $R_{a,b}(u)$ from
\cref{prop:triangulation_chaine}.  Just like $R_{a,b}(\emptyset)$ and
$R_{a,b}(\{u\})$ occurred in \cref{prop:triangulation_chaine} to count
triangulations with zero and one edge between two marked elements,
respectively, the family $R_{a,b}(\cdot)$ of multivariate polynomials
allows us to handle the general case with several marked elements.
\medskip

Now, consider a second chirotope $\chi$ on $X \cup \{x^*\}$, assume that $x^*$ is extreme,
and build the bowtie product 
$\kappa=\chi \join{x^*}{y_k^*}\xi$
on ground set $X\cup Y \cup \{y_1^*,\ldots, y_{k-1}^*\}$. We want to compute 
the associated collection of triangulation polynomials
\[Q^U_{\kappa}(y_1,\ldots,y_{k-1})=\sum_{T\in
  \mathcal{T}_{\kappa} \atop T_{\{y_1^*,\ldots, y_{k-1}^*\} \{y_1^*,\ldots, y_{k-1}^*\}} =U }\prod_{i=1}^{k-1}y_i^{\deg_{T}(y_i^*)}, \]
for $U \in \U_{k-1}$. The bijection of
  \cref{p:bij} between the triangulations of $\kappa$ and those of
$\chi$ and $\xi$ leads to the following formula, which generalizes
Proposition \ref{prop:triangulation_chaine}.

\begin{proposition}\label{p:fullpol_recurrence}
  Let $\chi$, $\xi$ and $\kappa=\chi \join{x^*}{y_k^*}\xi$ be as
  above.  For any $U$ in $\U_{k-1}$, we have
 \begin{equation}\label{eq:generalcase_arityk}
   \hspace{-1mm}
 Q^U_{\kappa}(y_1,\ldots,y_{k-1})=
     \sum_{S\subseteq [k-1] 
     \atop a,b \geq 2} \Big( [s^a]P_{\chi,x^*}(s) \Big) \cdot 
     R_{a,b}(\{y_i, i \in S\} \cdot \Big( [y_{k}^b]Q^{U \cup Y^*(S,k)}_{\xi} (y_1,\cdots,y_k)
     \Big), \end{equation}
     where $Y^*(S,k)=\{ \{y^*_i,y^*_k\}, i \in S\}$.
\end{proposition}
\begin{proof}
  We write $Y^*_{k-1}= \{y_1^*,\ldots, y_{k-1}^*\}$, $Y^*_{k}= \{y_1^*,\ldots, y_{k}^*\}$,
  and $\tilde Y=Y \cup Y^*_{k-1}$, for short.
  For $T\in \T_\kappa$ a triangulation of $\kappa$, we let
  $T'=\pi_{\tilde Y \to x^*}(T)$ and $T'' = \pi_{X \to y_k^*}(T)$, as usual.
  We also recall that $T_{X\tilde Y}$ denotes the set of edges in $T$ joining
  $X$ to $\tilde Y$. We recall from Proposition~\ref{p:bij}
  that the map $T \mapsto (T',T'',T_{X\tilde Y})$ is a bijection.

  Fix $U$ in $\U_{k-1}$. We note that $T$ and $T''$ coincide on $Y_{k-1}^*$. 
  Thus, the restriction of $T$ to $Y_{k-1}^*$ is $U$ if and only if that of $T''$ is $U$ as well.
  This implies that the restriction of $T''$ to $Y_{k}^*$
  is of the form $U \cup Y^*(S,k)$ for some $S\subseteq [k-1]$.
  To shorten notation, for a given $S$, we let $\T^{U,S}_\xi$ be the set of triangulations $T''$ of $\xi$
  whose restriction to $Y_{k}^*$ is equal to $U \cup Y^*(S,k)$.

  With the same argument as in \cref{prop:triangulation_chaine},
  we see that the degree of $y_i^*$ in $T$ are computed as follows:
  \begin{itemize}
    \item if $i$ is not in $S$, i.e.~if $y_i^*$ is not connected to $y_k^*$ in $T''$,
      then 
      \[ \deg_{T}(y_i^*)= \deg_{T''}(y_i^*).\]
    \item if $i$ is in $S$, i.e.~if $y_i^*$ is connected to $y_k^*$ in $T''$, then
      \[ \deg_{T}(y_i^*)= \deg_{T''}(y_i^*)+ \deg_{T_{X\tilde Y}}(y_i^*)-1.\]
  \end{itemize}
  Hence we have:
  \begin{equation}\label{eq:expansion_QU}
Q^U_{\kappa}(y_1,\ldots,y_{k-1})%
=\sum_{S \subseteq [k-1]} \, \sum_{T'\in \T_\chi} \, \sum_{T''\in \T^{U,S}_\xi} \, 
\left(\prod_{i=1}^{k-1}y_i^{\deg_{T''}(y_i^*)} \right)\cdot 
\left(\sum_{T_{X\tilde Y}} \prod_{i\in S}y_i^{\deg_{T_{X\tilde Y}}(y_i^*)-1}\right).
\end{equation}
In the last sum, $T_{X\tilde Y}$ runs over maximal sets of noncrossing edges between
$N_{T'}(x^*)$ and $N_{T''}(y_k^*)$ (see the statement of \cref{p:bij}). In particular,
even if it is not explicit in the notation, this last sum depends on $T'$ and $T''$.
Using \cref{lem:concave_k}, we see that
\begin{multline*}
  \sum_{T_{X\tilde Y}} \prod_{i\in S}y_i^{\deg_{T_{X\tilde Y}}(y_i^*)-1}
  = \sum_{ (d_i)_{i \in S} \atop d_i \ge 1}  \left(\prod_{i\in S}y_i^{d_i-1}\right)
\binom{ |N_{T'}(x^*)|+|N_{T''}(y_k^*)|-\big(\sum_{i \in S} d_i\big)-2}{|N_{T''}(y_k^*)|-(|S|+1)}
\\ = R_{|N_{T'}(x^*)|,|N_{T''}(y_k^*)|}(\{y_i,i \in S\}).
\end{multline*}
We split the sums in \cref{eq:expansion_QU}, depending on the values
of the parameters $a\eqdef|N_{T'}(x^*)|=\deg_{T'}(x^*)$ and $b\eqdef|N_{T''}(y_k^*)|=\deg_{T''}(y_k^*)$.
For given $a,b \ge 2$, there are $[s^a]P_{\chi,x^*}(s)$ triangulations $T'$ with $\deg_{T'}(x^*)=a$
and we have
\[\sum_{T''\in \T^{U,S}_\xi \atop \deg_{T''}(y_k^*)=b} \left(\prod_{i=1}^{k-1}y_i^{\deg_{T''}(y_i^*)} \right)
= [y_k^b] Q^{U \cup Y^*(S,k)}_{\xi}(y_1,\dots,y_k). \]
This implies
\[Q^U_{\kappa}(y_1,\ldots,y_{k-1})
= \sum_{S \subseteq [k-1]} \sum_{a,b \ge 2} \big([s^a]P_{\chi,x^*}(s)\big) \big([y_k^b] Q^{U \cup Y^*(S,k)}_{\xi}(y_1,\dots,y_k) \big)
\, R_{a,b}(\{y_i,i \in S\}),\]
as announced.\end{proof}

\subsection{A general algorithm}\label{sec:general_algo}

We now explain how to use \cref{p:fullpol_recurrence} to count the
triangulations of a chirotope presented by a chirotope tree $T$.  We
proceed as follows.
\begin{enumerate}
\item For each node $v$, we denote by $\chi_v$ the chirotope
  decorating the node $v$ and by $Y_v$ its set of proxy elements, and
  we compute the associated collection of triangulation polynomials
  $\{Q^U_{\chi_v} \colon U \in \mathcal U_{Y_v}\}$. These polynomials
  count triangulations of $\chi_v$, keeping track of the degree of
  each elements of $Y_v$ in the triangulation, and of the restriction
  of the triangulation to $Y_v$.  These polynomials can be computed in
  parallel for all nodes~$v$, using for each one a modified version of
  the Alvarez-Seidel algorithm \cite{alvarez2013simple}.

\item As long as the chirotope tree is not reduced to a single node,
  we pick an arbitrary leaf and perform an edge contraction (as
  defined in \cref{ssec:operation-trees}) to merge it it with its
  parent. We compute the collection of triangulation polynomials to be
  associated to the merged node using \cref{p:fullpol_recurrence}: we
  take $\chi$ to be the merged leaf, and use a single marked element
  $x^*$. In that case, as noted by \cref{rem:PandQ}, $P_{\xi,x^*}(u)
  =Q^U_\xi(u)$, with $U$ necessarily empty. For further details, see
  the proof of \cref{p:complexity_triangulations_restate} below.
\end{enumerate}
Each application of Step (2) yields a strictly smaller chirotope tree
that, by \cref{lem:merge_sign_function}, represents the same
chirotope, and for which we know the collection of triangulation
polynomials of every node.  Hence Step (2) can indeed be repeated
until the tree is reduced to a single node.  When this is the case,
there are no more marked elements, so the associated collection of
triangulation polynomials is actually a single integer, equal to the
number of triangulations of $\chi_T$.

\medskip

The complexity of this algorithm is given in the following
proposition.  We recall from the introduction (see also
\cref{lem:control_U}) that $\A$ stands for the constant $6+4 \sqrt 2
\approx 11.66$.

\begin{proposition}\label{p:complexity_triangulations_restate}
  Let $T$ be a chirotope tree.
  \begin{itemize}
  \item[(i)] If a node $v$ of $T$ is decorated with a chirotope
    $\chi_v$ with $n_v$ elements, including a set $Y_v$ of $k$
    proxies, then $\{Q_{\chi_v}^U\colon U \in \mathcal U_{Y_v}\}$ can
    be computed from $\chi_v$ in time $O\pth{A^k 2^{n_v} n_v^{k+2}}$.\smallskip
  \item[(ii)] If $T$ has $m$ edges, each of its node has degree at most $k$, and the ground set of
  $\chi_T$ has $n$ elements, then the above algorithm computes the number of triangulations of
  $\chi_T$ from the polynomials $\{Q_{\chi_v}^U\colon
  v \text{ a node of } T, U \in \mathcal U_{Y_v}\}$ in time
  $O\pth{A^{k} n^{2k} m}$.
  \end{itemize}
\end{proposition}
\begin{proof} 
  Statement~(i) follows from a simple modification of the
  Alvarez-Seidel algorithm~\cite{alvarez2013simple} so as to compute
  the degrees of marked vertices and keep track of the presence of
  edges between them through the aggregation.

  \medskip
  
  To prove Statement~(ii), let us now analyze one iteration of Step
  (2). Consider a leaf $\ell$ of the chirotope tree $T$, and denote by
  $v$ its parent node. We let $\xi$ denote the chirotope decorating
  $v$ and denote by $Y\cup \{y_1^*,\ldots, y_k^*\}$ its ground set,
  with $k\geq 2$ (where, as usual, the starred elements indicate the
  proxies).  Similarly, we let $\chi$ denote the chirotope decorating
  the leaf $\ell$ and $X\cup \{x^*\}$ its ground set (a leaf has only
  one proxy element, denoted $x^*$ here).

  \medskip

  Let $\U_k = \U_{\{y_1^*,\dots,y_k^*\}}$ and $\U_{k-1} =
  \U_{\{y_1^*,\dots,y_{k-1}^*\}}$.  For each $U \in \U_{k-1}$, we
  initialize ${\tilde Q}^{U}_{\xi}=0$.  Now, for every $\bar U \in
  \U_k$, we denote $p(\bar U)$ the restriction of $\bar U$ to
  $\{y_1^*,\ldots,y_{k-1}^*\}$ and $S_{\bar U}=\{i \in [k-1], y_i^*
  y_k^* \in \bar U\}$.  Then we compute
  \begin{equation}\label{eq:one_term_in_Q}
  \sum_{a,b \ge 2} \Big([s^a]P_{\chi,x^*}(s) \Big) \cdot  R_{a,b}(\{y_i, i \in 
  S_{\bar U}\}) \cdot 
  \Big( [y_{k}^b]Q^{\bar U}_{\xi} (y_1,\cdots,y_k) \Big)
  \end{equation}
  and add it to ${\tilde Q}^{p(\bar U)}_{\xi}$. 
  By \cref{p:fullpol_recurrence}, when all $\bar U \in \U_k$ have been treated,
  we have ${\tilde Q}^{U}_{\xi}=Q^{U}_{\xi}$ for all $U \in \U_{k-1}$. This proves the correctness of our algorithm.
  
  Let us analyze the complexity of this second step.  We first
  consider, for given $\bar U$ and $a,b$, the multiplication in
  \cref{eq:one_term_in_Q}.  The first factor $[s^a]P_{\chi,x^*}(s)$ is
  a scalar, while the two other polynomials are $(k-1)$-variate and
  have degree less than $n$ in each variable (in general the degree of
  a polynomial $Q^{V}_{\alpha}$ in a variable $y_i$ is less than the
  number of elements, including marked ones, in $\alpha$).  Therefore
  the last two polynomials have each at most $n^{k-1}$ monomials, and
  the multiplication has complexity $\O(n^{2(k-1)})$ (in the Real-RAM
  model\footnote{Recall that, in the Real-RAM model, adding or
  multiplying any two scalars takes constant time.  For a more
  realistic complexity analysis, we note that the bit complexity of
  every number arising in the computation is $\O(n)$, as each such
  number is bounded from above by the number of triangulations, which
  is $\O(30^n)$.}, using a brute-force multiplication algorithm).
  Since we only need to consider the case $a,b \le n$ (further terms
  are equal to 0), and since there are $\O(A^k)$ sets $\bar U$ in
  $\U_k$ (see \cref{lem:control_U}), this step has complexity
  $\O(n^{2k} A^k)$.

  \medskip

  Altogether, one iteration of Step~(2), that is the merge of a leaf
  of a chirotope tree with its parent, and the update of the
  associated collection of polynomials, can be done in time
  $\O(n^{2k}A^k)$. Notice that the strict upper bound of $n$ on the
  partial degree of these polynomials remains true, as it follows from
  the fact that in any chirotope tree representing a chirotope of
    size $n$, every node has at most $n$ elements. As $T$ has $m$
  edges, Step~(2) needs to be iterated $m$ times until the tree is
  reduced to a single vertex. This completes the proof of
  Proposition~\ref{p:complexity_triangulations_restate}.
\end{proof}

\begin{figure}
\centering
\includegraphics[scale=0.6,trim=0 1.3cm 0 1cm,clip]{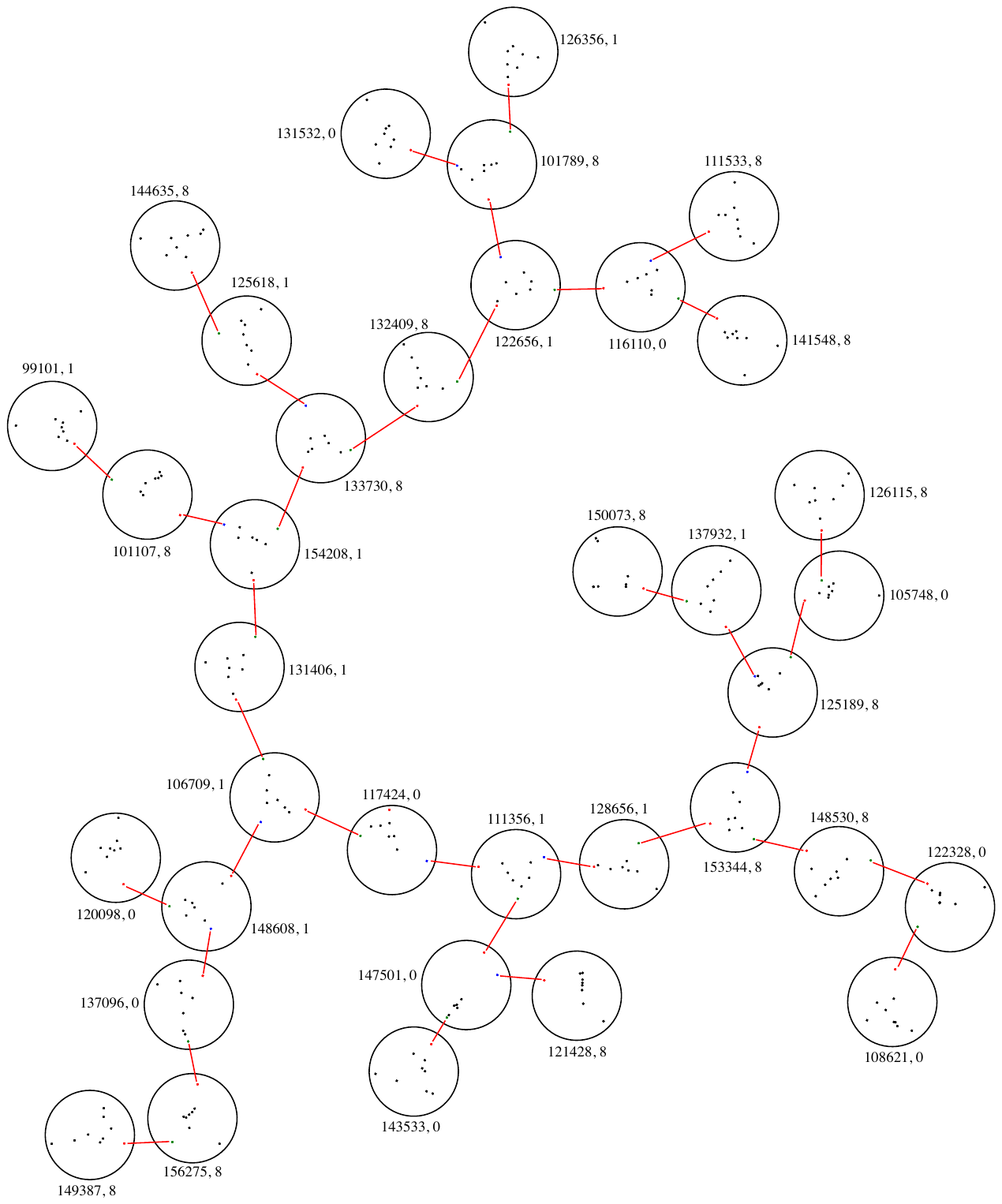}
\caption{A chirotope tree with 36 nodes, each decorated with a chirotope
  of size $9$, adding up to $254$ elements. The numbers next to each
  node indicate the identifier of the decorating chirotope in the
  order type database~\cite{AAK02}, as well as the label, in this database entry, of the proxy shown in red in the figure.
\label{f:exemple}}
\end{figure}

\subsection{A proof-of-concept implementation}
\label{a:arity3}

We implemented the method described in the previous section (both the
modification of the Alvarez-Seidel algorithm and the recursive
counting behind Proposition~\ref{p:complexity_triangulations}, restated above as \cref{p:complexity_triangulations_restate}) for the
special case of trees $T$ where
\begin{itemize}
\item  every node has degree at most~$3$,
\item every node of degree $3$ is decorated by a chirotope with only
  $3$ extreme elements, which are its proxies, and
\item every node of degree $2$ has its two proxies consecutive on the
  convex hull.
\end{itemize}
These conditions ensure that in any triangulation of a chirotope
decorating a node of $T$, there is always an edge between any two
proxies, i.e.~we only have one non-zero polynomial attached to each
node.  We propose this as a proof of concept: we made little effort in
optimizing or benchmarking it as either is beyond the scope of this
paper.

\subsubsection{Repository}

Our implementation runs within Sagemath 
to take advantage of basic
functions for the manipulation of multivariate polynomials, but is
mostly basic python. The code in Sagemath can be found on the
repository
\href{https://github.com/koechlin/chirotope_tree}{\texttt{chirotope\_tree}}. The
example of Figure \ref{f:exemple} can be obtained by running the
command \texttt{sage -python example.py}, which prints two numbers,
first the number of triangulations of the example, and the total time (in seconds) spent on computing it (ignoring the initial time to launch
sage).

\subsubsection{An example}

Our code computes in a few seconds on a standard laptop that the
number of triangulations of the chirotope of size $254$ presented by
the tree of Figure~\ref{f:exemple} is
\begin{align*}
  59&2\,966\,751\,293\,974\,711\,252\,579\,414\,478\,724\,131\,868\,483\,318\,559\,312\,640\,993\,804\,562\,350\,446\\
& 478\,462\,502\,194\,102\,338\,465\,347\,793\,563\,468\,964\,711\,887\,069\,260\,192\,677\,783\,5079\,385\,675\\
& 578\,362\,313\,461\,572\,573\,372\,584\,158\,103\,703\,847\,713\,232\,664 \approx 5.92966751.10^{180}.
\end{align*}

\bibliographystyle{bibli_perso}
\bibliography{biblio-OT}
\end{document}